\documentclass{article}
\usepackage{times}
\usepackage{paralist}
\usepackage{amsfonts}
\usepackage{amsmath}
\usepackage{amsthm}
\usepackage{latexsym}
\usepackage{url,hyperref}
\usepackage{boxedminipage}
\usepackage{color}
\usepackage{bm}
\usepackage{xspace}
\usepackage{multirow}
\usepackage[left=1in, right=1in, bottom=1.3in, top=1in]{geometry}

\newtheorem{theorem}{Theorem}[section]
\newtheorem{lemma}[theorem]{Lemma}
\newtheorem{proposition}[theorem]{Proposition}
\newtheorem{corollary}[theorem]{Corollary}
\newtheorem{remark}[theorem]{Remark}

\newtheorem{definition}[theorem]{Definition}

\newtheorem{algorithm}[theorem]{Algorithm}

\newcommand{\E}{\mbox{\bf E}}

\newcommand{\cI}{{\cal I}}
\newcommand{\cM}{{\cal M}}
\newcommand{\cF}{{\cal F}}
\newcommand{\cX}{{\cal X}}
\newcommand{\cY}{{\cal Y}}

\newcommand{\cZ}{{\cal Z}}

\newcommand{\RR}{{\mathbb R}}

\def\blam{{\bm{\lambda}}}

\def\b1{{\bf 1}}
\def\be{{\bf e}}

\def\bb{{\bf b}}

\def\bx{{\bf x}}
\def\by{{\bf y}}
\def\bv{{\bf v}}
\def\bw{{\bf w}}
\def\bz{{\bf z}}

\def\bp{{\bf p}}

\def\bbb{{b}}

\def\ourspan{\mbox{span}}

\newcommand{\partdiff}[2]{\frac{\partial {#1}}{\partial {#2}}}

\newcommand{\mixdiff}[3]{\frac{\partial^2 {#1}}{{\partial {#2}}{\partial {#3}}}}

\newcommand{\etal}{et al.\ }
\newcommand{\ceil}[1]{\lceil #1 \rceil}
\newcommand{\floor}[1]{\lfloor #1 \rfloor}
\newcommand{\eps}{\varepsilon}
\newcommand{\opt}{\text{\sc OPT}}
\newcommand{\poly}{\mathrm{poly}}

\newcommand{\mypara}[1]{\bigskip \noindent {\bf #1}}

\newcommand{\crs}{{\sc CR}\xspace}

\newcommand{\dmax}{d_{\max}}
\newcommand{\bmin}{b_{\min}}

\begin{document}

\title{Submodular Function Maximization via the Multilinear Relaxation and
Contention Resolution Schemes\footnote{A preliminary version of
this paper appeared in {\em Proc.\ of ACM STOC}, 2011.}%
}

\author{Chandra Chekuri\thanks{Dept. of Computer Science,
Univ.\ of Illinois, Urbana, IL 61801, USA. Partially supported
by NSF grants CCF-0728782 and CCF-1016684. E-mail: {\tt chekuri@illinois.edu}.}
\and
Jan Vondr\'ak\thanks{IBM Almaden Research Center, San Jose, CA 95120, USA.
E-mail: {\tt jvondrak@us.ibm.com}.}
\and
Rico Zenklusen\thanks{Dept.~of Mathematics, ETH Zurich, 8092 Zurich, Switzerland,
and Department of Applied Mathematics and Statistics, Johns Hopkins University,
Baltimore, MD 21218, USA.
E-mail: {\tt ricoz@math.ethz.ch}.
Supported by Swiss National Science Foundation grant
PBEZP2-129524, by NSF grants CCF-1115849 and CCF-0829878, 
and by ONR grants N00014-12-1-0033, N00014-11-1-0053
and N00014-09-1-0326.
}
}

\maketitle
\begin{abstract}
  We consider the problem of maximizing a non-negative submodular set
  function $f:2^N \rightarrow \RR_+$ over a ground set $N$ subject to
  a variety of packing type constraints including (multiple) matroid
  constraints, knapsack constraints, and their intersections. In this
  paper we develop a general framework that allows us to derive a
  number of new results, in particular when $f$ may be a {\em
    non-monotone} function. Our algorithms are based on
  (approximately) maximizing the multilinear extension $F$ of $f$
  \cite{CCPV07} over a polytope $P$ that represents the constraints,
  and then effectively rounding the fractional solution. Although this
  approach has been used quite successfully
  \cite{CCPV09,KulikST09,LeeMNS09,CVZ10,BansalKNS10}, it has been
  limited in some important ways. We overcome these limitations as
  follows.

  First, we give constant factor approximation algorithms to maximize
  $F$ over a down-closed polytope $P$ described by
  an efficient separation oracle.
  Previously this was known only for
  monotone functions \cite{Vondrak08}. For non-monotone functions, a
  constant factor was known only when the polytope was either the
  intersection of a fixed number of knapsack constraints
  \cite{LeeMNS09} or a matroid polytope \cite{Vondrak09,OV11}. Second,
  we show that {\em contention resolution schemes} are an effective
  way to round a fractional solution, even when $f$ is
  non-monotone. In particular, contention resolution schemes for
  different polytopes can be combined to handle the intersection
  of different constraints. Via LP duality we show that a
  contention resolution scheme for a constraint is related to
  the {\em correlation gap} \cite{ADSY10} of weighted rank functions
  of the constraint.  This leads to an optimal contention
  resolution scheme for the matroid polytope.

  Our results provide a broadly applicable framework for maximizing
  linear and submodular functions subject to independence
  constraints. We give several illustrative examples. Contention
  resolution schemes %
  may find other applications. 
\end{abstract}

\section{Introduction}
\label{sec:intro}

We consider the meta-problem of {\em maximizing} a non-negative submodular
set function subject to independence constraints. Formally, let $N$ be
a finite ground set of cardinality $n$, and let $f:2^N \rightarrow
\RR_+$ be a submodular set function over $N$.\footnote{A set function
  $f:2^N \rightarrow \RR$ is submodular iff $f(A) + f(B) \ge f(A\cup
  B) + f(A \cap B)$ for all $A,B \subseteq N$.}  Let $\cI \subseteq
2^N$ be a downward-closed family\footnote{A family of sets
  $\cI \subseteq 2^N$ is downward-closed if for any $A \subset B
  \subseteq N$, $B \in \cI$ implies that $A \in \cI$.} of subsets of
$N$. Our problem is then $\max_{S \in \cI} f(S)$. We are interested in
independence families induced by natural and useful constraints such
as matroid constraints, knapsack constraints, related special cases,
and their intersections. Throughout this paper we assume that $f$ is
given via a value oracle; that is, given a set $S \subseteq N$ the
oracle returns $f(S)$. The function $f$ could be monotone or
non-monotone\footnote{$f$ is {\em monotone} if $f(A) \le
f(B)$ whenever $A \subseteq B$.}; monotone functions typically
allow better approximation results.

Submodular function maximization has recently attracted considerable
attention in theoretical computer science. This is for a variety of
reasons, including diverse applications---a prominent application
field being algorithmic game theory, where submodular functions are
very commonly used as utility functions to describe diminishing
returns---and also the recognition of interesting algorithmic and
structural properties.  A number of well-known problems can be seen as
special cases of submodular function maximization. For example, the
APX-hard Max-Cut problem can be seen as (unconstrained) maximization
of the cut function $f:2^V \rightarrow \RR_+$ of a graph
$G=(V,E)$. (Note that $f$ here is non-monotone.)  Another well-known
special case of our problem is the Max-$k$-Cover problem, which can be
viewed as $\max \{f(S): |S| \leq k\}$ where $f(S) = |\bigcup_{j \in S}
A_j|$ is the coverage function for a collection of sets $\{A_i\}$.
Max-$k$-Cover is hard to approximate to within a factor of
$(1-1/e+\eps)$ for any fixed $\eps > 0$, unless $P=NP$ \cite{Feige98}.
Hence we focus on approximation algorithms\footnote{If $f$ is not
  assumed to be non-negative, even the unconstrained problem is
  inapproximable since deciding whether the optimum value is positive
  or zero requires an exponential number of queries.}.

Classical work in submodular function maximization was based on combinatorial
techniques such as the greedy algorithm and local search.  We mention
the work of Cornuejols, Fisher, Nemhauser and Wolsey
\cite{CornuejolsFN77,NWF78,FNW78,NW78} from the late 70's which showed a
variety of approximation bounds when $f$ is monotone submodular and
$\cI$ is the intersection of matroid constraints.
Recent algorithmic work has considerably extended and improved the
classical results. Local-search methods have been identified as
particularly useful, especially for non-monotone functions. Some of
the recent results include the first constant factor approximation for
the unconstrained submodular function maximization problem
\cite{FeigeMV07}, and a variety of approximation results for knapsack
and matroid constraints \cite{LeeMNS09,LeeSV09}. The greedy algorithm
has also been modified and made applicable to non-monotone functions
\cite{GuptaRST10}.

Despite the above-mentioned results, combinatorial techniques have
some limitations: (i) they have not been able to achieve
optimal approximation results, except in the basic case of a single
cardinality or knapsack constraint \cite{NWF78,Sviri04}; (ii) they 
do not provide the flexibility to combine constraints of different types.
A new approach which overcomes some of these obstacles and brings
submodular function maximization closer to the world of polyhedral
techniques is via the {\em multilinear relaxation}, introduced in this
context in \cite{CCPV07}.

\medskip
\noindent{\bf A relaxation and rounding framework based on
the multilinear relaxation.}
In this paper we introduce a general relaxation and rounding framework
for maximizing submodular functions, which builds upon, and
significantly extends, previous approaches.  When dealing
with linear (or convex) objective functions, a standard paradigm is to
design a linear or convex relaxation whose solution is then rounded
via a problem-specific procedure.
A difficulty faced in extending this 
approach to maximizing a submodular function
$f:2^N \rightarrow \mathbb{R}_+$ --- which we often interpret 
as a function on the vertices of
a $\{0,1\}^N$ hypercube that correspond to incidence vectors --- is to
find a suitable extension $g:[0,1]^N \rightarrow \mathbb{R}_+$ of $f$
to the full hypercube.
The goal is to leverage such an extension $g$ as follows.
Suppose we have a polytope $P_{\cI} \subseteq [0,1]^N$ that is a
relaxation for $\cI \subseteq 2^N$ in the sense that $\{ \b1_I \mid I
\in \cI\} \subset P_\cI$.
We want to approximately maximize the continuous problem
$\max_{\bx \in P_{\cI}} g(\bx)$ to find a fractional solution
$\bx^* \in P_\cI$ that is finally rounded to a feasible integral
solution.

The best-studied extension of a submodular function is the Lov\'{a}sz
extension~\cite{Lovasz}; however, being a convex function, it is
mostly suitable for submodular function minimization problems.
For maximization of submodular functions, the following {\em multilinear}
extension was introduced in \cite{CCPV07}, inspired by the work in
\cite{AgeevS04}:  
$$ F(\bx) = \sum_{S \subseteq N} f(S) \prod_{i \in S}
x_i \prod_{j \not \in S} (1-x_j).$$ 
The value $F(\bx)$ is equivalently the expected value of $f(R)$ where
$R$ is a random set obtained by picking each element $i$ independently
with probability $x_i$. We observe that if $f$ is modular\footnote{A
  function is modular if $f(A) + f(B) = f(A\cup B) + f(A \cap B)$ for
  all $A, B \subseteq N$. If $f$ is modular then $f(A) = w_0 + \sum_{i
    \in A} w_i$ for some weight function $w:N \rightarrow \RR$.} then
$F$ is simply a linear function.
In this paper we focus on the multilinear extension.
The two obvious questions that arise when trying to build a
general relaxation and rounding framework based on the multilinear
extension are the following.
First, can we (approximately) solve the
problem $\max_{\bx \in P_{\cI}} F(\bx)$?
This question is particularly interesting due to the fact
that the multilinear extension is in general not concave
(nor convex).
Second, can we round a fractional solution effectively?

Recent work has addressed the above questions in several ways. First,
Vondr\'ak \cite{Vondrak08} gave a continuous greedy algorithm that
gives an optimal $(1-1/e)$-approximation for the problem $\max_{\bx
  \in P} F(\bx)$ when $f$ is monotone submodular and $P$ is a solvable
polytope\footnote{We say that a polytope $P$ is solvable if one can do
  efficient linear optimization over $P$.}. When $f$ is non-monotone,
the picture is less satisfactory. Lee \etal \cite{LeeMNS09} gave a
local-search based algorithm that gives a $(1/4-\eps)$-approximation
to maximize $F$ over the polytope induced by a fixed number of
knapsack constraints. Vondr\'ak \cite{Vondrak09} obtained a
$0.309$-approximation for maximizing $F$ over a single matroid
polytope, and this ratio has been recently improved to $0.325$
\cite{OV11}. However, no approximation algorithm was known to maximize
$F$ over a general solvable polytope $P$.

In terms of rounding a fractional solution $\bx$, a natural strategy
to preserve the value of $F(\bx)$ in expectation is to independently
round each coordinate $i$ up to $1$ with probability $x_i$ and down to
$0$ otherwise.  However, this rounding strategy does not typically
preserve the constraints imposed by $\cI$. Various dependent rounding
schemes have been proposed.  It was shown in \cite{CCPV07} that
"pipage rounding" can be used to round solutions in the matroid
polytope without losing in terms of the objective function $F(\bx)$
(\cite{CVZ10} achieves the same via "swap-rounding").
In~\cite{KulikST09,LeeMNS09,BansalKNS10,KulikST10}, randomized
rounding coupled with alteration was used for knapsack
constraints. More recently, \cite{CVZ10} showed concentration
properties for rounding in a single matroid polytope when $f$ is
monotone, and \cite{Vondrak10} showed concentration for independent
rounding even when $f$ is non-monotone.  These led to a few additional
results. Despite this progress, the ``integrality gap'' of $\max
\{F(\bx): \bx \in P\}$ has been so far unknown even when $f$ is
monotone and $P$ the intersection of two matroid polytopes. (We remark
that for intersections of matroids, combinatorial algorithms are known
to yield good approximations \cite{LeeMNS09,LeeSV09}.)  However, even
for modular functions---i.e., classical linear
optimization---combining constraints such as matroids and knapsack
constraints has been difficult, and no general result was known that
matched the best bounds one can get for them separately.

In summary, previous results via the multilinear relaxation were known
only for rather restricted cases, both in terms of approximately
maximizing the multilinear extension, and in terms of effectively
rounding fractional solutions.
We next describe the contributions of this paper in this context.

\mypara{Our contribution at a high level:} In this paper we develop a
general framework for solving submodular maximization problems of the
form $\max \{f(S): S \in \cI \}$, where $f:2^N \rightarrow \RR_+$ is
submodular and $\cI \subset 2^X$ is a downward-closed family of
sets. Our framework consists of the following components.
\begin{itemize}
\item {\em Optimizing the multilinear relaxation:} We give the first
  constant factor approximation, with an additional negligible
  additive error, for the problem $\max \{F(\bx): \bx
  \in P\}$ where $F$ is the multilinear extension of any non-negative
  submodular function, and $P$ is any down-monotone\footnote{A polytope $P \subseteq [0,1]^N$ is down-monotone if for all $\bx,\by \in [0,1]^N$,
$\by \le \bx$ and $\bx \in P$ implies $\by \in P$.} solvable polytope.
\item {\em Dependent randomized rounding:} We propose a general
  (dependent) randomized rounding framework for modular and submodular
  functions under independence constraints via what we call {\em
    contention resolution schemes} (\emph{\crs schemes}).  Rounding an
  approximate maximizer of the relaxation $\max \{F(\bx): \bx \in P\}$
  via a \crs scheme that is tailored to the given constraints, leads
  to a solution with provable approximation guarantee.  A key
  advantage of \crs schemes is the ability to easily combine 
  \crs schemes designed for different constraints into a \crs scheme
  for the intersection of these constraints.
\item {\em Contention resolution schemes:} We present \crs schemes for
  a variety of packing constraints, including knapsack constraints,
  matroid constraints, sparse packing systems, column-restricted
  packing constraints, and constraints imposed by an unsplittable flow
  problem in paths and trees.  Our \crs scheme for the matroid
  polytope, which is provably optimal~\footnote{We would like to highlight that
  contention resolutions schemes are just one way of rounding a fractional solution.
  Hence, the use of
  an optimal contention resolution scheme does not imply that no better approximation
  factor can be obtained by a different procedure.}, is obtained by exploiting a
  tight connection between \crs schemes and the {\em correlation
    gap}~\cite{ADSY10} of the associated weighted rank functions.
  Previously, in the context of matroids, an optimal \crs scheme was
  only known for the uniform matroid of rank
  $1$~\cite{FeigeV06,FeigeV10}.
\end{itemize}
\medskip

The above ingredients can be put together to obtain a relaxation and rounding framework leading to a variety of new results that we discuss in more detail in Section~\ref{sec:apps}.
We summarize some of our results in Table~\ref{table:results}.

\begin{table*}[htb]\label{table:results}
\small
\def\arraystretch{1.2}
$$\begin{array}{|| c || c | c | c ||} \hline
\mbox{Constraint type} & \mbox{Linear maximization} & \mbox{Monotone submod. max.}
 & \mbox{Non-negative submod. max.}   \\
\hline \hline
O(1) \mbox{ knapsacks} & [1-\eps \mbox{ \cite{ChandraHW76,FriezeC84}}] & [1-1/e-\eps   \mbox{ \cite{KulikST10}}] & 0.325 \ \ [0.25 \mbox{ \cite{LeeMNS09}}]  \\
\hline
k \mbox{ matroids } \& \ \ell=O(1) \mbox{ knapsacks} & \frac{0.6}{k} \ \ [\Omega(\frac{1}{(k + \ell}) \mbox{ \cite{GuptaNR10,GuptaRST10}}] & \frac{0.38}{k} \ \ [\Omega(\frac{1}{(k + \ell}) \mbox{ \cite{GuptaNR10,GuptaRST10}}]
  & \frac{0.19}{k} \ \ [\Omega(\frac{1}{k + \ell})\mbox{ \cite{GuptaNR10,GuptaRST10}}] \\
\hline
k\mbox{-matchoid } \& \ \ell\mbox{-sparse PIP} & \Omega(\frac{1}{k+\ell})%
 &  \Omega(\frac{1}{k+\ell}) %
 & \Omega(\frac{1}{k+\ell}) %
  \\
\hline
\mbox{Unsplittable flow in paths and trees} & [\Omega(1) \mbox{ \cite{ChekuriMS03}}] & \Omega(1) & \Omega(1) \\
\hline
\end{array} $$
\caption{\small Approximation ratios for different types of constraints and objective functions. Results in square brackets were previously known.
  The ratios in the last column for non-monotone functions are based on a $0.325$ approximation for
  maximizng the multilinear relaxation described in the conference version
  of this paper \cite{CVZ11}; there is a $\frac{1}{e}-\eps \approx 0.367$-approximation 
  from subsequent work of \cite{FNS11} which results in improved bounds.
}
\end{table*}

\subsection{Maximizing the multilinear extension over a general polytope}

We now give a more detailed description of our technical results and
the general framework.  First, we give a constant factor approximation
for the problem $\max \{F(\bx): \bx \in P\}$, where $F$ is the
multilinear extension of a non-monotone submodular function $f$ and
$P$ is a down-monotone solvable polytope; the monotone case admits a
$(1-1/e)$-approximation \cite{Vondrak08} as we mentioned already.  The
condition of down-monotonicity of the polytope is necessary for the
non-monotone case; it follows from \cite{Vondrak09,Vondrak09a} that no constant
factor approximation is possible for the matroid base polytope which
is not down-monotone. 

The main algorithmic technique for non-monotone functions has been local
search. Fractional local search with additional ideas has been the
tool to solve the continuous problem in special cases of polytopes
\cite{LeeMNS09,Vondrak09,OV11}.  Previous fractional local search
methods (\cite{LeeMNS09} and~\cite{Vondrak09}) improved a current solution $\bx$ by considering
moves along a small number of coordinates of $\bx$. The analysis took
advantage of the combinatorial structure of the underlying constraint
(knapsacks or matroids) which was sufficiently simple that
swaps along a few coordinates sufficed. How do we obtain an algorithm
that works for {\em any} polytope $P$?

\smallskip
\noindent {\bf A new insight:} 
Our key high-level idea is simple yet insightful. Any point $\bx \in
P$ can be written as a convex combination of the vertices of $P$. We
view the problem of $\max \{F(\bx): \bx \in P\}$ as optimizing a
submodular function over the ground set consisting of the
(exponentially many) vertices of $P$ (duplicated many times in the
limit). 
From this viewpoint we obtain a new
fractional local search procedure: given a current point $\bx$, a
local swap corresponds to removing a vertex in the convex combination
of $\bx$ and adding a new vertex of $P$ (with appropriate scalar
multipliers). To implement this efficiently we can use linear
optimization over $P$.  (We remark that the continuous greedy
algorithm for the monotone case \cite{Vondrak08} can also be
interpreted with this insight.)

Our algorithms are derived using the above high-level idea.
We note that when specialized to the
matroid polytope or knapsack polytope which have combinatorial
structure, our algorithms become simpler and in fact resemble previous
algorithms.
Our algorithms and proofs of approximation guarantees are in
fact simpler than the previously given proofs for particular polytopes
\cite{LeeMNS09,Vondrak09,OV11}.

We present two algorithms following this idea. The first algorithm
is close in spirit to the local-search algorithm of 
Lee et al.~for knapsack constraints~\cite{LeeMNS09} and gives
a $0.25$-approximation. 
This algorithm, despite having a worse approximation guarantee then
the second one we present, allows us to further explain and formalize
the above high-level idea in a clean way.
The second algorithm uses some ideas of
\cite{Vondrak09} for the case of a matroid polytope and gives a
$0.309$-approximation with respect to the best {\em integer} solution
in $P$.

We would like to mention that subsequently to the conference version
of this paper, Feldman et al.~\cite{FNS11} presented an improved
algorithm to maximize the multilinear extension, leading to an
$(e^{-1}-\epsilon)\approx 0.367$-approximation with respect to the
best integer solution. Their algorithm is an adaptation of the
continuous greedy algorithm \cite{Vondrak08}.  The conference version
of this paper~\cite{CVZ11} contained a third and much more involved
algorithm (generalizing the simulated annealing approach
of~\cite{OV11}) that gives a $0.325$-approximation, again with respect
to the best integer solution in $P$.  For conciseness, and in view of
the recent results in~\cite{FNS11}, we do not include this third
algorithm in this paper, and concentrate on the first two algorithms
mentioned in the preceding paragraph, which allow us to demonstrate
the main new algorithmic insights.  We summarize our results in the
following theorem.

\begin{theorem}
\label{thm:multilinear-solve}
Let $f$ be a nonnegative submodular function and $P\subseteq \mathbb{R}^n$
be a solvable down-monotone polytope satisfying that there is a
$\lambda\in \Omega(\frac{1}{\poly(n)})$ such that
for each coordinate $i\in [n]$,
$\lambda \cdot \be_i \in P$.
then there is a $(0.25-o(1))$-approximation algorithm
for the problem $\max \{ F(\bx): \bx \in P \}$ where $F$ is the
multilinear extension of $f$.  There is also an algorithm for this
problem which returns a solution $\by \in P$ of value $F(\by) \geq
0.309 \cdot \max \{F(\bx): \bx \in P \cap \{0,1\}^N \}$.
\end{theorem}

We remark that a known limit on the approximability of $\max \{F(\bx): \bx \in P\}$ is an
information theoretic hardness of $0.478$-approximation in the
value oracle model, even in the special case of a matroid polytope \cite{OV11}.

\subsection{Contention resolution schemes} 
We show that a certain natural class of rounding schemes that we call
{\em contention resolution schemes} (\emph{\crs schemes}) provides a
useful and general framework for rounding fractional solutions under
submodular objective functions.  For a ground set $N$, let $P_{\cI}$
be a convex relaxation of the constraints imposed by $\cI \subseteq
2^N$, and let $\bx \in P_{\cI}$. From the definition of $F$, a natural
strategy to round a point $\bx$ is to independently round the
coordinates; however, this is unlikely to preserve the constraints
imposed by $\cI$. Let $R(\bx) \subseteq N$ be a random set obtained by
including each element $i \in N$ independently with probability $x_i$.
The set $R(\bx)$ is not necessarily feasible.  We would like to remove
(randomly) some elements from $R(\bx)$, so that we obtain a feasible
set $I \subseteq R(\bx)$.  The property we would like to achieve is
that every element $i$ appears in $I$ with probability at least $cx_i$
for some parameter $c>0$. We call such a scheme ``$c$-balanced
contention resolution'' for $P_{\cI}$. We stress that a $c$-balanced
\crs scheme needs to work for all $\bx \in P_{\cI}$.
However, often, stronger schemes---i.e.~with larger values for $c$---can
be obtained if they only need
to work for all points in a scaled-down version
$b P_\cI = \{b\cdot \bx \mid \bx\in P_\cI \}$ of $P_\cI$, where
$b\in [0,1]$. Such schemes, which we call $(b,c)$-balanced schemes,
will prove to be useful when combining \crs schemes for different
constraints as we will discuss in Section~\ref{sec:framework}.
Below is a formal definition of \crs schemes. Let $\text{support}(\bx) = \{i \in N \mid x_i > 0\}$.

\begin{definition}
  \label{defn:crscheme}
  Let $b,c \in [0,1]$. A $(b,c)$-balanced \crs scheme $\pi$
  for $P_{\cI}$ is a procedure that for every $\bx \in b P_{\cI}$
  and $A\subseteq  N$, returns a \emph{random} set
  $\pi_\bx(A)\subseteq A\cap \text{support}(\bx)$ and satisfies the following properties: 
\begin{enumerate}[(i)]
\item  $\pi_\bx(A)\in \mathcal{I}$ with probability $1$
\;\;$\forall A\subseteq N, \bx\in b P_{\cI}$, and
\item\label{item:crs} for all $i \in \text{support}(\bx)$,
  $\Pr[i\in \pi_\bx(R(\bx)) \mid i \in R(\bx)] \geq c$
  \;\;$\forall \bx \in b P_{\cI}$. 
\end{enumerate}
The scheme is said to be {\em monotone} if $\Pr[i \in
\pi_\bx(A_1)]\geq \Pr[i\in \pi_\bx(A_2)]$ whenever $i \in A_1
\subseteq A_2$.  A $(1,c)$-balanced \crs scheme is also called a
\emph{$c$-balanced \crs scheme}. The scheme is \emph{deterministic} if
$\pi$ is a deterministic algorithm (hence $\pi_\bx(A)$ is a single set
instead of a distribution). It is \emph{oblivious} if $\pi$ is
deterministic and $\pi_\bx(A) = \pi_\by(A)$ for all $\bx,\by$ and $A$,
that is, the output is independent of $\bx$ and only depends on $A$.
The scheme is \emph{efficiently implementable} if $\pi$ is
a polynomial-time algorithm that given $\bx,A$ outputs $\pi_\bx(A)$.
\end{definition}

We emphasize that a \crs scheme is defined with respect to a specific
polyhedral relaxation $P_{\cI}$ of $\cI$.
Note that on the left-hand side of condition~\eqref{item:crs} for a
\crs scheme, the probability is with respect to two random sources:
first the set $R(\bx)$ is a random set, and second, the procedure
$\pi_\bx$ is typically randomized.
We note that a $(b,c)$-balanced \crs scheme $\pi$ can
easily be transformed into a $b c$-balanced \crs
scheme; details are given in Section~\ref{sec:CRS}.

The theorem below highlights the utility of of \crs schemes; when
rounding via monotone contention resolution schemes, one can claim an
expectation bound for submodular functions.  A similar theorem was shown
in \cite{BansalKNS10} for monotone functions. We state and prove ours
in a form suitable for our context.

\begin{theorem}
\label{thm:submod-balanced}
Let $b,c\in [0,1]$, and let $f:2^N \rightarrow \RR_+$ be a
non-negative submodular function
with multilinear relaxation $F$, and $\bx \in b\cdot P_{\cI}$,
where $P_{\cI}$ is a convex relaxation for $\cI \subseteq 2^N$.
Furthermore, let $\pi$ be a monotone $(b,c)$-balanced \crs scheme
for $P_{\cI}$, and let $I=\pi_{\bx}(R(\bx))$.
If $f$ is monotone then
$$\E[f(I)] \geq c \, F(\bx).$$
Furthermore, there is a function $\eta_f:2^N\rightarrow 2^N$ that depends
on $f$ and can be evaluated in linear time, such that even for
$f$ non-monotone
$$\E[f(\eta_f(I))] \geq c \, F(\bx).$$
\end{theorem}

As we will see in Section~\ref{sec:CRS}, the function $\eta_f$ can be
chosen to always return a subset of its argument. We therefore call
it a \emph{pruning operation}.

We observe that several previous rounding procedures for packing (and
also covering) problems rely on the well-known technique of {\em
  alteration} of a set obtained via independent rounding and are
examples of \crs schemes (see
\cite{Srinivasan01,CCKR01,CCGK07,ChekuriMS03,BansalKNS10}).  However,
these schemes are typically oblivious in that they do not depend on
$\bx$ itself (other than in picking the random set $R$), and the
alteration is deterministic. Our definition is inspired by the ``fair
contention resolution scheme'' in \cite{FeigeV06,FeigeV10} which
considered the special case of contention for a single item. The
dependence on $\bx$ as well as randomization is necessary (even in
this case) if we want to obtain an optimal scheme.  One key question
to consider is whether some given down-monotone polytope
$P_\mathcal{I}$ admits a ``good'' $(b,c)$-balanced CR scheme, which
corresponds to having values of $b$ and $c$ that are as close
to $1$ as possible. 

One natural way to apply a $(b,c)$-balanced CR scheme to a point
$\hat{x}\in P_{\cI}$ that approximately maximizes $F$ is as follows.
In a first step we scale down $\hat{x}$ to obtain $x=b\cdot \hat{x}$.
By non-negativity and concavity of $F$ along non-negative directions
one obtains $F(x) \geq b\cdot F(\hat{x})$. Applying a $(b,c)$-balanced CR
scheme $\pi$ to $x$ leads to a set $I=\pi(R(x))$ which, according to
Theorem~\ref{thm:submod-balanced}, satisfies
$f(I)\geq c F(x) \geq c b F(\hat{x})$. This also highlights
the motivation why we want to have $b$ and $c$ as close
to $1$ as possible.

As we will show, many
natural constraint systems admit good $(b,c)$-balanced CR schemes,
including matroid constraints, knapsack constraints, and a variety of
packing integer programs.
In particular, to deal with the rather general class of matroid
constraints, we exploit a close connection between the existence of CR
schemes and a recently introduced concept, called \emph{correlation
  gap}~\cite{Yan11}.

\mypara{Contention resolution via correlation gap and an optimal scheme for matroids:}
Until recently there was no contention resolution scheme for the
matroid polytope; an optimal $(b, \frac{1-e^{-b}}{b})$-balanced scheme
was previously known for the very special case of the uniform matroid
of rank one~\cite{FeigeV06,FeigeV10}.  We note that the recent work of
Chawla \etal \cite{ChawlaHMS09,ChawlaHMS10} implicitly contains a
$(b,1-b)$-balanced deterministic scheme for matroids; their
motivation for considering this notion was mechanism design.  In this paper we
develop an optimal scheme for an arbitrary matroid\footnote{We also
describe the $(b,1-b)$ scheme in
 Section~\ref{subsec:CRmatroids} for completeness. This scheme
is simpler and computationally advantageous when compared to
the optimal scheme.}.

\begin{theorem}
  \label{thm:intro-matroid}
  There is an optimal $(b, \frac{1-e^{-b}}{b})$-balanced contention
  resolution scheme for any matroid polytope. Moreover the scheme is monotone
  and efficiently implementable.
\end{theorem}

The main idea in proving the preceding theorem is consider
a randomized \crs scheme and view it abstractly as a convex
combination of deterministic \crs schemes. This allows, via LP duality,
to show that the best contention resolution scheme for a constraint
system is related to the notion of correlation gap for
weighted rank functions of the underlying constraint.  We reiterate
that the scheme depends on the fractional solution $\bx$ that we
wish to round; the alteration of the random set $R(\bx)$ is itself a
randomized procedure that is tailored to $\bx$, and is found by solving
a linear program. We are inspired
to make the general connection to correlation gap due to the recent work
of Yan~\cite{Yan11}; he applied a similar idea in the context of
greedy posted-price ordering schemes for Bayesian mechanism design,
improving the bounds of \cite{ChawlaHMS09,ChawlaHMS10}. 

\subsection{A framework for rounding via contention resolution schemes}
\label{sec:framework}
We now describe our framework for the problem $\max_{S \in \cI} f(S)$.
The framework assumes the following: (i) there is a polynomial-time
value oracle for $f$, (ii) there is a solvable down-monotone polytope
$P_{\cI}$ that contains the set $\{\b1_S | S \in \cI\}$, and (iii)
there is a monotone $c$-balanced contention resolution scheme $\pi$
for $P_{\cI}$.  Then we have the following simple algorithm:
\begin{enumerate}
\item Using an approximation algorithm, obtain in polynomial time a
  point $\bx^* \in P_{\cI}$ such that 
  $$F(\bx^*) \ge \alpha \cdot \max\{F(\bx)  \mid \bx \in P_{\cI} \cap \{0,1\}^N \} \ge \alpha \cdot \max_{S \in \cI} f(S) .$$ 
\item Round the point $\bx^*$ using the CR scheme $\pi$
to obtain $I=\pi_{\bx^*}(R(\bx^*))$, and return its pruned
version $\eta_f(I)$.
\end{enumerate}

\begin{theorem}
  \label{thm:framework}
  The preceding framework gives a randomized $(\alpha \,
  c)$-approximation algorithm for $\max_{S \in \cI} f(S)$, whenever
  $f$ is non-negative submodular, $\alpha$ is the approximation ratio
  for $\max\{F(\bx) \mid \bx \in P_{\cI} \cap \{0,1\}^N \}$ and
  $P_\cI$ admits a monotone $c$-balanced \crs scheme. If $f$ is
  monotone then the pruning step is not needed. If $f$ is modular then
  the ratio is $c$ and the \crs scheme is not even constrained to be monotone.
\end{theorem}

\begin{proof}
  We have $F(\bx^*) \ge \alpha \opt$ with $\opt = \max_{S \in \cI}
  f(S)$.  Theorem~\ref{thm:submod-balanced} shows that
  $\E[f(\eta_f(I))] \geq c F(\bx^*)$, hence $\E[f(\eta_f(I))] \ge
  \alpha c \opt$. If $f$ is monotone, the pruning step is not required
  by Theorem~\ref{thm:submod-balanced}.
  
  For modular $f$, $F(\bx)$ is a linear function, and hence $\alpha=1$
  can be obtained by linear programming. Moreover, if $F(\bx)$ is a
  linear function, then by linearity of expectation, $\E[f(I)] \ge c
  F(\bx^*)$ without any monotonicity assumption on the scheme.
 \end{proof}

 For non-monotone submodular functions,
 Theorem~\ref{thm:multilinear-solve} gives $\alpha = 0.309$; the currently best
 known approximation is $(\frac{1}{e}-\eps) \simeq 0.367$ due
 to \cite{FNS11}. For monotone submodular functions an optimal
 bound of $\alpha=1-\frac{1}{e}$ is given in \cite{Vondrak08}.

\mypara{Combining schemes for different constraints:} We are
particularly interested in the case when $\cI = \cap_{i=1}^h \cI_i$ is
the intersection of several different independence systems on $N$;
each system corresponds to a different set of constraints that we
would like to impose. Assuming that we can apply the above framework 
to each $\cI_i$ separately, we can obtain
an algorithm for $\cI$ as follows.

\begin{lemma}
  \label{lem:compose}
  Let $\cI = \cap_{i=1}^h \cI_i$ and $P_{\cI} = \cap_i P_{\cI_i}$.
  Suppose each $P_{\cI_i}$ has a monotone $(b,c_i)$-balanced \crs
  scheme. Then $P_{\cI}$ has a monotone $(b,\prod_i c_i)$-balanced
  \crs scheme. In the special case that each element of $N$
  participates in at most $k$ constraints and $c_i = c$ for all $i$
  then $P_{\cI}$ has a monotone $(b, c^k)$-balanced \crs scheme.
  Moreover, if the scheme for each $P_{\cI_i}$ is implementable in
  polynomial time then the combined scheme for $P_{\cI}$ can be
  implemented in polynomial time.
\end{lemma}

Therefore, we can proceed as follows.  Let $P_{\cI_i}$ be a polytope
that is the relaxation for $\cI_i$. In other words $\{\b1_{S} : S \in
\cI_i\}$ is contained in $P_{\cI_i}$.  Let $P_{\cI} = \cap_i
P_{\cI_i}$. It follows that $\{\b1_{S} : S \in \cI\}$ is contained in
$P_{\cI}$ and also that there is a polynomial-time separation oracle
for $P_{\cI}$ if there is one for each $P_{\cI_i}$.  Now suppose there
is a monotone $(b, c_i)$-balanced contention resolution scheme for
$P_{\cI_i}$ for some common choice of $b$. It follows from
Lemma~\ref{lem:compose} that $P_{\cI}$ has a monotone $(b,\prod_i
c_i)$-balanced contention resolution scheme, which can be transformed
into a $(b\prod_i c_i)$-balanced scheme for $P_\cI$.  We can then
apply Theorem~\ref{thm:framework} to obtain a randomized $(\alpha b
\prod_i c_i)$-approximation for $\max_{S \in \cI} f(S)$ where $\alpha$
depends on whether $f$ is modular, monotone submodular or non-monotone
submodular.

In this paper we focus on the framework with a small list of
high-level applications. We have not attempted to optimize for the
best possible approximation for special cases. We add two remarks that
are useful in augmenting the framework.

\begin{remark}
  Whenever the rounding step of our framework is performed by a \crs
  scheme that was obtained from a $(b,c)$-balanced CR scheme---in particular
  in the context mentioned above when combining \crs schemes for different
  constraints---we can often strengthen the procedure as follows.
  Instead of approximately solving $\max_{\bx\in P_\cI} F(\bx)$,
  we can approximately solve $\max_{y\in b P_\cI} F(\by)$ to obtain $\by^*$,
  and then directly apply the $(b,c)$-balanced scheme to $\by^*$, without
  transforming it first to a $bc$-balanced scheme.
  This may be advantageous if
  the problem $\max_{\by \in b P_{\cI}} F(\by)$ admits a direct approximation
  better than one obtained by scaling from $\max_{\by \in P_{\cI}} F(\by)$.
  A useful fact here is that  the continuous greedy algorithm for monotone submodular
  functions \cite{Vondrak09,CCPV09} finds
  for every $b \in [0,1]$ a point $\by^* \in b P_{\cI}$ such
  that $F(\by^*) \geq (1-e^{-b}) \max_{\bx \in P_{\cI}} F(\bx)$.
  This is indeed a stronger guarantee than the one obtained by first applying
  the continuous greedy to $P_\cI$ to obtain $\bx^*$, and then used the
  scaled-down version $b\bx^*$, which leads to a guarantee of 
  only $F(b\bx^*) \geq bF(\bx^*) \geq b(1-e^{-1})\max_{\bx\in P_\cI} F(\bx)$.
\end{remark}

\begin{remark}
  \label{remark:subadditive}
  A non-negative submodular set function $f$ is also subadditive,
  that is, $f(A) + f(B) \ge f(A \cup B)$. In some settings when
  considering the problem $\max_{S \in \cI} f(S)$, it may be
  advantageous to partition the given ground set $N$ into $N_1,\ldots,
  N_h$, separately solve the problem on each $N_i$, and then return
  the best of these solutions. This loses a
  factor of $h$ in the approximation but one may be able to obtain
  a good \crs scheme for each $N_i$ separately while it may not be 
  straightforward to obtain one for the entire set $N$.
\end{remark}

An application of the technique mentioned in Remark~\ref{remark:subadditive}
can be found in Section~\ref{sec:CPIP}, where we use it in the context of
column-restricted packing constraints.

\mypara{Organization:} The rest of the paper is divided into three
parts.  Some illustrative applications of our framework are discussed
in Section~\ref{sec:apps}.  Constant factor approximation algorithms
for maximizing $F$ over a solvable polytope are described in
Section~\ref{sec:multilinear}.  Section~\ref{sec:CRS} discusses the
construction of \crs schemes. This include a discussion of the
connection between contention resolution schemes and correlation gap
and its use in deriving optimal schemes for matroids. Furthermore, in
the same section, we present \crs schemes for knapsack
constraints, sparse packing systems, and UFP in paths and trees.

\section{Applications}
\label{sec:apps}

In this section we briefly outline some concrete results that can be
obtained via our framework. The meta-problem we are interested in
solving is $\max_{S \in \cI} f(S)$ where $\cI$ is a downward-closed
family over the given ground set $N$ and $f$ is a non-negative
submodular set function over $N$. Many interesting problems can be
cast as special cases depending on the choice of $N$, $\cI$ and $f$.
In order to apply the framework and obtain a polynomial-time
approximation algorithm, we need a solvable relaxation $P_{\cI}$ and a
corresponding $(b,c)$-balanced \crs scheme. Note that the framework is essentially
indifferent to $f$ as long as we have a polynomial-time value oracle
for it. We therefore focus on some broad classes of constraints
and corresponding natural polyhedral relaxations, and discuss 
\crs schemes that can be obtained for them. These schemes are formally
described in Section~\ref{sec:CRS}.

\smallskip
\noindent{\em Matroids and matchoids:}
Let $\cM = (N,\cI)$ be a matroid constraint on $N$. A natural
candidate for $P_{\cI}$ is the integral matroid polytope $\{ x \in
[0,1]^n \mid x(S) \le r(S), S \subseteq N\}$ where $r:2^N \rightarrow
\mathbb{Z}_+$ is the rank function of $\cM$.  We develop an optimal
$(1-1/e)$-balanced \crs scheme for the matroid polytope. More
generally, for any $b \in (0,1]$ we design a
$(b,\frac{1-e^{-b}}{b})$-balanced \crs scheme, which lends itself well
to combinations with other constraints. The \crs scheme for the
matroid polytope extends via Lemma~\ref{lem:compose} to the case when
$\cI$ is induced by the intersection of $k$ matroid constraints on
$N$. A more general result is obtained by considering $k$-uniform
matchoids, a common generalization of $k$-set packing and intersection
of $k$ matroids \cite{LeeSV10}, defined as follows. Let $G=(V, N)$ be
a $k$-uniform hypergraph; we associate the edges of the hypergraph
with our ground set $N$. For each $v \in V$, there is a matroid
$\cM_v=(N_v, \cI_v)$ over $N_v$, set of hyperedges in $N$ that contain
$v$. This induces an independence family $\cI$ on $N$ where $\cI = \{
S \subseteq N \mid S \cap N_v \in \cI_v, v \in V\}$.  $k$-uniform
matchoids generalize the intersection of $k$ matroids in that they
allow many matroids in the intersection as long as a given element of
the ground set participates in at most $k$ of them.  A natural
solvable relaxation for $\cI$ is the intersection of the matroid
polytopes at each $v$. Via the \crs scheme for the single matroid
and Lemma~\ref{lem:compose} we
obtain a $(b, (\frac{1-e^{-b}}{b})^k)$-balanced \crs scheme for any $b
\in (0,1]$ for $k$-uniform matchoids.  The choice of $b =
\frac{2}{k+1}$ gives a $\frac{2}{e(k+1)}$-balanced \crs scheme for
every $k$-uniform matchoid.
  
\medskip
\noindent{\em Knapsack / linear packing constraints:}
Let $N = \{1,2,\ldots, n\}$. Given a non-negative $m \times n$ matrix
$A$ and non-negative vector $\bb$, let $\cI = \{ S \mid A \b1_S \le
\bb \}$ where $\b1_S$ is the indicator vector of set $S \subseteq
N$. It is easy to see that $\cI$ is an independence family.  A natural
LP relaxation for the problem is $P_\cI = \{ \bx \mid A\bx \le \bb, x
\in [0,1]^n\}$.  The {\em width} of the system of inequalities is
defined as $W=\floor{\min_{i,j} b_i/A_{i,j}}$.  Some special cases of
interest are (i) $A$ is a $\{0,1\}$-matrix, (ii) $A$ is
column-restricted, that is, all non-zero entries in each column are
the same and (iii) $A$ is $k$-column sparse, that is at most $k$
non-zero entries in each column. Several combinatorial problems can be
captured by these, such as matchings and independent sets in graphs and
hypergraphs, knapsack and its variants, and maximum throughput routing
problems. 
However, the maximum independent set problem in graphs, which is
a special case as mentioned, does not allow a $n^{1-\eps}$-approximation
for any fixed $\eps > 0$, unless P$=$NP~\cite{Hastad}.
Therefore attention has
focused on restricting $A$ in various ways and obtaining upper bounds
on the integrality gap of the relaxation $P_\cI$ when the objective
function is linear. Several of these results are based on randomized
rounding of a fractional solution and one can interpret the rounding
algorithms as \crs schemes. We consider a few such results below.
  \begin{itemize}
  \item For a constant number of knapsack constraints ($m=O(1)$),
    by guessing and enumeration tricks, one can ``effectively'' get a
    $(1-\eps,1-\eps)$-balanced \crs scheme for any fixed $\eps > 0$. %
  \item When $A$ is $k$-column sparse, there is a $(b,
    1-2kb)$-balanced \crs scheme.  If $A$ has in addition width $W
    \geq 2$, there is a $(b, 1 - k (2eb)^{W-1})$ \crs scheme for any $b
    \in (0,1)$.  These results follow from \cite{BansalKNS10}.
  \item When $A$ is a $\{0,1\}$-matrix induced by the problem of routing
    unit-demand paths in a capacitated path or tree, there is a $(b, 1 -
    O(b))$ \crs scheme implicit in \cite{CCKR01,CCGK07,ChekuriMS03}. This can be
    extended to the unsplittable flow problem (UFP) in capacitated paths and trees via
    grouping and scaling techniques
    \cite{KolliopoulosS04,ChekuriMS03,ChekuriEK09}.
\end{itemize}

\smallskip
Section~\ref{sec:CRS} has formal details of the claimed
\crs schemes. There are other rounding schemes in the literature for
packing problems, typically developed for linear functions, that can
be reinterpreted as \crs schemes. Our framework can then be used to
obtain algorithms for non-negative submodular set functions. See
\cite{ChanH09} for a recent and illuminating example.

\mypara{Approximation algorithms.}  The \crs schemes mentioned above
when instantiated with suitable parameters and plugged into our
general framework yield several new randomized polynomial-time
approximation algorithms for problems of the form $\max_{S \in \cI}
f(S)$, where $f$ is non-negative submodular. We remark that
these results are for some-what abstract problems and one can obtain
more concrete results by specializing them and improving the constants.
We have not attempted to do so in this paper.

\begin{itemize}
\item If $\cI$ is the intersection of a fixed number of knapsack
  constraints, we achieve a $0.309$-approximation, improving the
  $(0.2-\eps)$-approximation from \cite{LeeMNS09} and a recent
  $(0.25-\eps)$-approximation \cite{KulikST10}. This is obtained
  via the $(1-\eps,1-\eps)$-balanced \crs scheme for a fixed number
  of knapsack constraints.
\item If $\cI$ is the intersection of a $k$-uniform matchoid and
  $\ell$ knapsack constraints with $\ell$ a fixed constant, we obtain
  an $\Omega(\frac{1}{k})$-approximation (constant independent of $\ell$),
  which improves the bound of $\Omega(\frac{1}{k+\ell})$ from \cite{GuptaNR10}.
  We remark that this is a new result even for linear objective functions. 
  We obtain this by choosing $b = \Omega(1/k)$ and using the 
  $(b, (\frac{1-e^{-b}}{b})^k)$-balanced \crs scheme for $k$-uniform matchoids
  and the $(1-\eps,1-\eps)$-balanced \crs scheme for a fixed number of
  knapsack constraints (this requires a separate preprocessing step).
\item If $\cI$ is the intersection of a $k$-uniform matchoid and an
  $\ell$-sparse knapsack constraint system of width $W$, we give an
  $\Omega(\frac{1}{k+\ell^{1/W}})$-approximation, improving the
  $\Omega(\frac{1}{k \ell})$ approximation from \cite{GuptaNR10}.
  This follows by combining the \crs schemes for $k$-uniform matchoid
  and $\ell$-column sparse packing constraints with a choice of
  $b = \Omega(\frac{1}{k+\ell^{1/W}})$.
\item We obtain a constant factor approximation for maximizing a
  non-negative submodular function of routed requests in a capacitated
  path or tree. Previously an $O(1)$ approximation was known for
  linear functions \cite{CCKR01,CCGK07,ChekuriMS03,ChekuriEK09}.
\end{itemize}

\section{Solving the multilinear relaxation for non-negative
  submodular functions}
\label{sec:multilinear}

In this section, we address the question of solving the problem $\max
\{F(\bx): \bx \in P\}$ where $F$ is the multilinear extension of a
submodular function. As we already mentioned, due to
\cite{Vondrak08,CCPV09}, there is a $(1-1/e)$-approximation for the
problem $\max \{F(\bx): \bx \in P\}$ whenever $F$ is the multilinear
extension of a monotone submodular function and $P$ is any solvable
polytope.  Here, we consider the maximization of a possibly {\em
  non-monotone submodular function} over a down-monotone solvable
polytope. We assume in the following that $P \subseteq [0,1]^N$ is a
down-monotone solvable polytope and $F:[0,1]^N \rightarrow \RR_+$ is
the multilinear extension of a submodular function. We present two
algorithms for this problem.  As we noted in the introduction, there
is no constant-factor approximation for maximizing non-monotone
submodular functions over general---i.e., not necessarily
down-monotone---solvable polytopes \cite{Vondrak09}. The approximation
that can be achieved for matroid base polytopes is proportional to
$1-1/\nu$ where $\nu$ is the fractional packing number of bases (see
\cite{Vondrak09}), and in fact this trade-off generalizes to arbitrary
solvable polytopes; we discuss this in
Appendix~\ref{sec:general-polytopes}.

\subsection{Continuous local-search}

Here we present our first algorithm for the problem $\max \{F(\bx): \bx \in P\}$.
We remark that in the special case of multiple knapsack constraints,
this algorithm is equivalent to the algorithm of \cite{LeeMNS09}.

First we consider a natural local-search algorithm that tries to find
a local optimum for $F$ in the polytope $P$. For a continuous function
$g$ defined over a convex set $C \subseteq \mathbb{R}^n$, a point $\bx
\in C$ is a local optimum (in particular, a maximum), if $g(\bx) \ge
g(\bx')$ for all $\bx' \in C$ in a neighborhood of $\bx$.  If $g$ is
differentiable over $C$, a first-order necessary condition for $\bx$
to be a local maximum is that $(\by-\bx) \cdot \nabla g(\bx) \le 0$
for all $\by \in C$. If $g$ is in addition a concave function then
this is in fact sufficient for $\bx$ to be a global maximum. However,
in general the first-order necessary condition is not sufficient to
guarantee even a local optimum. Although sufficient conditions based on
second-order partial derivatives exist, it is non-trivial to find a
local optimum or to certify that a given point $\bx$ is a local
optimum.  Our algorithms and analysis rely only on finding a point
which satisfies (approximately) the first-order necessary condition.
Hence, this point is not necessarily a local optimum in the classical
sense.  Nevertheless, for notational convenience we refer to any such point
as a local optimum (sometimes such a point is referred to as a
constrained critical point). A simple high-level procedure to find
such a local optimum for $F(\bx)$ in $P$---which does not consider
implementability---is the following. We will subsequently discuss how
to obtain an efficient version of this high-level approach that
returns an approximate local optimum.

\begin{algorithm}{Continuous local search:}
\label{alg:local-search}
Initialize $\bx:=0$. As long as there is $\by \in P$ such that
$(\by-\bx) \cdot \nabla F(\bx) > 0$, %
move $\bx$ continuously in the direction $\by-\bx$. If
there is no such $\by \in P$, return $\bx$.
\end{algorithm}

This algorithm is similar to gradient descent (or
rather ascent), and without considering precision and convergence
issues, it would be equivalent to it.  The importance of the
particular formulation that we stated here will become more clear when
we discretize the algorithm, in order to argue that it terminates in
polynomial time and achieves a solution with suitable properties.

The objective function $F$ is not concave; however, submodularity
implies that along any non-negative direction $F$ is concave (see
\cite{Vondrak08,CCPV09}). This leads to the following basic lemma and
its corollary about local optima that we rely on in the analysis of
our algorithms. In the following $\bx \vee \by$ denotes
the vector obtained by taking the coordinate-wise maximum of 
the vectors $\bx$ and $\by$; and $\bx \wedge \by$ denotes
the vector obtained by taking the coordinate-wise minimum.

\begin{lemma}
\label{lem:G-bound}
For any two points $\bx, \by \in [0,1]^N$ and the multilinear extension $F:[0,1]^N \rightarrow \RR$ of a submodular function,
$$ (\by-\bx) \cdot \nabla F(\bx) \geq F(\bx \vee \by) + F(\bx \wedge \by) - 2 F(\bx).$$
\end{lemma}

\begin{proof}
By submodularity, $F$ is concave along any line with a nonnegative direction vector,
such as $(\bx \vee \by) - \bx \geq 0$. Therefore,
$$\begin{aligned}
F(\bx \vee \by) - F(\bx) &\leq ((\bx \vee \by) - \bx) \cdot \nabla F(\bx),
\text{ and similarly }\\
 F(\bx \wedge \by) - F(\bx) &\leq ((\bx \wedge \by) - \bx) \cdot \nabla F(\bx),
\end{aligned}$$
because of the concavity of $F$ along direction $(\bx \wedge \by) - \bx \leq 0$.
Adding up these two inequalities, we get
$ F(\bx \vee \by) + F(\bx \wedge \by) - 2 F(\bx) \leq ((\bx \vee \by) + (\bx \wedge \by) - 2 \bx)
 \cdot \nabla F(\bx).$
It remains to observe that $(\bx \vee \by) + (\bx \wedge \by) = \bx + \by$, which proves the lemma.
\end{proof}

\begin{corollary}
\label{cor:local-opt}
If $\bx$ is a local optimum in $P$, i.e.~$(\by - \bx) \cdot \nabla F(\bx) \leq 0$ for all $\by \in P$, then
$2 F(\bx) \geq F(\bx \vee \by) + F(\bx \wedge \by)$  for any $\by \in P$.
\end{corollary}

\subsection{Discretized local search}

What follows is a discretization of Algorithm~\ref{alg:local-search},
which is the one we actually use in our framework.  Let $M=
\max\{f(i), f(N-i) : i\in N\}$. Notice that $M$ is an upper bound on
the maximum absolute marginal value of any element, i.e., $M\geq
\max_{S,i} |f_S(i)| = \max \{ f(i)-f(\emptyset), f(N-i) - f(N): i \in
N \}$.  By subadditivity, we have $|f(S)| \leq Mn$ for all $S$.  It
can be also verified easily that $|\partdiff{F}{x_i}| \leq M$ and
$|\mixdiff{F}{x_i}{x_j}| \leq 2M$ for all $i,j$ (see
\cite{Vondrak09}).  We pick a parameter $q=n^a$ for some sufficiently
large constant $a>3$ and maintain a convex combination $\bx =
\frac{1}{q} \sum_{i=1}^{q} \bv_i$, where $\bv_i$ are certain points in
$P$ (without loss of generality vertices, with possible repetition).
Each discrete step corresponds to replacing a vector in the convex
combination by another.  Instead of the gradient $\nabla{F}(\bx)$, we
use an estimate of its coordinates $\partdiff{F}{x_i}$ by random
sampling.  We use the following lemma to control the errors in our
estimates.

\begin{lemma}
  Let $\tilde{F}(\bx) = \frac{1}{H} \sum_{h=1}^{H} f(R_h)$ where $R_h$
  is a random set sampled independently with probabilities $x_i$. Let
  $H = n^{2a+1}$, $\delta = M / n^{a-1}$ and $M = \max\{f(i), f(N-i):
  i\in N\}$. Then the probability that $|\tilde{F}(\bx) - F(\bx)| >
  \delta$ is at most $2e^{-n/8}$.
\end{lemma}

\begin{proof}
  Let us define $X_h = \frac{1}{2Mn} (f(R_h) - F(\bx))$, a random
  variable bounded by $1$ in absolute value.  By definition, $\E[X_h]
  = 0$. By the Chernoff bound, $ \Pr[|\sum_{h=1}^{H} X_h| > t] <
  2e^{-t^2 / 2H}$ (see Theorem A.1.16 in \cite{AlonS-book}). We set $H
  = n^{2a+1}$ and $t = \frac12 n^{a+1}$, and obtain $
  \Pr[|\tilde{F}(\bx) - F(\bx)| > M / n^{a-1}] = \Pr[|\sum_{h=1}^{H}
  X_h| > \frac12 n^{a+1}] < 2e^{-n/8}.$
\end{proof}

Given estimates of $F(\bx)$, we can also estimate $\partdiff{F}{x_i} =
F(\bx \vee \be_i) - F((\bx \vee \be_i) - \be_i) = \E[f(R+i) -
f(R-i)]$.  The above implies the following bound.

\begin{corollary}
  Let $\delta = M / n^{a-1}$.  If the total number of evaluations of
  $F$ and $\partdiff{F}{x_i}$ is bounded by $n^b$ and each estimate is
  computed independently using $n^{2a+1}$ samples, then with
  probability at least $1 - O(n^b e^{-n/8})$ {\em all} the estimates are
  within $\pm \delta$ additive error .
\end{corollary}

The algorithm works as follows. The input to the algorithm is a submodular function $f$ given by a value oracle,
and a polytope $P$ given by a separation oracle.

\begin{algorithm}{Fractional local search.}
\label{alg:fr-local-search}
Let $q = n^a, \delta = M/n^{a-1}$. Let $\bx:= \frac{1}{q} \sum_{i=1}^{q} \bv_i$, and initialize $\bv_i = 0$ for all $i$.
Use estimates $\tilde{\nabla}F(\bx)$ of $\nabla F(\bx)$ within $\pm \delta$ in each coordinate.
As long as there is $\by \in P$ such that $(\by-\bx) \cdot \tilde{\nabla}F(\bx) > 4 \delta n$
(which can be found by linear programming), we modify $\bx:= \frac{1}{q} \sum_{i=1}^{q} \bv_i$
by replacing one of the vectors $\bv_i$ in the linear combination by $\by$, so that we maximize $F(\bx)$.
If there is no such $\by \in P$, return $\bx$.
\end{algorithm}

\begin{lemma}
Algorithm~\ref{alg:fr-local-search} terminates in polynomial time with high probability.
\end{lemma}

\begin{proof}
We show that if all estimates of $\nabla F$ computed during the algorithm
are within $\pm \delta$ in each coordinate---which happens with high
probability---then the algorithm terminates in polynomial time.
This implies the lemma since with high probability, we have
that a polynomial number of estimates of $\nabla F$ are indeed
all within $\pm \delta$ in each coordinate. Hence, we assume in the following that
all estimates $\tilde{\nabla} F$ of $\nabla F$ are within $\pm \delta$.

In each step, the algorithm continues only if it finds $\by \in P$ such that
$(\by - \bx) \cdot \tilde{\nabla}F(\bx) \geq 4 \delta n$. Since $\tilde{\nabla}F$
approximates $\nabla F$ within $\pm \delta$ in each coordinate, this means that
$(\by - \bx) \cdot \nabla F(\bx) \geq 3 \delta n$.
Denote by $\bx'$ a random vector that is obtained by replacing a random vector $\bv_i$ by $\by$,
in the linear combination $\bx = \frac{1}{q} \sum_{i=1}^{q} \bv_i$.
The expected effect of this change is
$$ \E[F(\bx') - F(\bx)] = \frac{1}{q} \sum_{i=1}^{q} \left(F\left(\bx + \frac{1}{q} (\by - \bv_i)\right) - F(\bx)\right)
= \frac{1}{q^2} \sum_{i=1}^{q} (\by - \bv_i) \cdot \nabla F(\tilde{\bx}_i) $$
where $\tilde{\bx}_i$ is some point on the line between $\bx$ and $\bx + \frac{1}{q} (\by - \bv_i)$,
following from the mean-value theorem. 
Since $q = n^a$ and the second partial derivatives of $F$ are bounded by $2M$, we get by standard bounds that
$||\nabla F(\tilde{\bx}_i) - \nabla F(\bx)||_1 \leq \frac{n^2}{q} \cdot 2M = \frac{2M}{n^{a-2}} = 2 \delta n$.
Using also the fact that $\by - \bv_i \in [-1,1]^n$,
$$ \E[F(\bx') - F(\bx)] \geq \frac{1}{q^2} \sum_{i=1}^{q} \left((\by - \bv_i) \cdot \nabla F(\bx)
  - 2 \delta n \right) = \frac{1}{q} \left((\by - \bx) \cdot \nabla F(\bx) - 2 \delta n \right) \geq \frac{1}{q} \cdot \delta n $$
using the fact that $(\by -\bx) \cdot \nabla F(\bx) \geq 3 \delta n$. Therefore, if we
exchange $\by$ for the vertex $\bv_i$ that maximizes our gain, we gain at least
$F(\bx') - F(\bx) \geq \frac{1}{q} \delta n = \frac{M}{n^{2a-2}}$. 
Also we have the trivial bound $\max F(\bx) \leq n M$; therefore the number of steps is bounded by $n^{2a-1}$.
\end{proof}

\begin{lemma}
\label{lem:fract-opt}
If $\bx$ is the output of Algorithm~\ref{alg:fr-local-search}, then
with high probability
$$ 2 F(\bx) \geq F(\bx \vee \by) + F(\bx \wedge \by) - 5 \delta n $$
for every $\by \in P$.
\end{lemma}

\begin{proof}
  If the algorithm terminates, it means that for every $\by \in P$,
  $(\by - \bx) \cdot \tilde{\nabla}F(\bx) \leq 4 \delta n$.
  Considering the accuracy of our estimate of the gradient
  $\tilde{\nabla}F(\bx)$ (with high probability), this means that
  $(\by - \bx) \cdot {\nabla}F(\bx) \leq 5 \delta n$.  By
  Lemma~\ref{lem:G-bound}, we have $(\by - \bx) \cdot \nabla F(\bx)
  \geq F(\bx \vee \by) + F(\bx \wedge \by) - 2 F(\bx)$.  This proves
  the lemma.
\end{proof}

\subsection{Repeated local search: a 0.25-approximation}

Next, we show that how to design a $0.25$-approximation to the multilinear optimization problem using two runs
of the fractional local-search algorithm. The following is our algorithm.

\begin{algorithm}
\label{alg:2-stage}
Let $\bx$ be the output of Algorithm~\ref{alg:fr-local-search} on the polytope $P$.
Define $Q = \{\by \in P: \by \leq \b1-\bx\}$ and let $\bz$ be the output of Algorithm~\ref{alg:fr-local-search} on the polytope $Q$.
Return the better of $F(\bx)$ and $F(\bz)$.
\end{algorithm}

We use the following property of the
multilinear extension of a submodular function.  Let us replace each
coordinate by a $[0,1]$ interval and let us represent a certain value
$x_i$ of the $i$'th coordinate by a subset of $[0,1]$ of the corresponding measure.

\begin{definition}
\label{def:rep}
  Let $\cX \in \mathcal{L}^N$, where $\mathcal{L}$ denotes
  the set of all measurable subsets of $[0,1]$.
  We say that $\cX$ represents a vector $\bx \in [0,1]^N$,
  if $\cX_i$ has measure $x_i$ for each $i \in N$.
\end{definition}

From a "discrete point of view", we can imagine that each coordinate
is replaced by some large number of elements $M$ and a value of $x_i$
is represented by any subset of size $M x_i$. This can be carried out
if all the vectors we work with are rational.
In the following, we consider functions on subsets of this new ground
set. We show a natural property, namely that a function derived from
the multilinear extension of a submodular function is again
submodular. (An analogous property in the discrete case was
proved in \cite{MSV08,LeeMNS09}.)

\begin{lemma}
\label{lem:fract-submod}
Let $F:[0,1]^N \rightarrow \RR$ be a multilinear extension of a
submodular function $f$.  Define a function $F^*$ on $\mathcal{L}^N$,
by $F^*(\cX) = F(\bx)$, where $\bx \in [0,1]^N$ is the vector
represented by $\cX$.
Then $F^*$ is submodular:
$$ F^*(\cX \cup \cY) + F^*(\cX \cap \cY) \leq F^*(\cX) + F^*(\cY),$$
where the union and intersection is interpreted component-wise.
\end{lemma}

\begin{proof}
  We have $F(\bx) = \E[f(\hat{\bx})]$ where $\hat{x}_i = 1$
  independently with probability $x_i$.  An equivalent way to generate
  $\hat{\bx}$ is to choose any set $\cX \in \mathcal{L}^N$
  representing $\bx$, generate uniformly and independently a number
  $r_i \in [0,1]$ for each $i \in N$, and set $\hat{x}_i=1$ iff $r_i
  \in \cX_i$. Since the measure of $\cX_i$ is $x_i$, $\hat{x}_i=1$
  with probability exactly $x_i$.  Therefore,
$$ F^*(\cX) = F(\bx) = \E[f(\hat{\bx})] = \E[f(\{i: r_i \in \cX_i\})]. $$
Similarly,
$$ F^*(\cY) = \E[f(\{i: r_i \in \cY_i\})]. $$
This also holds for $\cX \cup \cY$ and $\cX \cap \cY$: since $(\cX \cup \cY)_i = \cX_i \cup \cY_i$ and
 $(\cX \cap \cY)_i  = \cX_i \cap \cY_i$, we get
$$ F^*(\cX \cup \cY) %
 = \E[f(\{i: r_i \in \cX_i\} \cup \{i: r_i \in \cY_i\})] $$
and 
$$ F^*(\cX \cap \cY) %
 = \E[f(\{i: r_i \in \cX_i\} \cap \{i: r_i \in \cY_i\})]. $$
Hence, by the submodularity of $f$,
\begin{eqnarray*}
F^*(\cX \cup \cY) + F^*(\cX \cap \cY) 
& = & \E[f(\{i: r_i \in \cX_i\} \cup \{i: r_i \in \cY_i\}) 
  + f(\{i: r_i \in \cX_i\} \cap \{i: r_i \in \cY_i\})] \\
& \leq & \E[f(\{i: r_i \in \cX_i\}) + f(\{i: r_i \in \cY_i\})] \\
& = &  F^*(\cX) + F^*(\cY).
\end{eqnarray*}
\end{proof}

From here, we obtain our main lemma - the average of the two fractional local optima is at least $\frac14 \opt$.

\begin{lemma}
\label{lem:add-complement}
Let $\opt = \max \{F(\bx): \bx \in P\}$.  Let $\bx$ be the output of Algorithm~\ref{alg:fr-local-search}
on polytope $P$, and $\bz$ an output of Algorithm~\ref{alg:fr-local-search} on polytope $Q = \{\by \in P:\by \leq \b1-\bx\}$,
with parameter $\delta$ as in Algorithm~\ref{alg:fr-local-search}.
Then with high probability, $2 F(\bx) + 2 F(\bz) \geq \opt - 10 \delta n.$
\end{lemma}

\begin{proof}
Let $\opt = F(\bx^*)$ where $\bx^* \in P$. By Lemma~\ref{lem:fract-opt}, the output of the algorithm $\bx$ satisfies
with high probability
\begin{equation}
\label{eq:local-1}
 2 F(\bx) \geq F(\bx \vee \bx^*) + F(\bx \wedge \bx^*) - 5 \delta n.
\end{equation}
In the restricted polytope $Q = \{\by \in P: \by \leq \b1-\bx\}$,
consider the point $\bz^* = (\bx^* - \bx) \vee {\bf 0} \in Q$. Again
by Lemma~\ref{lem:fract-opt}, the output of the algorithm $\bz$ satisfies
\begin{equation}
\label{eq:local-2}
 2 F(\bz) \geq F(\bz \vee \bz^*) + F(\bz \wedge \bz^*) - 5 \delta n.
\end{equation}
Now we use a representation of vectors by subsets as described in Def.~\ref{def:rep}.
We choose $\cX, \cX^*, \cZ, \cZ^* \in \mathcal{L}^N$ to represent $\bx,\bx^*,\bz,\bz^*$ as follows:
for each $i \in N$, $\cX_i = [0,x_i)$, $\cZ_i = [x_i,x_i+z_i)$ (note that $x_i+z_i \leq 1$),
$\cX^*_i = [0,x^*_i)$ and $\cZ^*_i = [0,z^*_i) = [0, \max \{x^*_i-x_i,0\})$.
Note that $(\cX \cap \cZ)_i = \emptyset$ for all $i\in N$.

Defining $F^*$ as in Lemma~\ref{lem:fract-submod}, we have $F^*(\cX) =
F(\bx)$, $F^*(\cX^*) = F(\bx^*) = \opt$, $F^*(\cZ) = F(\bz)$ and $F^*(\cZ^*) = F(\bz^*)$.
Using relations like $[0,x_i) \cup [0,x^*_i) = [0, \max \{x_i,x^*_i\})$, we also get
$F^*(\cX \cup \cX^*) = F(\bx \vee \bx^*)$ and $F^*(\cX \cap \cX^*) = F(\bx \wedge \bx^*).$
Furthermore, we have $(\cX^*_i \setminus \cX_i) \cup \cZ_i = [x_i, \max \{x^*_i, x_i+z_i\})
 = [x_i, x_i + \max \{z^*_i, z_i\})$. This is an interval of length $\max \{z^*_i, z_i\} = (\bz \vee \bz^*)_i$
and hence $ F^*((\cX^* \setminus \cX) \cup \cZ) = F(\bz \vee \bz^*)$,
where $(\cX^* \setminus \cX)\cup \cZ$ is interpreted component-wise.

The property of the first local optimum (\ref{eq:local-1}) can be thus
written as $ 2 F(\bx) \geq F^*(\cX \cup \cX^*) + F^*(\cX \cap \cX^*) - 5 \delta n.$ The
property of the complementary local optimum (\ref{eq:local-2}) can be
written as $ 2 F(\bz) \geq F^*((\cX^* \setminus \cX) \cup \cZ) - 5 \delta n$ (we
discarded the nonnegative term $F(\bz \wedge \bz^*)$ which is not
used in the following). 
Therefore, $2 F(\bx) + 2 F(\bz) \geq F^*(\cX \cup \cX^*) + F^*(\cX \cap \cX^*) + F^*((\cX^* \setminus \cX) \cup \cZ) - 10 \delta n$.
By Lemma~\ref{lem:fract-submod}, $F^*$ is submodular. Hence we get
$$\begin{aligned}
F^*(\cX \cap \cX^*) + F^*((\cX^* \setminus \cX) \cup \cZ) & \geq  F^*((\cX \cap \cX^*) \cup (\cX^* \setminus \cX) \cup \cZ) \\
 & =  F^*(\cX^* \cup \cZ)
\end{aligned}$$
(we discarded the intersection term).
Finally, using the fact that $\cX \cap \cZ = \emptyset$ and again the submodularity of $F^*$, we get
$$ F^*(\cX \cup \cX^*) + F^*(\cX^* \cup \cZ) \geq F^*((\cX \cup \cX^*) \cap (\cX^* \cup \cZ)) = F^*(\cX^*) $$
(we discarded the union term). To summarize,
$$\begin{aligned}
2 F(\bx) + 2 F(\bz) & \geq   F^*(\cX \cup \cX^*) + F^*(\cX \cap \cX^*) + F^*((\cX^* \setminus \cX) \cup \cZ) - 10 \delta n\\
& \geq  F^*(\cX \cup \cX^*) + F^*(\cX^* \cup \cZ) - 10 \delta n \\
& \geq  F^*(\cX^*) - 10 \delta n ~=~ \opt - 10 \delta n.
\end{aligned}$$
\end{proof}

Since the parameter $\delta$ in Algorithm~\ref{alg:fr-local-search} is
chosen as $\delta = \frac{M}{n^{a-1}}$ for some constant
$a>3$, 
we obtain the following.

\begin{corollary}\label{cor:1over4}
  For any solvable down-monotone polytope $P \subseteq [0,1]^N$ and
  multilinear extension of a submodular function $F:[0,1]^N
  \rightarrow \RR_+$, Algorithm~\ref{alg:2-stage} finds with high
  probability a solution of value at least $\frac14 OPT -
  O(\frac{M}{n^{a-2}})$ for the problem $\max \{F(\bx): \bx \in P\}$.
\end{corollary}

We remark that in many settings of interest, $\opt = \max \{F(\bx) :
\bx \in P\} \geq M/\poly(n)$ and thus we can make the additive error
small relative to the optimum by choosing $a$ large enough. This leads
to a multiplicative $(1/4 - o(1))$-approximation. 
A concrete setting of interest is when $P$ is not too thin in any dimension, 
as highlighted by the following lemma which, together
with Corollary~\ref{cor:1over4},
implies Theorem~\ref{thm:multilinear-solve}.

\begin{lemma}\label{lem:largeOPT}
Let $f:2^N \rightarrow \mathbb{R}_{\geq 0}$
be a nonnegative submodular function with
multilinear extension $F$, and let $P\subseteq \mathbb{R}^n$
be a solvable down-monotone polytope satisfying that there is a
$\lambda\in \Omega(\frac{1}{\poly(n)})$ such that
for each coordinate $i\in [n]$,
$\lambda \cdot \be_i \in P$.
Furthermore, let $\opt = \max\{F(x) \mid x\in P\}$,
and $M=\max\{f(i), f(N-i): i\in N\}$.
Then
\begin{equation*}
\opt \geq \Omega\left(\frac{M}{\poly(n)}\right).
\end{equation*}
\end{lemma}
\begin{proof}
Let $i\in N$. Since $F$ is concave along any nonnegative
direction and $\lambda \be_i\in P$, we have
\begin{equation}\label{eq:iToOpt}
\opt \geq F(\lambda \be_i) \geq \lambda F(\be_i) = \lambda f(i)
\quad \forall i\in N.
\end{equation}
Furthermore,
\begin{equation}\label{eq:nMinIToOpt}
f(N-i) \leq \sum_{j\in N-i}f(j) \leq \frac{n}{\lambda}\opt,
\end{equation}
where we used~\eqref{eq:iToOpt} for the last
inequality. Equations~\eqref{eq:iToOpt}
and~\eqref{eq:nMinIToOpt} imply the desired results.
\end{proof}

\subsection{Restricted local search: a 0.309-approximation}
Next, we present a modified local-search algorithm which is a
generalization of the algorithm for matroid polytopes from
\cite{Vondrak09}.  We remark that this algorithm is in fact simpler
than the $\frac14$-approximation from the previous section, in that it
does not require a second-stage complementary local search.
Both algorithms work for any down-monotone polytope $P$. 
However, our analysis of the restricted local-search algorithm is with respect
to the best integer solution in the polytope; we do not know whether
the approximation guarantee holds with respect to the best fractional solution.
\begin{algorithm}
\label{alg:local-t}
Fix a parameter $t \in [0,1]$. Using
Algorithm~\ref{alg:fr-local-search}, find an approximate local optimum
$\bx$ in the polytope $P \cap [0,t]^N$.  Return $\bx$.
\end{algorithm}
We show that with the choice of $t = \frac12 (3 - \sqrt{5})$, this algorithm achieves
a $\frac14 (-1+\sqrt{5}) \simeq 0.309$-approximation with respect to the optimal
integer solution in $P$.

\begin{lemma}
\label{lem:local-t}
Let $\bx$ be an output of Algorithm~\ref{alg:fr-local-search} on $P \cap [0,t]^N$.
Define $\bw \in [0,1]^N$ by $w_i = t$ if $x_i \geq t - 1/n$ and $w_i = 1$ if $x_i < t - 1/n$.
Let $\bz$ be any point in $P$ and let $\bz' = \bz \wedge \bw$. Then with high probability,
$$ 2 F(\bx) \geq F(\bx \vee \bz') + F(\bx \wedge \bz') - 5 \delta n^2.$$
\end{lemma}

We remark that the above inequality would be immediate from
Lemma~\ref{lem:fract-opt}, if $\bz' \in P \cap [0,t]^N$. However,
$\bz'$ is not necessarily constrained by $[0,t]^N$.

\begin{proof}
  Consider $\bz' = \bz \wedge \bw$ as defined above. By
  down-monotonicity, $\bz' \in P$.  Also, the coordinates where $z'_i > t$ are exactly those where $x_i < t - 1/n$.
  So we have $\bx + \frac{1}{n} (\bz' - \bx) \in P \cap [0,t]^N$. By the
  stopping rule of Algorithm~\ref{alg:fr-local-search},
$$ \frac{1}{n} (\bz' - \bx) \cdot \nabla F(\bx) \leq  5 \delta n.$$
By Lemma~\ref{lem:G-bound}, this implies $F(\bx \vee \bz') + F(\bx \wedge \bz') - 2 F(\bx) \leq (\bz'-\bx) \cdot \nabla F(\bx)
 \leq 5 \delta n^2$.
\end{proof}

In the rest of the analysis, we follow \cite{Vondrak09}.

\begin{definition}
  For $\bx \in [0,1]^N$ and $\lambda \in [0,1]$, we define the
  associated ``threshold set'' as $T_{>\lambda}(\bx) = \{i: x_i > \lambda \}$.
\end{definition}

\begin{lemma}
\label{lem:threshold}
Let $\bx \in [0,1]^N$. For any partition $N = C \cup \bar{C}$,
$$ F(\bx) \geq \E[f((T_{>\lambda}(\bx) \cap C) \cup (T_{>\lambda'}(\bx) \cap \bar{C}))] $$
where $\lambda, \lambda' \in [0,1]$ are independently and uniformly
random.
\end{lemma}

This appears as Lemma~A.5 in \cite{Vondrak09}.  We remark that the
right-hand side with $C = \emptyset$ or $C=N$ gives the Lov\'asz
extension of $f$ and the lemma follows by comparing the multilinear
and Lov\'asz extension. For a non-trivial partition $(C,\bar{C})$, the
lemma follows by two applications of this fact.  The next lemma is
exactly as in \cite{Vondrak09} for the special case of a matroid
polytope; we rephrase the proof here in our more general setting.

\begin{lemma}
\label{lem:0.309-bound}
  Assume that $t \in [0,\frac12 (3 - \sqrt{5})]$.  Let $\bx$ be an
  output of Algorithm~\ref{alg:fr-local-search} on $P \cap [0,t]^N$
  (with parameter $a \geq 4$), and let $\bz = \b1_C$ be any integer
  solution in $P$.  Then with high probability,
$$ F(\bx) \geq \left(t - \frac12 t^2 \right) f(C) - O\left( \frac{M}{n^{a-3}} \right).$$
\end{lemma}

\begin{proof}
  Define $A = \{ i: x_i \geq t - 1/n \}$ and let $\bw = t \b1_{A} +
  \b1_{\bar{A}}$, $\bz' = \bz \wedge \bw$ as in Lemma~\ref{lem:local-t}.
  Since $\bz = \b1_C$, we have $\bz' = t  \b1_{A \cap C} + \b1_{C \setminus A}$.  By Lemma~\ref{lem:local-t},
  we get
\begin{equation}
\label{eq:1}
 2 F(\bx) \geq F(\bx \vee \bz') + F(\bx \wedge \bz') - 5 \delta n^2.
\end{equation}
First, let us analyze $F(\bx \wedge \bz')$. Since $\bz' = t \b1_{A \cap C} + \b1_{C \setminus A}$
and $\bx \in [0,t]^N$, we have $\bx \wedge \bz' = \bx \wedge \b1_C$.
We apply Lemma~\ref{lem:threshold}, which states that
$$ F(\bx \wedge \bz') = F(\bx \wedge \b1_C) \geq \E[f(T_{>\lambda}(\bx) \cap C)].$$
Due to the definition of $T_{>\lambda}(\bx)$, with probability $t-1/n$ we
have $\lambda < t-1/n$ and $T_{>\lambda}(\bx)$ contains $A = \{ i: x_i \geq t-1/n \}$.
Then, $f(T_{>\lambda}(\bx) \cap C) + f(C \setminus A) \geq f(C)$
by submodularity.  We conclude that
\begin{equation}
\label{eq:2}
F(\bx \wedge \bz') \geq \left(t - \frac{1}{n}\right) (f(C) - f(C \setminus A)).
\end{equation}
Next, let us analyze $F(\bx \vee \bz')$. We apply
Lemma~\ref{lem:threshold}.  We get
$$ F(\bx \vee \bz') \geq \E[f((T_{>\lambda}(\bx \vee \bz') \cap C)
 \cup (T_{>\lambda'}(\bx \vee \bz') \cap \bar{C}))].$$
 
 The random threshold sets are as follows:
 $T_{>\lambda}(\bx \vee \bz') \cap C = T_{>\lambda}(\bz')$ is
 equal to $C$ with probability $t$, and equal to $C \setminus A$ with probability $1-t$,
 by the definition of $\bz'$.
 $T_{>\lambda'}(\bx \vee \bz') \cap \bar{C} =
 T_{>\lambda'}(\bx) \cap \bar{C}$ is empty with probability $1-t$,
 because $\bx \in [0,t]^N$.
 (We ignore the contribution when $T_{>\lambda'}(\bx) \cap \bar{C}
 \neq \emptyset$.)  Because $\lambda,\lambda'$ are independently
 sampled, we get
$$ F(\bx \vee \bz') \geq (1-t)(t f(C) + (1-t) f(C \setminus A)).$$
Provided that $t \in [0,\frac12 (3 - \sqrt{5})]$, we have $t \leq (1-t)^2$.
Then, we can write
\begin{equation}
\label{eq:3}
F(\bx \vee \bz') \geq t(1-t) f(C) + t f(C \setminus A).
\end{equation}
Combining equations (\ref{eq:1}), (\ref{eq:2}) and (\ref{eq:3}), we get
\begin{eqnarray*}
2 F(\bx) & \geq & F(\bx \vee \bz')  +  F(\bx \wedge \bz') - 5 \delta n^2 \\
 & \geq & t(1-t) f(C) + t \, f(C \setminus A) + (t - \frac{1}{n}) (f(C) - f(C \setminus A))  - 5 \delta n^2 \\
 & \geq & (2t - t^2) f(C) - O\left( \frac{M}{n^{a-3}} \right)
\end{eqnarray*}
using $\delta n^2 = M / n^{a-3}$.
\end{proof}

Next, we show how the error term in Lemma~\ref{lem:0.309-bound} can be compared to the optimal value.
Note that here we use the fact that in this section, we compare to $0/1$ solutions only.
The following Lemma is essentially a specialization of
Lemma~\ref{lem:largeOPT} to the $0/1$ case.

\begin{lemma}
\label{lem:M-vs-OPT}
Suppose that $\be_i \in P$ for each $i \in N$. Let $OPT = \max \{F(\bx): \bx \in P \cap \{0,1\}^N \}$
and $M = \max_{i \in N} \{f(i), f(N-i) \}$.
Then $OPT \geq \frac{1}{n} M$.
\end{lemma}

\begin{proof}

If $M = f(i)$
  for some $i \in N$, then clearly $OPT = \max \{ F(\bx): \bx \in P
  \cap \{0,1\}^N \} \geq M$, because $F(\be_i) = f(i)$ and $\be_i \in
  P$. If $M = f(N-i)$ for some $i \in N$, then consider $\sum_{j \neq
    i} f(j) \geq f(N-i) = M$ which holds by submodularity and
  nonnegativity of $f$.  We have $f(j) \geq \frac{1}{n} M$ for some $j
  \neq i$. By the above argument, $OPT \geq \frac{1}{n} M$.
\end{proof}

Clearly, if $\be_i \notin P$, then coordinate $x_i$ cannot participate
in an integer optimum, $\max \{ F(\bx): \bx \in P \cap \{0,1\}^N
\}$. We can remove all such coordinates from the problem. Therefore,
we can in fact assume that $\be_i \in P$ for all $i \in N$.

\begin{corollary}
  Assume $\be_i \in P$ for all $i \in N$. Then for $t = \frac12 (3 - \sqrt{5})$, Algorithm~\ref{alg:local-t} with high probability
   a $\frac14 (-1+\sqrt{5} - o(1))$-approximation for the problem  $\max \{F(\bx): \bx \in P \cap \{0,1\}^N \}$.
\end{corollary}

\section{Contention resolution schemes}
\label{sec:CRS}

In this section we discuss contention resolution schemes in more
detail and prove our results on the existence of contention resolution
schemes and their application to submodular maximization
problems.

\subsection{Contention resolution basics}

Recall the definition, from Section~\ref{sec:intro}, of a
$(b,c)$-balanced \crs scheme $\pi$ for a polytope $P_\cI$.  We first
prove the claim that a $(b,c)$-balanced \crs scheme $\pi$ can be transformed
into a $b c$-balanced \crs scheme $\pi'$ as follows. Let $\bx \in
P_{\cI}$ and $A\subseteq N$. We define $\pi'_{\bx}(A)$ as follows.
First each element of $A$ is removed independently of the others with
probability $1-b$ to obtain a random set $A'\subseteq A$.  We then set
$\pi'_{\bx}(A) = \pi_{\bx}(A')$.  The key observation is that if $A$
is a set drawn according to the distribution induced by $R(\bx)$, then
$A'$ has a distribution given by $R(b\bx)$. Hence, for any $i\in N$
\begin{align*}
\Pr[i\in \pi'_\bx(R(\bx))\mid i\in R(\bx)]
&= \frac{\Pr[i\in \pi'_\bx(R(\bx))]}{\Pr[i\in R(\bx)]}
= \frac{\Pr[i\in \pi_\bx(R(b\bx))]}{\Pr[i\in R(\bx)]}
= \frac{b \Pr[i\in \pi_\bx(R(b\bx))]}{\Pr[i\in R(b\bx)]}\\
&= b \Pr[i\in \pi_\bx(R(b\bx)) \mid i\in R(b\bx)] \geq bc,
\end{align*}
where the last inequality follows from the fact that $\pi$
is $(b,c)$-balanced.

\medskip
\noindent
{\bf Monotonicity of \crs schemes for submodular function maximization:}
The inequality that relates contention resolution to submodular
maximization is given in Theorem~\ref{thm:submod-balanced}. A proof 
of this inequality also appears in \cite{BansalKNS10} for monotone functions
without the pruning procedure. Before
presenting the proof, we provide some intuition on why monotonicity of
the \crs scheme is needed in the context of submodular function
maximization, and we specify the pruning procedure $\eta_f$.  It is
easy to see that if $P_{\cI}$ has a $c$-balanced \crs scheme then it
implies a $c$-approximation for maximizing a linear function over
$P_{\cI}$.  If $\bx$ is a fractional solution then its value is
$\sum_i w_i x_i$, where $w_i$ are some (non-negative) weights; since
each element $i$ is present in the final solution produced by a
$c$-balanced \crs scheme with probability at least $c x_i$, by
linearity of expectation, the expected weight of a solution returned
by a $c$-balanced scheme is at least $c \sum_i w_i x_i$.  More
generally, we would like to prove such a bound for any submodular
function $f$ via $F$. However, this is no longer obvious since
elements do not appear independently in the rounding scheme; recall
that $F(\bx)$ is the expected value of $f$ on a set produced by
independently including each $i$ with probability $x_i$. Monotonicity
is the property that is useful in this context, because elements of
smaller sets contribute more to a submodular function than elements of
larger sets.

To prove Theorem~\ref{thm:submod-balanced}, we first introduce the
claimed pruning function $\eta_f$.  To prune a set $I$ via the pruning
function $\eta_f$, an arbitrary ordering of the elements of $N$ is
fixed: for notational simplicity let $N=\{1,\dots,n\}$ which gives a
natural ordering.  Starting with $J=\emptyset$ the final set
$J=\eta_f(I)$---which we called the \emph{pruned} version of $I$---is
constructed by going through all elements of $I$ in the order induced
by $N$.  When considering an element $i$, $J$ is replaced by $J+i$ if
$f(J+i)-f(J)>0$.

\begin{proof}[Proof of Theorem~\ref{thm:submod-balanced}]
Let $R=R(\bx)$ and $I=\pi_\bx(R)$,
and let $J=\eta_f(I)$ if $f$ is non-monotone
and $J=I$ otherwise. Hence, in both cases, $J$ is the set returned by
the suggested rounding procedure.

Assume that $N=\{1,\dots, n\}$ is the same ordering of the elements
as used in the pruning operation (in case no pruning was applied, any
order is fine).
The main property we get by pruning is the following.
Notice that this property trivially holds when $f$
is monotone.
\begin{equation}\label{eq:pruneProp}
f_{J\cap [i-1]}(i) \geq 0 \quad \forall i\in J.
\end{equation}
Furthermore, for each $i\in I$ 
\begin{equation}\label{eq:posNotPruned}
f_{J\cap [i-1]}(i) > 0 \quad \Rightarrow \quad
  i\in J.
\end{equation}
Again, notice that this property holds trivially in
the monotone case in which we have $J=I$.

The main step that we will prove is
that for any fixed $i\in \{1,\dots,n\}$,
\begin{equation}\label{eq:pruneToProve}
\E[f(J\cap [i])-f(J\cap [i-1])]\geq c\E[f(R\cap [i])-f(R\cap [i-1])].
\end{equation}
We highlight that there are two sources of randomness over which the
expectation is taken on the left-hand side of the above inequality:
one source is the randomness in choosing the set $R$, and the other
source is the potential randomness of the CR scheme used to obtain the
set $I$ from $R$, which is later deterministically pruned to get $J$.
The theorem then follows from~\eqref{eq:pruneToProve} since
$$ \E[f(J)] = f(\emptyset) + \sum_{i=1}^{n} \E[f(J \cap [i]) - f(J \cap [i-1])]
 \geq f(\emptyset) + c \, \sum_{i=1}^{n} \E[f(R \cap [i]) - f(R \cap [i-1])]
  \geq c \, \E[f(R)].$$

Hence, it remains to prove~\eqref{eq:pruneToProve}.
Consider first the non-monotone case. Here we have
\begin{equation*}
\begin{aligned}
\E[f(J\cap [i])-f(J\cap [i-1])] &= \E[\b1_{i\in J}f_{J\cap [i-1]}(i)]\\
&= \Pr[i\in R]\cdot \E[\b1_{i\in J}f_{J\cap [i-1]}(i)\mid i\in R]\\
&\overset{\eqref{eq:pruneProp}}{=}\Pr[i\in R]\cdot \E[\b1_{i\in J}
  \max\{0,f_{J\cap [i-1]}(i)\}\mid i\in R]\\
&\overset{\eqref{eq:posNotPruned}}{=}\Pr[i\in R]\cdot \E[\b1_{i\in I}
  \max\{0,f_{J\cap [i-1]}(i)\}\mid i\in R]\\
&\geq \Pr[i\in R]\cdot \E[\b1_{i\in I}
  \max\{0,f_{R\cap [i-1]}(i)\} \mid i\in R] \quad \quad \mbox{(since $f$ is submodular)}\\
&= \Pr[i\in R]\cdot \E[\E[\b1_{i\in I}
  \max\{0,f_{R\cap [i-1]}(i)\}\mid R]\mid i\in R]\\
&= \Pr[i\in R]\cdot \E[\E[\b1_{i\in I}\mid R]
  \max\{0,f_{R\cap [i-1]}(i)\} \mid i\in R].
\end{aligned}
\end{equation*}

On the product space associated with the distribution of $R$ conditioned
on $i\in R$, both of the terms $\E[\b1_{i\in I}\mid R]$ and
$\max\{0,f_{R\cap [i-1]}(i)\}$ are non-increasing functions,
because of the monotonicity of the \crs scheme used to obtain $I$ from
$R$ and $f$ being submodular, respectively. Notice that the randomness
in both terms $\E[\b1_{i\in I} \mid R]$ and $\max\{0,f_{R\cap [i-1]}(i)\}$
stems only from the random set $R$, and not from the potential
randomness of the CR scheme.
Hence, by the FKG inequality
we obtain

\begin{equation*}
\begin{aligned}
\Pr[i\in R]\cdot \E[\E[\b1_{i\in I}\mid R]
  &\max\{0,f_{R\cap [i-1]}(i)\} \mid i\in R]\\
  &\geq \Pr[i\in R] \cdot \E[\b1_{i\in I}\mid i\in R]\cdot
  \E[\max\{0,f_{R\cap [i-1]}(i)\}\mid i\in R]\\
&= \Pr[i\in R] \cdot \Pr[i\in I\mid i\in R] \cdot
   \E[\max\{0,f_{R\cap [i-1]}(i)\}\mid i\in R]\\
&\geq c \Pr[i\in R] \cdot \E[\max\{0,f_{R\cap [i-1]}(i)\}\mid i\in R]\\
&\geq c \Pr[i\in R] \cdot \E[f_{R\cap [i-1]}(i) \mid i\in R]\\
&= c \Pr[i\in R] \cdot \E[f_{R\cap [i-1]}(i)]\\
&=c \E[f(R\cap [i])-f(R\cap [i-1])],
\end{aligned}
\end{equation*}
where in the second to last equality we use again the property that
$f_{R\cap [i-1]}(i)$ is independent of $i\in R$.
Hence, this shows~\eqref{eq:pruneToProve} as desired, and completes the proof.
\end{proof}

The following subsection on strict CR schemes discusses an alternate way
of rounding that does not rely on pruning and is oblivious to the underlying
submodular function. As we highlight below, such a procedure is useful when 
the value of several submodular functions should approximately be preserved,
simultaneously.
However, since we do not rely on this alternate procedure later,
this part can safely be skipped.

\paragraph{Strict contention resolution schemes}

An alternative way to round in the context of non-monotone
submodular functions, that does not rely on pruning, can be obtained by
using a stronger notion of \crs schemes. More precisely, we say that
a $(b,c)$-balanced \crs scheme $\pi$ for $P_\cI$ is \emph{strict},
if it satisfies the second condition of a \crs scheme with equality,
i.e., $\Pr[i\in \pi_\bx(R(\bx))]=c$. We have the following
(the proof can be found in Appendix~\ref{sec:appendixCR}).
\begin{theorem}
\label{thm:submod-balanced-strict}
Let $f:2^N \rightarrow \RR_+$ be a non-negative submodular function
with multilinear relaxation $F$,
and $\bx$ be a point in $P_{\cI}$, a convex relaxation for $\cI
\subseteq 2^N$. Let $\pi$ be a monotone and strict $(b,c)$-balanced \crs scheme
for $P_{\cI}$, and let $I=\pi_{\bx}(R(\bx))$. Then
$$\E[f(I)] \geq c \, F(\bx).$$
\end{theorem}
The advantage of using a strict \crs scheme compared to applying the
pruning step is that this version of rounding is oblivious to the
underlying submodular function $f$. This could potentially be useful
in settings where one is interested in simultaneously maximizing more
than one submodular function. Assume for example that $\bx$ is a point
such that $F_1(\bx)$ and $F_2(\bx)$ have simultaneously high
values, where $F_1$ and $F_2$ are the multilinear relaxations of two
submodular functions $f_1$ and $f_2$.  Then using a rounding that is
oblivious to the underlying submodular function leads to a randomly
rounded set $I$ satisfying $\E[f_1(I)]\geq c F_1(\bx)$ and
$\E[f_2(I)]\geq c F_2(\bx)$. 

Any monotone but not necessarily strict
$(b,c)$-balanced \crs scheme $\pi$ can be transformed into a monotone
$(b,c)$-balanced \crs scheme that is arbitrarily close to being strict
as follows.  For each element $i\in N$, one can estimate the
probability $c'_i = \Pr[i\in \pi_\bx(R(\bx))\mid i\in I]\geq c$ via
Monte-Carlo sampling within a polynomially small error (assuming that
$c$ is a constant).  Then we can modify the \crs scheme by removing
from its output $I$, element $i\in I$ with probability $1 - c / c'_i$.
The resulting scheme is arbitrarily close to being strict and can be
used in place of a strict scheme in Theorem~\ref{thm:submod-balanced-strict}
with a weaker guarantee; in applications to approximation,
the ratio is affected in the lower-order terms. We omit further details.

\medskip
\noindent {\bf Combining \crs schemes:} Next, we discuss how to
combine contention resolution schemes for different constraints.  We
consider a constraint $\cI = \cap_{i=1}^h \cI_i$ and its polyhedral
relaxation $P_{\cI} = \cap_i P_{\cI_i}$, such that $P_{\cI_i}$ has a
monotone $(b,c_i)$-balanced \crs scheme $\pi^i$.  We produce a
contention resolution scheme $\pi$ for $\cI$ which works with respect
to the natural combination of constraint relaxations --- an
intersection of the respective polytopes $P_{\cI_i}$. This ensures
that the relaxed problem is still tractable and we can apply our
optimization framework.

In case some elements $D\subseteq N$ are not
part of the constraint $\cI_i$, we assume without loss of
generality that $\pi^i$ never removes elements in $D$, i.e.,
$\pi^i_\bx(A)\cap D = A\cap D$ for any $\bx\in b P_{\cI_i}$
and $A\subseteq N$.
The combined contention resolution scheme $\pi$ for
$P_{\cI}$ is defined by
\begin{equation*}
\pi_\bx(A) = \bigcap_{i} \pi^i_\bx(A) \quad \text{for } A\subseteq N,
\bx\in b P_\cI.
\end{equation*}

A straightforward union bound
would state that the combined scheme $\pi$ 
is $(b,1-\sum_i (1-c_i))$-balanced for
$P_{\cI}$.  Using the FKG inequality, we obtain a stronger result in this
setting, namely a $(b, \prod_i c_i)$-balanced scheme.
Moreover, if each
constraint admits a $(b,c)$-balanced scheme and each element
participates in at most $k$ constraints, then we obtain a
$(b,c^k)$-balanced scheme.  This is the statement of
Lemma~\ref{lem:compose} which we prove here using the combined scheme
$\pi$ defined above.

\begin{proof}[Proof of Lemma~\ref{lem:compose}]

  Let us consider the $\leq k$ constraints that element $i$ participates
  in.  For simplicity we assume $k=2$; the general statement follows
  by induction. For notational convenience we define $R=R(\bx)$,
  $I_1=\pi^1_\bx(R)$ and $I_2=\pi^2_\bx(R)$.

Conditioned on $R$, the choices of $I_1,I_2$ are
  independent, which means that
$$ \Pr[i \in I_1 \cap I_2 \mid R] = \Pr[i \in I_1 \ \& \ i \in I_2 \mid R]
= \Pr[i \in I_1 \mid R] \Pr[i \in I_2 \mid R]. $$
Taking an expectation over $R$ conditioned on $i \in R$, we get
$$ \Pr[i \in I_1 \cap I_2 \mid i \in R] = \E_R[ \Pr[i \in I_1 \cap I_2 \mid R] \mid i \in R]
 = \E_R[ \Pr[i \in I_1 \mid R] \Pr[i \in I_2 \mid R] \mid i \in R].$$
Both $\Pr[i \in I_1 \mid R]$ and $\Pr[i \in I_2 \mid R]$ are non-increasing
functions of $R$ on the product space of sets containing $i$,
so by the FKG inequality,
$$ \E_R[ \Pr[i \in I_1 \mid R] \Pr[i \in I_2 \mid R] \mid i \in R]
 \geq \E_R[  \Pr[i \in I_1 \mid R] \mid i \in R] \cdot
 \E_R[  \Pr[i \in I_2 \mid R] \mid i \in R].$$
Since these expectations are simply probabilities conditioned on $i \in R$, we conclude:
$$ \Pr[i \in I_1 \cap I_2 \mid i \in R] \geq 
 \Pr[i \in I_1 \mid i \in R] \Pr[i \in I_2 \mid i \in R].$$
Monotonicity of the above scheme is also easily implied: consider
$j \in T_1 \subset T_2 \subseteq N$, then
  \begin{eqnarray*}
\Pr[j \in I | R = T_1] = \prod_i \Pr [j \in I_i | R = T_1]
    \ge \prod_i \Pr [j \in I_i | R = T_2] = \Pr[j \in I | R = T_2].
  \end{eqnarray*}
where the inequality follows from the fact that each of the schemes is monotone.
The polynomial time implementability of the composed scheme follows
easily from the polynomial time implementability of
$\pi^1$ and $\pi^2$.
\end{proof}

\subsection{Obtaining \crs schemes via distributions of deterministic
\crs schemes}\label{subsec:distrCRschemes}

We now describe a general way to obtain \crs schemes relying on an LP
approach.
More precisely, we will observe that any \crs scheme can be interpreted
as a distribution over deterministic \crs schemes.
Exploiting this observation, we formulate an exponential-sized LP 
whose optimal solution corresponds to an optimal \crs scheme. The separation
problem of its dual then gives a natural characterization
for the existence of strong CR schemes, which can be made algorithmic in
some interesting cases including matroid constraints, as we show in
Section~\ref{subsec:CRmatroids}.
Furthermore, in Section~\ref{subsec:corrGap}, we will use this point
of view to draw a connection to a recently introduced concept,
known as correlation gap.

\medskip

Recall the formal definition of \crs schemes given in
Definition~\ref{defn:crscheme}, in particular the differences between
oblivious, deterministic and general (randomized) schemes.  First, we
note that the simplest \crs schemes are the oblivious ones. An
oblivious scheme does not depend on $\bx$ and is deterministic; hence
it is essentially a single mapping $\pi:2^N \rightarrow \cI$ that
given $A \subseteq N$ returns a set $\pi(A)$ such that $\pi(A)
\subseteq A$ and $\pi(A) \in \cI$.  Several alteration based schemes
are oblivious --- see \cite{CCKR01,ChekuriMS03} for some examples. A
typical oblivious scheme fixes an ordering of the elements of $N$
(that depends on the combinatorial properties of $\cI$); it starts
with an empty set $A'$, and considers the elements of $A$ according
to the fixed order and adds the current element $i$ to the set $A'$ if
$A'\cup \{i\} \in \cI$, otherwise it discards $i$. Finally it outputs
$A'$. These greedy ordering based insertion schemes are easily seen to
be monotone. A deterministic \crs scheme is more general than an
oblivious scheme in that the output can depend on $\bx$; in other
words, for each $\bx \in P_\cI$, $\pi_\bx$ is a mapping from $2^N$ to
$\cI$. The advantage or need for such a dependence is demonstrated by
matroid polytopes.  Let $P(\cM)$ be the convex hull 
of the independent sets of a matroid $\cM$; oblivious schemes cannot give a
$c$-balanced \crs scheme for any constant $c$.  However, we can show
that for any $b \in [0,1]$ a good deterministic \crs scheme
exists: for any $\bx \in
P_\cM$, there is an ordering $\sigma_\bx$ that can be efficiently
computed from $\bx$ such that a greedy insertion scheme based on the
ordering $\sigma_\bx$ gives a $(b,1-b)$-balanced scheme. Such a scheme
for $b=1/2$ is implicitly present in \cite{ChawlaHMS09}, however for
completeness, we give the details of our scheme in
Section~\ref{subsec:CRmatroids}. The algorithm in \cite{ChanH09} for
geometric packing problems was reinterpreted as a deterministic \crs
scheme following our work; it is also based on computing an ordering
that depends on $\bx$ followed by a greedy insertion procedure via the
computed ordering (see also more recent work \cite{EneHR12}). Such
ordering based deterministic schemes are easily seen to be
monotone.

In contrast to deterministic schemes, general (randomized) \crs
schemes are such that $\pi_\bx(A)$ is a \emph{random} feasible subset
of $A$. Randomization is necessary to obtain an optimal result even
when considering contention for a single item
\cite{FeigeV06,FeigeV10}. For the time being, we do not require the
\crs schemes to be monotone; this is a point we discuss later.
A non-oblivious $(b,c)$-balanced \crs scheme $\pi$, deterministic or
randomized, can depend on $\bx$, and hence it is convenient to view it
is a collection of separate schemes, one for each $\bx \in b
P_\cI$. They are only tied together by the uniform guarantee $c$. In
the following we will fix a particular $\bx$ and focus on finding the
best scheme $\pi_\bx$ for it. As we already discussed, if $\pi$ is
deterministic, then $\pi_\bx$ is a mapping from $2^N$ to $\cI$.
We observe that a randomized scheme $\pi_\bx$ is 
a distribution over deterministic schemes; note that here we are
ignoring computational issues as well as monotonicity.
We formalize this now.
Call a mapping $\phi$ from $2^N$ to $\cI$ {\em valid} if
$\phi(A)\subseteq A$ \;$\forall A\subseteq N$.  Let $\Phi^*$ be the
family of all valid mappings from $2^N$ to $\cI$.  Any probability
distribution $(\lambda_\phi), \phi \in \Phi^*$ induces a randomized
scheme $\pi_\bx$ as follows. For a set $A$, the algorithm $\pi_\bx$
first picks $\phi \in \Phi^*$ according to the given probability
distribution and then outputs $\phi(A)$. Conversely, for every
randomized scheme $\pi_\bx$, there is an associated probability
distribution $(\lambda_\phi), \phi \in \Phi^*$ \footnote{Let $k$ be an
  upper bound on the number of random bits used by $\pi_\bx$.  For any
  fixed string $r$ of $k$ random bits, let $\phi^r$ be the valid
  mapping from $2^N$ to $\cI$ generated by the algorithm $\pi_\bx$
  with random bits set to $r$. The distribution where for each $r$ the
  probability assigned to $\phi^r$ is $1/2^k$ is the desired one.}.
Based on the preceding observation, one can write an LP to
express the problem of finding
a \crs scheme that is $(b,c)$-balanced for $\bx$ with a value of $c$
as high as possible. 
More precisely, for each $\phi \in \Phi^*$, we
define $q_{i,\phi} = \Pr[i \in \phi(R)]$, where, as usual, $R:=R(\bx)$
is obtained by including each $j\in N$ in $R$ with probability
$\mathrm{x}_j$, independently of the other elements.
Thus, for a given distribution $(\lambda_\phi)_{\phi\in \Phi^*}$,
the probability that the corresponding \crs scheme $\pi_\bx$
returns a set $\pi_\bx(R)$ containing $i$, is given
by $\sum_{\phi \in \Phi^*}q_{i,\phi}\lambda_\phi$.
Hence, the problem of finding the
distribution $(\lambda_\phi)_{\pi\in \Phi^*}$
that leads to a $(b,c)$-balanced CR scheme for $\bx$ with $c$
as high as possible can be formulated as the following linear
program (LP1),
with corresponding dual (DP1).
\begin{equation*}
\begin{array}{rrrcll}
\multirow{4}{*}{(LP1)} & \max & c \hfill &&&\\
& s.t. & \sum_{\phi \in \Phi^*} q_{i,\phi} \lambda_\phi  &\geq
    & \mathrm{x}_i c &\forall i\in N\\
& &\sum_{\phi \in \Phi^*} \lambda_{\phi} & = &1 &\\
& & \lambda_\phi &\geq &0 &\forall \phi \in \Phi^* \\[0.5em]
\multirow{4}{*}{(DP1)}& \min & \mu \hfill &&&\\
& \text{s.t.} & \sum_{i\in N} q_{i,\phi} y_i & \leq & \mu
   &\forall \phi \in \Phi^*\\
& &\sum_{i\in N} \mathrm{x}_i y_i & = &1 &\\
& & y_i &\geq &0 &\forall i\in N
\end{array}
\end{equation*}

In general we may also be interested in
a restricted set of mappings $\Phi \subseteq \Phi^*$.
In the above LP we can replace $\Phi^*$ by $\Phi$ to obtain
the best $c$ that can be achieved by taking probability distributions
over valid mappings in $\Phi$. Let $c(\bx,\Phi)$ be the optimum
value of the LP for a given $\bx$ and a set $\Phi \subseteq \Phi^*$.
It is easy to see that $c(\bx,\Phi) \le c(\bx, \Phi^*)$ for any
$\Phi$. From the earlier discussion, $c(\bx,\Phi^*)$ is the best
scheme for $\bx$. We summarize the discussion so far by the following.

\begin{proposition}
There exists a $(b,c)$-balanced \crs scheme
for $P_\cI$ iff \, $\inf_{\bx \in bP_\cI} c(\bx,\Phi^*) \ge c$.
\end{proposition}

\mypara{Proving the existence of a $(b,c)$-balanced \crs scheme:} To show that
$P_\cI$ has a $(b,c)$-balanced \crs scheme we need to show that $c(\bx,\Phi^*)
\ge c$ for all $\bx \in bP_\cI$.  By LP duality this is equivalent to
showing that the optimum value of the dual (DP1) is at least $c$ for
all $\bx$. We first reformulate the dual in a convenient form so that
proving a lower bound $c$ on the dual optimum reduces to a more
intuitive question. We will then address the issue of efficiently
constructing a \crs scheme that nearly matches the lower bound.

Below we will use $R$ to denote a random set obtained by picking
each $i \in N$ independently with probability $\mathrm{x}_i$ and use
probabilities and expectations with respect to this random process.
The optimum value of the dual can  be rewritten as:
$$\min_{\by \ge 0} \max_{\phi \in \Phi^*} \frac{\sum_i q_{i,\phi} y_i}{\sum_{i \in N} \mathrm{x}_i y_i} = \min_{\by \ge 0} \max_{\phi \in \Phi^*} \frac{\sum_i y_i \Pr[i \in \phi(R)]}{\sum_{i \in N} \mathrm{x}_i y_i} = \min_{\by \ge 0} \max_{\phi \in \Phi^*} \frac{\E_R\left[ \sum_{i \in \phi(R)} y_i\right]}{\sum_{i \in N} \mathrm{x}_i y_i}
$$
For any fixed weight vector $\by \ge 0$ we claim that
$$\max_{\phi \in \Phi^*} \E_R\left[ \sum_{i \in \phi(R)} y_i\right]
= \E_R\left[ \max_{S \subseteq R, S \in \cI} \sum_{i \in S} y_i\right],
$$
which follows by considering the specific mapping $\phi \in \Phi^*$ that
for each $A \subseteq N$ sets $\phi(A) = \max_{A' \subseteq A, A' \in \cI} y(A)$.
Thus, the dual optimum value is 
\begin{equation}
  \min_{\by \ge 0} \frac{\E_R\left[ \max_{S \subseteq R, S \in \cI} \sum_{i \in S} y_i\right]}{{\sum_{i \in N} \mathrm{x}_i y_i}}.
\label{eq:optDP1}
\end{equation}

The above expression can be explained as an ``integrality gap'' of
$P_\cI$ for a specific rounding strategy; here the problem of interest
is to find a maximum weight independent set in $\cI$.  The vector
$\by$ corresponds to weights on $N$. The vector $\bx$ corresponds to a
fractional solution in $bP_\cI$ (it is helpful here to think of
$b=1$). Thus $\sum_{i \in N} \mathrm{x}_i y_i$ is the value of the
fractional solution. The numerator is the expected value of a maximum
weight independent set in $R$.  Since we are minimizing over $\by$,
the ratio is the worst case gap between the value of an integral
feasible solution (obtained via a specific rounding) and a fractional
solution.

Thus, to prove the existence of a $(b,c)$-balanced \crs scheme it is
sufficient (and necessary) to prove that for all $\by\ge 0$ and $\bx
\in bP_\cI$
$$
 \frac{\E_{R(\bx)}\left[ \max_{S \subseteq R, S \in \cI} \sum_{i \in S} y_i\right]}{{\sum_{i \in N} \mathrm{x}_i y_i}} \ge c.
$$

\mypara{Constructing \crs schemes via the ellipsoid algorithm:} We
now discuss how to efficiently compute the best \crs scheme for a
given $\bx$ by solving (LP1) via the dual (DP1). We observe that
$c(\bx,\Phi^*)$, the best bound for a given $\bx$, could be smaller
than the bound $c$. It should not be surprising that the separation
oracle for the dual (DP1) is related to the preceding
characterization. The separation oracle for (DP1) is the following:
given $\mu$ and weight vector $\by$, normalized such that $\sum_i
\mathrm{x}_iy_i = 1$, check whether there is any $\phi \in \Phi^*$
such that $\sum_{i\in N} q_{i,\phi}y_i > \mu$ and if so output a
separating hyperplane.  To see whether there is a violated constraint,
it suffices to evaluate $\max_{\phi \in \Phi^*} q_{i,\phi}y_i$ and
compare it with $\mu$.  Following the previous discussion, this
expression is equal to $\E_{R(\bx)}\left[ \max_{S \subseteq R, S \in
    \cI} \sum_{i \in S} y_i\right]$.  One can accurately estimate this
quantity as follows. First, we sample a random set $R$ using marginals
given by $\bx$. Then we find a maximum $\by$-weight subset of $R$ that
is contained in $\cI$. This gives an unbiased estimator, and to get a
high-accuracy estimate we repeat the process sufficiently many times
and take the average value. Thus, the algorithmic problem needed for
the separation oracle is the maximum weight independent set problem
for $\cI$: given weights $\by$ on $N$ and a $A \subseteq N$ output a
maximum weight subset of $A$ in $\cI$.  The sampling creates an
additive error $\epsilon$ in estimating $\E_{R(\bx)}\left[ \max_{S
    \subseteq R, S \in \cI} \sum_{i \in S} y_i\right]$ which results
in a corresponding loss in finding the optimum solution value $\mu^*$
to (DP1).  To implement the ellipsoid algorithm we also need to find a
separating hyperplane if there is a violated constraint.  A natural
strategy would be to output the hyperplane corresponding to the
violating constraint found while evaluating $\max_{\phi \in \Phi^*}
q_{i,\phi}y_i$. However, we do not necessarily have the exact
coefficients $q_{i,\phi}$ for the constraint since we use random
sampling. We describe in Section~\ref{sec:appendix-ellipsoid} of the
appendix the technical details in implementing the ellipsoid algorithm
with sufficiently accurate estimates obtained from sampling.  For now
assume we can find a separating hyperplane corresponding to the most
violated constraint. The ellipsoid algorithm can then be used to find
a polynomial number of dual constraints that certify that the dual
optimum is at least $\mu^*-\eps$ where $\mu^*$ is the actual dual
optimum value. By strong duality $\mu^* = c(\bx,\Phi^*)$.  We then
solve the primal (LP1) by restricting it to the variables that
correspond to the dual constraints found by the ellipsoid algorithm.
This gives a primal feasible solution of value $c(\bx,\Phi^*)-\eps$
and this solution is the desired \crs scheme.  We observe that the
primal can be solved efficiently since the number of variables and
constraints is polynomial; here too we do not have the precise
coefficients $q_{i,\phi}$ but we can use the esimates that come from
the dual --- see Section~\ref{sec:appendix-ellipsoid}. To summarize,
an algorithm for finding a maximum weight independent set in $\cI$,
together with sampling and the ellipsoid algorithm, can be used to
efficiently find a $(b,c(\bx,\Phi^*) -\eps)$-balanced \crs scheme
where $\eps$ is an error tolerance; the running time depends
polynomially on the input size and $1/\eps$. The proof can be easily
adapted to show that an $\alpha$-approximation for the max-weight
independent set problem gives a $\alpha \cdot c(\bx,\Phi^*) -\eps$
$\crs$ scheme.

\mypara{Monotonicity:}
The discussion so far did not consider the issue of monotonicity.
One way to adapt the above approach to monotone schemes is to define
$\Phi$ to be the family of all deterministic monotone \crs schemes and
solve (LP1) restricted to $\Phi$. A deterministic scheme $\phi$ is
monotone if it has the property that $i \in \phi(A)$ implies that $i
\in \phi(A')$ for all $A' \subset A$. Distributions of deterministic
monotone schemes certainly yield a monotone \crs
scheme. Interestingly, it is not true that all monotone randomized
\crs schemes can be obtained as distributions of deterministic ones.
Now the question is whether we can solve (LP1) restricted to monotone
deterministic schemes. In general this is a non-trivial problem.
However, the ellipsoid-based algorithm to compute $c(\bx,\Phi^*)$ that
we described above gives the following important property. In each
iteration of the ellipsoid algorithm, the separation oracle uses a
maximum-weight independent set algorithm for $\cI$ to find a
violating constraint; this constraint corresponds to a deterministic
scheme $\phi$ that is obtained by specializing the algorithm to the
given weight vector $\by$. Therefore, if the
maximum-weight independent set algorithm is monotone, then
all the constraints generated in the ellipsoid algorithm correspond to
monotone schemes. Since we solve the primal (LP1) only for the
schemes generated by the separation oracle for the dual
(DP1), it follows that there is an optimum solution to (LP1) that is a
distribution over monotone schemes! In such a case $c(\bx,\Phi^*) =
c(\bx,\Phi)$ and there is no loss in using monotone schemes.  
For matroids the greedy algorithm to find a maximum weight
independent set is a monotone algorithm.  Thus, for matroids, the
above approach of solving (DP1) and (LP1) can be used to obtain a
close to optimal monotone $(b,c)$-balanced \crs scheme. It
remains to determine the value of the optimal $c$ and we analyze it in
Section~\ref{subsec:CRmatroids}.  It may be the case that there is no
monotone maximum weight independent set algorithm for some given
$\cI$, say the intersection of two matroids. In that case we can
use an approximate montone algorithm instead.

We summarize the above discussions in the following theorem.

\begin{theorem}\label{thm:linkCRschemeLP}
  There is a $(b,c)$-balanced \crs scheme for $P_\cI$ iff
  $\E_{R(\bx)}\left[ \max_{S \subseteq R, S \in \cI} \sum_{i \in S}
    y_i\right] \geq c \sum_i y_i \mathrm{x}_i$ for all $\bx \in
  bP_\cI$ and $\by \ge 0$. Moreover, if there is a polynomial-time
  deterministic algorithm to find a maximum weight independent set in
  $\cI$, then for any $b$ and $\eps > 0$, there is a randomized
  efficiently implementable $(b,c^*-\eps)$-balanced \crs scheme for
  $P_\cI$ where $c^*$ is the smallest value of $c$ such that there is
  a $(b,c)$-balanced \crs scheme for $P_\cI$; the running time is
  polynomial in the input size and $1/\eps$.  In addition, if the
  maximum-weight independent set algorithm is monotone, the resulting
  \crs scheme is monotone.  
\end{theorem}

Before leveraging the above theorem to design close to optimal
CR schemes for matroids, we highlight an interesting connection
between CR schemes and a concept known as
\emph{correlation gap}. This connection is a further insight
that we gain through the linear programs (LP1) and (DP1).

\subsection{Connection to correlation gap}\label{subsec:corrGap}

In this section we highlight a close connection between \crs schemes
and a concept known as {\em correlation gap}~\cite{ADSY10}.
The correlation gap is a function on set functions that measures 
how much the expected value of a set function with respect to some random
input can vary, if only the marginal probabilities
of the input are fixed.
We first show how one can naturally extend this notion to
sets $\mathcal{I}\subseteq 2^N$.
Then, by exploiting the dual
LP formulation of \crs schemes (DP1), we present a close relationship of
the notion of correlation gap, interpreted in terms of constraints, and
the existence of strong \crs schemes.

\medskip

\begin{definition}
\label{def:correlation-gap}
For a set function $f:2^N \rightarrow \RR_+$, the \emph{correlation
  gap} is defined as
\begin{equation*}
\kappa(f) = \inf_{\bx \in [0,1]^N} \frac{\E[f(R(\bx))]}{f^+(\bx)},
\end{equation*}
where $R(\bx)$ is a random set independently containing each
element $i$ with probability $x_i$,
and
\begin{equation*}
f^+(\bx) = \max \{ \sum_S \alpha_S f(S): \sum_S \alpha_S \b1_S
  = \bx, \sum_S \alpha_S = 1, \alpha_S \geq 0\}
\end{equation*}
is the maximum possible expectation of $f$ over distributions
with expectation $\bx$.
Furthermore, for a class of functions $\cal C$, the correlation
gap is defined by $\kappa({\cal C}) = \inf_{f \in {\cal C}} \kappa(f)$.
\end{definition}
In other words, the correlation gap is the worst-case ratio between
the multilinear extension $F(\bx) = \E[f(R(\bx))]$ and the concave
closure $f^+(\bx)$.
We remark that we define the correlation gap as a
number $\kappa \in [0,1]$, to be in line
with the parameter $c$ in our notion of a $(b,c)$-balanced \crs scheme
(the higher the better). The definition in~\cite{ADSY10} uses the inverse ratio.

The relationship between \crs schemes and correlation gap
arises as follows.

\begin{definition}
\label{def:constraint-gap}
For $\cI \subseteq 2^N$, we define the correlation gap as
$ \kappa(\cI) = \inf_{\bx \in P_\cI, \by \geq 0}
 \frac{1}{\sum_i \mathrm{x}_i y_i} \E[\max_{S \subseteq R, S \in \cI}
  \sum_{i \in S} y_i],$
where $R=R(\bx)$ contains element $i$ independently
with probability $\mathrm{x}_i$.
\end{definition}

The reason we call this quantity a correlation gap (considering Definition~\ref{def:correlation-gap}),
is that this quantity is equal to the correlation gap of the {\em weighted rank function} corresponding to $\cI$
(see Lemma~\ref{lem:corrGapRankFunc} below).

\begin{theorem}\label{thm:corrGapToC}
The correlation gap of $\cI$ is equal to
the maximum $c$ such that $\cI$ admits a
$c$-balanced \crs scheme.
\end{theorem}
\begin{proof}
The correlation gap of $\cI$ is equal
to the optimum value of (DP1).
By LP duality, this is equal to the optimum of the primal (LP1),
which is the best value of $c$
for which there is a $c$-balanced \crs scheme.
\end{proof}

The following lemma shows a close connection between the
correlation gap of a solution set $\cI$ and the correlation
gap of the respective rank function.
More precisely, the correlation gap of $\cI$ corresponds
to the worst (i.e. smallest) correlation gap of the respective rank function
over all weight vectors.

\begin{lemma}\label{lem:corrGapRankFunc}
For $\cI \subseteq 2^N$ and weight vector $\by \geq 0$,
let $r_\by(R) = \max_{S \subseteq R, S \in \cI} \sum_{i \in S} y_i$
denote the associated weighted rank function. Then
$\kappa(\cI)=\inf_{\by\geq 0}\kappa(r_\by)$.
\end{lemma}
\begin{proof}
Using the notation $r_{\by}(R)$ for the weighted rank function
with weights $\by$, the correlation gap of $\cI$ can be
rewritten as
$ \kappa(\cI) = \inf_{\bx \in P_\cI, \by \geq 0}
\frac{\E[r_\by(R(\bx))]}{\sum_i \mathrm{x}_i y_i}, $
 where $R(\bx)$ contains elements independently with
probabilities $\mathrm{x}_i$.
We first observe that for any $\bx\in P_\cI$, we have
$r_\by^+(\bx)=\sum_{i}\mathrm{x}_i y_i$. 
Hence, let $\bx\in P_\cI$, and
consider a convex combination 
$\bx = \sum_{S \in \cI} \alpha_S \b1_S$,
$\sum \alpha_S = 1$, $\alpha_S \geq 0$
with $r_\by^+(\bx)=\sum_{S\in \cI}\alpha_S y(S)$.
Since the weighted rank function of a feasible set $S\in \cI$ is
simply its weight we obtain
$$ r_\by^+(\bx) = \sum_{S \in \cI} \alpha_S y(S)
 = \by \cdot \sum_{S \in \cI} \alpha_S \b1_S
 = \by \cdot \bx = \sum_i \mathrm{x}_i y_i ,$$
as claimed. Therefore,
\begin{equation*}
\kappa(\cI)  = \displaystyle \inf_{\bx \in P_\cI, \by \geq 0}
\frac{\E[r_\by(R(\bx))]}{\sum_i \mathrm{x}_i y_i}
= \inf_{\bx \in P_\cI, \by \geq 0} \frac{\E[r_\by(R(\bx))]}{r_\by^+(\bx)}.
\end{equation*}
To prove the claim it remains to show that
\begin{equation}\label{eq:restrictP}
\inf_{\bx \in P_\cI, \by \geq 0} \frac{\E[r_\by(R(\bx))]}{r_\by^+(\bx)}=
\inf_{\bx \in [0,1]^N, \by\geq 0}\frac{\E[r_\by(R(\bx))]}{r_\by^+(\bx)}.
\end{equation}
Let $\by\geq 0$. 
We will prove~\eqref{eq:restrictP} by showing that for any point
$\bx\in [0,1]^N$ there is a point $\bx'\in P_\cI$
with $\bx'\leq \bx$ (coordinate-wise), and
satisfying $r^+_\by(\bx')\geq r^+_\by(\bx)$.
Since $r_\by$ is monotone, we then obtain
$\E[r_\by(R(\bx))]/r^+_\by(\bx)\geq \E[r_\by(R(\bx))]/r^+_\by(\bx')$, showing that
the infinum over $\bx$ on the right-hand side of~\eqref{eq:restrictP} can indeed be
restricted to the polytope $P_\cI$.
Let $\bx=\sum_{S\subseteq N}\alpha_S \b1_S$,
$\sum_{S\subseteq N}\alpha_S=1, \alpha_S\geq 0$ be a convex combination
of $\bx$ such that $r^+_\by(\bx)=\sum_{S \subseteq N}\alpha_S r_\by(S)$.
For every $S\subseteq N$, let $I(S)\subseteq S$ be a maximum weight independent
set, hence $r_\by(S)=y(I(S))$. The point
$\bx'=\sum_{S\subseteq N}\alpha_S \b1_{I(S)}$ clearly
satisfies $\bx'\leq \bx$, and furthermore
$$r^+_\by(\bx')\geq \sum_{S\in \cI}
\Bigg(\sum_{\substack{W\subseteq N, I(W)=S}} \alpha_W\Bigg)
r_\by(S) = \sum_{S\subseteq N} \alpha_S r_\by(S) =r^+_\by(\bx).
$$
\end{proof}

\subsection{Contention resolution for matroids}\label{subsec:CRmatroids}

In this section we prove the following theorem on \crs schemes for 
matroids. 
\begin{theorem}\label{thm:existenceCRSmatroid}
  For any matroid $\cM=(N,\cI)$ on $n$ elements 
  there exists a $\Big(b, \frac{1 - (1-\frac{b}{n})^n }{b}
  \Big)$-balanced \crs scheme for the polytope $P(\cM)$.
\end{theorem}
We later address monotonicity of the scheme and constructive aspects.
To prove Theorem~\ref{thm:existenceCRSmatroid} we rely on the
characterization formalized in Theorem~\ref{thm:linkCRschemeLP}.  It
suffices to prove for $\bx\in b\cdot P_{\mathcal{I}}$ and any
non-negative weight vector $\by \ge 0$ that
$\E_R\left[\max_{S\subseteq R, S \in \cI}\sum_{i \in S}y_i\right]\geq
c\sum_{i\in N} x_i y_i$, with $c=\frac{1-(1-\frac{b}{n})^n}{b}$ where
$R$ contains each $i\in N$ independently with probability $x_i$ and
$\bx \in b\cdot P_{\mathcal{I}}$.
For
a given weight vector $\by \ge 0$ on $N$ and a set $S \subseteq N$ let
$r_{\by}(S)$ denote the weight of a maximum weight independent set
contained in $S$; in other words $r_{\by}$ is the weighted rank
function of the matroid $\cM$.  Restating, it remains to prove
\begin{equation}\label{eq:matToProve}
\E[r_{\by}(R)] \geq \frac{1-(1-\frac{b}{n})^n}{b}\sum_{i \in N} y_i x_i,
\end{equation}

It is well-known that a simple greedy
algorithm can be used to compute $r_{\by}(S)$ (in fact an independent
set $S' \subseteq S$ of maximum weight with respect to $y_i$):
Start with $S'=\emptyset$, consider the elements of $S$
in non-increasing order of their weight $y_i$ and add the current
element $i$ to $S'$ if $S'+i$ is independent, otherwise discard $i$.

To show~\eqref{eq:matToProve}, which is a general property of
weighted matroid rank functions, we prove a more general result
that holds for any non-negative monotone submodular function.
The main ingredient for this is a lower bound on the multilinear
extension, which is stated in Lemma~\ref{lem:correl-gap-stronger}.
A slightly weaker form of Lemma~\ref{lem:correl-gap-stronger}, which
we state as Lemma~\ref{lem:correl-gap-weak}, will be presented first,
due to its consice proof. The proof of
Lemma~\ref{lem:correl-gap-stronger} is deferred to the appendix.
Both lemmas can be seen as an extension of the property that 
the correlation gap for monotone submodular functions is $1-1/e$~\cite{CCPV07}.

\begin{lemma}
\label{lem:correl-gap-weak}
If $f:2^N \rightarrow \RR_+$ is a monotone submodular function, $F:[0,1]^N \rightarrow \RR_+$
its multilinear extension, and $f^+:[0,1]^N \rightarrow \RR_+$ its concave closure, then
for any $b \in [0,1]$ and $\bp\in [0,1]^N$,
$$ F(b\cdot \bp) \geq (1 - e^{-b}) f^+(\bp).$$
\end{lemma}

\begin{proof}
We use another extension of a monotone submodular function, defined in \cite{CCPV07}:
$$ f^*(\bp) = \min_S \left( f(S) + \sum_i p_i f_S(i) \right).$$
It is shown in \cite{CCPV07} that $f^*(\bp) \geq f^+(\bp)$ for all $\bp \in [0,1]^N$.
Consider the function $\phi(t) = F(t \bp)$ for $t \in [0,1]$, i.e. the multilinear extension
on the line segment between $\bf 0$ and $\bp$. We prove that $\phi(t)$ satisfies a differential
equation similar to the analysis of the continuous greedy algorithm \cite{CCPV09},
which leads immediately to the statement of the lemma. We have
$$ \frac{d\phi}{dt} = \bp \cdot \nabla F(t \bp) = \sum_i p_i \partdiff{F}{x_i}\Big|_{\bx=t\bp}.$$
By properties of the multilinear extension, we have $\partdiff{F}{x_i}\Big|_{\bx=t\bp} = \E[f(R+i)-f(R-i)]
 \geq \E[f_R(i)]$, where $R$ is a random set sampled independently with probabilities $x_i = t p_i$
(see \cite{CCPV09} for more details). Therefore,
$$ \frac{d\phi}{dt} = \sum_i p_i \partdiff{F}{x_i} \Big|_{\bx = t \bp} \geq \sum_i p_i \E[f_R(i)]
 = \E[\sum_i p_i f_R(i)] \geq \E[f^*(\bp) - f(R)] $$
by the definition of $f^*(\bp)$. Finally, $\E[f(R)] = F(t \bp) = \phi(t)$,
hence we obtain the following differential inequality:
$$ \frac{d \phi}{dt} \geq f^*(\bp) - \phi(t) $$
under the initial condition $\phi(0) \geq 0$. We solve this as follows:
$ \frac{d}{dt} (e^t \phi(t)) = e^t \phi(t) + e^t \frac{d\phi}{dt} \geq e^t f^*(\bp) $
which implies that
$$ e^b \phi(b) \geq e^0 \phi(0) + \int_0^b e^t f^*(\bp) dt \geq (e^b - 1) f^*(\bp).$$
Considering that $\phi(b) = F(b \bp)$ and $f^*(\bp) \geq f^+(\bp)$, this proves the lemma.
\end{proof}

A more fine-grained analysis leads to the following strengthened
version of Lemma~\ref{lem:correl-gap-weak}, whose proof 
can be found in Appendix~\ref{sec:appendixCR}.

\begin{lemma}\label{lem:correl-gap-stronger}
If $f:2^N \rightarrow \RR_+$ is a monotone submodular
function, $F:[0,1]^N \rightarrow \RR_+$
its multilinear extension, and $f^+:[0,1]^N \rightarrow \RR_+$ its
concave closure, then for any $b \in [0,1]$ and $\bp\in [0,1]^N$,
\begin{equation*}
F(b\cdot \bp) \geq \Big(1-\Big(1-\frac{b}{n}\Big)^n\Big) f^+(\bp).
\end{equation*}
\end{lemma}

Lemma~\ref{lem:correl-gap-stronger}  implies~\eqref{eq:matToProve},
and therefore completes the proof of Theorem~\ref{thm:existenceCRSmatroid}, by
setting $f = r_\by$ and $b\cdot \bp = \bx$.
Notice that the multilinear extension of $r_{\by}$ evaluated at $\bx$ is
$\E[r_{\by}(R)]$. Furthermore, $r_{\by}(\bp)= \sum_{i\in N}y_i p_i = \sum_{i \in N} y_i \frac{x_i}{b}$
if $\bp$ is in the matroid polytope. Hence we
obtain~\eqref{eq:matToProve}:
\begin{align*}
\E[r_{\by}(R)] &\geq \left(1-\left(1-\frac{b}{n}\right)^n \right)
r_{\by}^+\left(\bp \right) 
 = \left(1-\left(1-\frac{b}{n}\right)^n \right)
\sum_{i \in N} y_i \frac{x_i}{b}\\
 &= \frac{1-\left(1-\frac{b}{n}\right)^n}{b} \sum_{i \in N} y_i x_i.
\end{align*}

\medskip Theorem~\ref{thm:linkCRschemeLP} also shows that an
efficient algorithm for computing $r_{\by}$ results in an efficiently
implementable near-optimal \crs scheme. It is well-known that a simple
greedy algorithm can be used to compute $r_{\by}(S)$ (in fact an
independent set $S' \subseteq S$ of maximum weight): Start with
$S'=\emptyset$, consider the elements of $S$ in non-increasing
order of their weight and add the current element $i$ to $S'$ if
$S'+i$ is independent, otherwise discard $i$. Moreover, it is easy
to see that this algorithm is monotone --- the ordering of the
elements by weight does not depend on the set $S$ and
hence if an element $i$ is included when evaluating $r_{\by}(A)$
then it will be included in evaluating $r_{\by}(B)$ for any $B \subset A$.
We thus obtain our main result for \crs schemes in the context of matroids
by combining Theorem~\ref{thm:linkCRschemeLP} 
for a choice of $\epsilon$ satisfying $\epsilon \leq \frac{b}{10n}$
with Theorem~\ref{thm:existenceCRSmatroid}, and by using
the inequality
$(1-\frac{b}{n})^n \leq e^{-b} - \frac{b^2}{10n}$~\footnote{%
This inequality can be obtained by observing that
$1-x+\frac{x^2}{3} \leq e^{-x}$ for $x\in [0,1]$, and
hence $(1-\frac{b}{n})^n \leq (e^{-\frac{b}{n}}-\frac{b^2}{3n^2})^n$.
Let $y= e^{-\frac{b}{n}}$ and $z=\frac{b^2}{3n^2}$ for simplicity.
One can easily check that for these values of $y$ and $z$ we have
$(y-z)^n \leq y^n - n y^{n-1} z + \frac{n^2}{2}y^{n-2} z^2$.
Expanding the last expression and using $n\geq 2$, since the
inequality is trivially true for $n=1$, the desired inequality
follows.
}.

\begin{corollary}\label{corr:mainMatroidRes}
For any matroid $\cM$, and $\bx\in b\cdot P_\cI$, there is an efficiently
 implementable $\big(b,\frac{1-e^{-b}}{b}\big)$-balanced
and monotone \crs scheme.
\end{corollary}

As shown by the following theorem, the \crs schemes that can be
obtained according to Corollary~\ref{corr:mainMatroidRes} are, up to
an additive $\eps$, asymptotically optimal.
\begin{theorem}\label{thm:upperBoundCRSs}
For any $b\in (0,1]$, there is no
$(b,c)$-balanced CR scheme for uniform matroids of rank one
on $n$ elements with $c > \frac{1-(1-\frac{b}{n})^n}{b}$.
\end{theorem}

\begin{proof}
Let $\cM=(N,\cI)$ be the uniform matroid of rank $1$ over $n=|N|$ elements,
and consider the point $\bx\in b\cdot P_{\cI}$ given by $x_i=b/n$ for $i\in N$.
Let $R$ be a random set containing each element $i\in N$ independently
with probability $x_i$. The expected rank of $R$ is given by
\begin{equation}\label{eq:expRank}
\E[r(R)] = 1-\Pr[R=\emptyset] = 1-\left(1-\frac{b}{n}\right)^n.
\end{equation}
Moreover, any $(b,c)$-balanced \crs scheme
returning a set $I\in \cI$ satisfies
\begin{equation}\label{eq:expISize}
\E[|I|] = \sum_{i\in N}\Pr[i\in I] \geq \sum_{i\in N} \frac{b c}{n}=b c.
\end{equation}
Since $I$ is an independent subset of $R$ we have
$\E[r(R)]\geq \E[|I|]$, and the claim follows
by~\eqref{eq:expRank} and~\eqref{eq:expISize}.
\end{proof}

\mypara{A simple $(b,1-b)$-balanced \crs scheme:} Here we describe a
sub-optimal $(b,1-b)$-balanced \crs scheme for matroid polytopes. Its
advantage is that it is {\em deterministic}, simpler and computationally
less expensive than the optimal scheme that requires
solving a linear program. Moreover, Lemma~\ref{lem:span-prob} that is
at the heart of the scheme, is of independent interest and may find
other applications. A similar lemma was independently shown in
\cite{ChawlaHMS09} (prior to our work but in a different context).
Let $\cM=(N,\cI)$ be a matroid. For $S\subseteq N$ recall
that $r(S)$ is the rank of $S$ in $\cM$. The span of a set
$S$ denoted by $\ourspan(S)$ is the set of all elements
$i \in N$ such that $r(S+i)= r(S)$. 

\begin{lemma}
  \label{lem:span-prob}
  If $\cM=(N,\cI)$ is a matroid, $\bx \in P(\cM)$, $b \in [0,1]$ and $R$ a
  random set such that $\Pr[i \in R] = b x_i$, then there is an
  element $i_0$ such that $\Pr[i_0 \in \text{span}(R)] \leq b$.
\end{lemma}

\begin{proof}
  Let $r(S) = \max \{|I|: I \subseteq S \ \& \ I \in \cI\}$ denote
  the rank function of matroid $\cM = (N,\cI)$.  Since $\bx \in P(\cM)$,
  it satisfies the rank constraints $x(S) \leq r(S)$.  For $S =
  \ourspan(R)$, we get
  $$ x(\ourspan(R)) \leq r(\ourspan(R)) = r(R) \leq |R|. $$
  Recall that $R$ is a random set where $\Pr[i \in R] = bx_i$. We take the expectation
  on both sides: 
$$\E[x(\ourspan(R))] = \sum_i x_i \Pr[i \in \ourspan(R)], \text{ and } \E[|R|] = \sum_i \Pr[i \in R]   = b \sum_i x_i.$$
  Therefore,
  $$ \sum_{i \in N} x_i \Pr[i \in \ourspan(R)] \leq b \sum_{i \in N} x_i.$$
  This implies that there must be an element $i_0$ such that $\Pr[i_0 \in \ourspan(R)] \leq b$.
\end{proof}

We remark that the inequality $\sum_i x_i \Pr[i \in \ourspan(R)]
\leq \E[|R|]$ has an interesting interpretation: If $\bx \in P(\cM)$, we
sample $R$ with probabilities $x_i$, then let $S = \ourspan(R)$ and
sample again $S' \subseteq S$ with probabilities $x_i$, then $\E[|S'|]
\leq \E[|R|]$.  We do not use this in the following, though.

\begin{theorem}
  \label{thm:deterministic-matroid-contention}
  For any matroid $\cM$ and any $b \in [0,1]$, there is a deterministic $(b,
  1-b)$-balanced \crs scheme.
\end{theorem}

\begin{proof}
  Let $\bx \in P(\cM)$ and sample $R$ with probabilities $b x_i$.  We
  define an ordering of elements as follows. By
  Lemma~\ref{lem:span-prob}, there is an element $i_0$ such that
  $\Pr[i \in \ourspan(R)] \leq b$.  We place $i_0$ at the end of the
  order. Then, since $\bx$ restricted to $N \setminus \{i_0\}$ is in the
  matroid polytope of $\cM \setminus \{i_0\}$, we can recursively find
  an ordering by the same rule. If the elements are labeled
  $1,2,\ldots,|N|$ in this order, we obtain that $ \Pr[i \in
  \ourspan(R \cap [i])] \leq b $ for every $i$.  In fact, we are
  interested in the event that $i$ is in the span of the preceding
  elements, $R \cap [i-1]$. This is a subset of $R \cap [i]$, and
  hence
$$ \Pr[i \in \ourspan(R \cap [i-1])] \leq \Pr[i \in \ourspan(R \cap [i])] \leq b. $$

The \crs scheme is as follows:
\begin{itemize}
\item Sample $R$ with probabilities $b x_i$.
\item For each element $i$, if $i \in R \setminus \ourspan(R \cap [i-1])$,
then include it in $I$.
\end{itemize}

Obviously, $r(I \cap [i]) = r(I \cap [i-1]) + 1$ whenever $i \in I$,
so $r(I) = |I|$ and $I$ is an independent set.

To bound the probability of appearance of $i$, observe that the
appearance of elements in $[i-1]$ is independent of the appearance of
$i$ itself, and hence the events $i \in R$ and $i \notin \ourspan(R
\cap [i-1])$ are independent.  As we argued, $ \Pr[i \in \ourspan(R
\cap [i-1])] \leq b$.  We conclude:
$$ \Pr[i \in I \mid i \in R] = \Pr[i \notin \ourspan(R \cap [i-1])] \geq 1-b.$$
\end{proof}

To implement the scheme we need to make Lemma~\ref{lem:span-prob}
algorithmic. We can accomplish it by random sampling. Fix an element $i$.
Pick a random set $R$ and check if $i \in \ourspan(R)$; repeat sufficiently
many times to obtain an accurate estimate of $\Pr[i \in \ourspan(R)]$.
We note that although the scheme itself is deterministic once
we find an ordering of the elements, the construction of the ordering
is randomized due to the estimation of $\Pr[i\in \ourspan(R)]$ via
sampling.

\subsection{Contention resolution for knapsacks}

Here we sketch a contention resolution scheme for knapsack
constraints.  This essentially follows from known techniques; we
remark that Kulik, Shachnai and Tamir~\cite{KulikST10,KulikST13}
showed how to round a fractional solution to the problem $\max \{
F(\bx): \bx \in P \}$ for any constant number of knapsack constraints
and any non-negative submodular function, while losing a $(1-\eps)$
factor for an arbitrarily small $\eps>0$.  Our goal is to show that
these techniques can be implemented in a black-box fashion and
integrated in our framework. 

Let $N = \{1,2,\ldots,n\}$ and let $a_1, a_2, \ldots, a_n \in [0,1]$
be sizes of the $n$ items. The independence system induced by a single
knapsack constraint is $\cF=\{ S : \sum_{i \in S} a_i \le 1\}$ and its
natural relaxation has a variable $x_i$ for $1 \le i \le n$ and is
defined as $P_{\cF} = \{ \bx \in [0,1]^n: \sum_i a_i x_i \leq 1 \}$.
We refer to this as the knapsack polytope.

We prove the following lemma.

\begin{lemma}
\label{lem:knapsack}
For any $b \in (0,1/2)$ there is is a monotone $(b,1-2b)$-balanced
\crs scheme for the knapsack polytope.  If, for some $\delta \in (0,\frac12)$,
$a_i \le \delta$ for $1 \le i \le n$, then for
any $b \in (0, \frac{1}{2e})$ there is a monotone $(b, 1 -
(2eb)^{(1-\delta)/\delta})$-balanced \crs scheme. Further, for any
$0 < 2\delta < \eps < \frac12$, if $a_i \le \delta$ for $1 \le i \le n$ then there is a monotone
$(1-\eps, 1-e^{-\Omega(\eps^2/\delta)})$-balanced \crs scheme.
\end{lemma}

\begin{proof}
  The \crs scheme is the same for all the cases and works as follows:
  given $\bx \in b \cdot P_\cF$ we sample $R$ with probabilities
  $x_i$. To obtain $I$ from $R$ we sort the items from $R$ in an order
  of decreasing size and set $I$ to be the largest prefix of this
  sequence that fits in the knapsack. Equivalently, we consider the
  items from $R$ in an order of decreasing size and add the current
  item to $I$ if it maintains feasibility in the knapsack, else we
  discard it. It is easy to see that this scheme is monotone.

  First, we consider the general case where there are no restrictions
  on the item sizes.  Let $N_{\text{big}} = \{ i \in N \mid a_i >
  1/2\}$ be the big items in $N$ and let $N_{\text{small}} = N
  \setminus N_{\text{big}}$ be the small items.
  The probability of at least one big element
  being in $R$ is at most $2b$ since
    $$\Pr[  N_{\text{big}} \cap R \neq \emptyset] \le  \sum_{i \in N_{\text{big}}} x_i
    \le 2 \sum_{i \in N_{\text{big}}} a_i x_i \le 2b.$$ The first
    inequality is via the union bound, the second inequality is
    using the fact that $a_i > 1/2$ for all big items, and the third
    inequality follows from $\bx \in b\cdot P_{\mathcal{F}}$.  Thus $\Pr[
    N_{\text{big}} \cap R = \emptyset] \ge 1-2b$.
    
    Fix some $j\in N$. We need to lower bound $\Pr[j \in I \mid j \in R]$ where
    $I$ is the output of the \crs scheme we described.  First consider
    the case that $j$ is big. Since all big items are considered before any
    small item, $j$ is accepted if it is the unique big item in $R$.
	Since items are included in $R$ independently, we have
    $$ \Pr[ j \in I \ | \ j \in R] \geq
    \Pr[ (N_{\text{big}} \setminus \{j\}) \cap R = \emptyset \ | \ j
    \in R] = \Pr[ (N_{\text{big}} \setminus \{j\}) \cap R = \emptyset]
    \geq \Pr[ N_{\text{big}} \cap R = \emptyset] \geq 1 - 2b.$$ Now we
    consider the case that $j$ is small. Since $a_{j} \le 1/2$, $j$
    will be accepted if $a(R \setminus\{j\}) \le 1/2$.  For any $S
    \subseteq N$, we have $\E[a(S \cap R)] = \sum_{i \in S} a_i \Pr[i
    \in R] \le \sum_{i \in S} a_i x_i$.  In particular $\E[a(R)] \le
    \sum_{i \in N} a_i x_i \le b$. By Markov's inequality, $\Pr[a(R
    \setminus\{j\}) > 1/2] \le \Pr[a(R) > 1/2] \le 2b$.  Therefore
    each small item is accepted with probability at least $1-2b$.  The
    same analysis holds for a simpler \crs scheme based on an ordering
    of elements in $R$ in which all big items are considered before
    any small item.

    Now we consider the case that for $1 \le i \le n$, $a_i \le
    \delta$ for some parameter $\delta \le 1/2$. Fix some $j \in N$.
    It is clear that $j \in I$ if $j \in R$ and $\sum_{i \in
      R\setminus \{j\}} a_i \le 1-\delta$.  Conditioned on $j \in R$,
    the probability of this event is at least $1 - \Pr[\sum_{i \in R}
    a_i > 1-\delta]$. We upper bound $\Pr[\sum_{i \in R} a_i >
    1-\delta]$ via Chernoff bounds. Let $Y_i$ be the indicator random
    variable for $i$ to be chosen in $R$; $\Pr[Y_i=1] = x_i$. Let $Y =
    \sum_{i} a_i Y_i$.  We have $\E[Y] = \sum_{i} a_i x_i \le b <
    1-\delta$ by feasibility of $\bx$. We are interested in $\Pr[Y >
    1-\delta] = \Pr[\sum_{i \in R} a_i > 1-\delta]$.  We can assume
    that $\E[Y] = b$ by adding dummy elements if necessary; this can
    only increase $\Pr[Y > 1-\delta]$.
    
    We use the standard Chernoff bound, $\Pr[Z > (1+\alpha) \mu] \le
    \left(e^\alpha/(1+\alpha)^{1+\alpha}\right)^\mu$ where $Z$ is a
    sum of random variables in $[0,1]$ and $\mu = \E[Z]$.  To apply
    this bound to our setting, we consider $Z = Y/\delta = \sum_i a_i
    Y_i/\delta = \sum_i Z_i$ where $Z_i = a_i Y_i / \delta$ is a
    random variable in $[0,1]$ since $a_i \le \delta$.  Thus, $\Pr[Y >
    1-\delta] = \Pr[Z > (1+\alpha)\E[Z]] \le
    \left(e^\alpha/(1+\alpha)^{1+\alpha}\right)^\mu$ where $\mu =
    \E[Z] = b/\delta$ and $(1+\alpha) = (1-\delta)/b$.  Using $\delta
    \leq \frac12$, we obtain that
  $$ \Pr[Y > 1-\delta] \leq \left( \frac{e}{1+\alpha} \right)^{(1+\alpha) \mu} \le
  \left(\frac{eb}{1-\delta} \right)^{(1-\delta)/\delta} \le
  (2eb)^{(1-\delta)/\delta}.$$

  Finally, let's consider the case where $b = 1-\epsilon$ and $a_i \le
  \delta \le \frac{\epsilon}{2} \le \frac14$ for all $i$. Here we use
  the Chernoff-Hoeffding bound $ \Pr[Z > (1+\alpha) \mu] <
  e^{-\alpha^2 \mu / 3}$ for $\alpha \in (0,1)$ and $Z$ being a sum of
  random variables bounded by $[0,1]$. We estimate the probability
  that conditioned on $j \in R$, all of $R$ fits in the
  knapsack. Since $a_j \leq \delta \leq \frac{\epsilon}{2}$, this
  probability is
	$$ \Pr\left[ \sum_{i \in R} a_i \leq 1 \ | \ j \in R\right] 
	\geq \Pr\left[ \sum_{i \in R \setminus \{j\}} a_i \leq 1 -
        \epsilon/2\right]. $$
   We have $\E[ \sum_{i \in R \setminus \{j\}}
        a_i ] = \sum_{i \in N \setminus \{j\}} a_i x_i \leq 1 -
        \epsilon$.  We can in fact assume that $\mu = \E[ \sum_{i \in
          R \setminus \{j\}} a_i ] = 1-\epsilon$, by adding dummy
        elements that can only increase the probability of overflowing
        $1 - \epsilon/2$.  Applying the Chernoff bound for random
        variables bounded by $\delta$ (after rescaling as above), we
        obtain
	$$ \Pr\left[ \sum_{i \in R \setminus \{j\}} a_i > 1 - \epsilon/2\right] 
	 \leq \Pr\left[ \sum_{i \in R \setminus \{j\}} a_i > (1 + \epsilon/2) \mu \right]
	 \le e^{-\epsilon^2 \mu/(12\delta)} = e^{-\Omega(\epsilon^2 / \delta)}.$$	
\end{proof}

The $(1-\eps, 1-e^{-\Omega(\eps^2/\delta)})$-balanced \crs scheme is
directly applicable only if the item sizes are relatively small
compared to the knapsack capacity. However, standard enumeration
tricks allow us to apply this scheme to general instances as well.
This can be done for any constant number of knapsack constraints. We
formulate this as follows.

\begin{corollary}
\label{coro:knapsacks}
For any constant $k \geq 1$ and $\eps>0$, there is a constant
$n_0$ (that depends only on $\eps$) such that for any submodular
maximization instance involving $k$ knapsack constraints (and possibly
other constraints), there is a set $T$ of at most $n_0$ elements and a
residual instance on the remaining elements such that
\begin{itemize}
\item Any $\alpha$-approximate solution to the residual instance
  together with $T$ is an $\alpha(1-k\eps)$-approximate solution
  to the original instance.
\item In the residual instance, each knapsack constraint admits a
  $(1-\eps,1-\eps)$-balanced \crs scheme.
\end{itemize}
\end{corollary}

\begin{proof}
  Given $\eps>0$, let $\delta = O(\eps^2 / \log (1/\eps))$ and $n_0 =
  1/(\delta \eps)$. Select $T$ greedily from the optimal solution, by
  picking elements as long as their marginal contribution is at least
  $\delta \eps \opt$; note that $|T| \leq n_0$.  We define the
  residual instance so that $S$ is feasible in the residual instance
  iff $S \cup T$ is feasible in the original instance. The objective
  function in the new instance is $g$ defined by setting $g(S) = f(S
  \cup T)$ for each set $S \subseteq N\setminus T$; note that $g$ is a
  non-negative submodular function if $f$ is.  In addition, in the
  residual instance we remove all elements whose size for some
  knapsack constraint is more than $\delta \cdot r$ where $r$ is the
  residual capacity. The number of such elements in a knapsack can be
  at most $1/\delta$ and hence they can contribute at most $\eps
  \opt$; we forgo this value for each knapsack. We obtain a residual
  instance where all sizes are at most $\delta$ with the capacities
  normalized to $1$. By Lemma~\ref{lem:knapsack}, each knapsack admits
  a $(1-\eps,1-e^{-\Omega(\eps^2/\delta)})=(1-\eps,1-\eps)$-balanced
  CRS.
\end{proof}

An advantage of this black box approach is that knapsack constraints
can be combined arbitrarily with other types of constraints. They do
not affect the approximation ratio significantly. However, the
enumeration stage affects the running time by an $O(n^{n_0})$ factor.

\subsection{Sparse packing systems}

We now consider packing constraints of the type $A \bx \leq \bb$,
where $\bx \in \{0,1\}^N$ is the indicator vector of a solution.  We
can assume without loss of generality that the right-hand side is $\bb
= \b1$.  We say that the system is $k$-sparse, if each column of $A$
has at most $k$ nonzero entries (i.e., each element participates in at
most $k$ linear constraints). The approximation algorithms in
\cite{BansalKNS10} can be seen to give a contention resolution scheme for
$k$-sparse packing systems.
\smallskip
\smallskip

\noindent
{\bf \crs scheme for $k$-sparse packing systems:}
\begin{itemize}
\item We say that element $j$ participates in constraint $i$, if $a_{ij} > 0$. We call an element $j$ {\em big} for this constraint, if $a_{ij} > 1/2$. Otherwise we call element $j$ {\em small} for this constraint.
\item Sample $R$ with probabilities $x_i$.
\item For each constraint $i$: if there is exactly one big element in $R$
that participates in $i$, mark all the small elements in $R$ for this
constraint for deletion; otherwise check whether $\sum_{j\in R}a_{ij}>1$
and if so, mark all elements participating in $i$ for deletion.
\item Define $I$ to be $R$ minus the elements marked for deletion.
\end{itemize}

Based on the analysis in \cite{BansalKNS10}, we obtain the following.

\begin{lemma}
\label{lem:k-sparse}
For any $\bbb \in (0,\frac{1}{2k})$, the above is a monotone $(\bbb, 1-2k \bbb)$-balanced \crs scheme for $k$-sparse packing systems.
\end{lemma}

\begin{proof}
  Let $\bx=b\cdot \by$ with $\by\in [0,1]^N$, $A\by \leq \b1$.
  Consider a fixed element $j^*$. It appears in $R$ with probability
  $x_{j^*}$. We analyze the probability that it is removed due
  to some constraint where it participates. First, note that whether
  big or small, element $j^*$ cannot be removed due to a constraint
  $i$ if the remaining elements have size less than $1/2$, i.e. if
  $\sum_{j \in R \setminus \{j^*\}} a_{ij} < 1/2$. This is because in
  this case, there is no other big element participating in $i$, and
  element $j^*$ is either big in which case it survives, or it is
  small and then $\sum_{j \in R} a_{ij} \leq 1$, i.e. the constraint
  is satisfied.

  Thus it remains to analyze the event $\sum_{j \in R \setminus
    \{j^*\}} a_{ij} \geq 1/2$. Note that this is independent of item
  $j^*$ appearing in $R$. By the feasibility of $\frac{1}{b}\bx$,
  $\E[\sum_{j \in R \setminus \{j^*\}} a_{ij}] = \sum_{j
    \neq j^*} x_{j} a_{ij}$ $\leq \bbb$. By Markov's
  inequality, $\Pr[\sum_{j \in R \setminus \{j^*\}} a_{ij} \geq 1/2]
  \leq 2 \bbb$.  So an element is removed with probability at most
  $2\bbb$ for each constraint where it participates. By the union
  bound, it is removed by probability at most $2 k \bbb$.

\end{proof}

Recall the notion of width for a packing system: $W = \lfloor
 \frac{1}{\max_{i,j} a_{ij}} \rfloor$, where $a_{ij}$ are the entries of the
packing matrix (recall that we normalize the right-hand side to be
$\bb=\b1$).  Assuming that $W \geq 2$, one can use a simpler \crs scheme
and improve the parameters.

\mypara{\crs scheme for $k$-sparse packing systems of width $W$:}
\begin{itemize}
\item Sample $R$ with probabilities $x_i$.
\item For each constraint $i$ for which $\sum_{j \in R} a_{ij} > 1$, mark all elements participating in $i$ for deletion.
\item Define $I$ to be $R$ minus the elements marked for deletion.
\end{itemize}

\begin{lemma}
\label{lem:k-sparse-W}
For any $\bbb \in (0,\frac{1}{2e})$, the above is a monotone
$(\bbb,1-k (2 e \bbb)^{W-1})$-balanced \crs scheme for any $k$-sparse
system of packing constraints of width $W \geq 2$.
\end{lemma}

\begin{proof}
  Again, let $\bx=b\by$ with $\by\in [0,1]^N$, $A\by \leq \b1$.  Let
  us consider an element $j'$ and a constraint $i$ 
  that $j'$ participates in. If we condition on $j'$ being
  present in $R$, we have $\mu_i = \E[\sum_{j \in R \setminus \{j'\}}
  a_{ij} \mid j' \in R] = \sum_{j \neq j'} a_{ij} x_{ij} \leq
  \bbb$. By the width property, we have $a_{ij'} \leq 1/W \leq
  1/2$. We use the Chernoff bound for a sum $X$ of independent
  $[0,1]$ random variables with $\mu=\E[X]$: $\Pr[X
  > (1+\delta) \mu] \leq (e^\delta / (1+\delta)^{1+\delta})^\mu \leq
  (e/(1+\delta))^{(1+\delta)\mu}$, with $1+\delta = (1 - a_{ij'}) /
  \mu_i \geq 1/(2\bbb)$. Since our random variables are bounded by
  $[0,\max a_{ij}]$, we obtain by scaling
$$ \Pr\left[ \sum_{j \in R} a_{ij} > 1 \mid j' \in R \right]
 = \Pr\left[ \sum_{j \in R \setminus \{j'\}} a_{ij} > 1 - a_{ij'} \right] 
 $$
$$ \leq \left( \frac{e}{1+\delta} \right)^{(1+\delta) \mu_i / \max a_{ij}}
 \leq \left( {2e \bbb} \right)^{(1 - a_{ij'}) / \max a_{ij}}
 \leq \left( {2e \bbb} \right)^{W-1}.  $$
Therefore, each element is removed with probability at most $(2e \bbb)^{W-1}$ for each constraint where it participates.
\end{proof}

We remark that a $k$-sparse packing system can be viewed as the
intersection of multiple knapsack constraints on the elements where
each element participates in at most $k$ constraints. One can use the
composition lemma (Lemma~\ref{lem:compose}) and the \crs-schemes for a
single knapsack constraint given by Lemma~\ref{lem:knapsack} to obtain
\crs-schemes for $k$-sparse packing systems. The schemes that we
described and analyzed above can be seen as direct implementations
of the composition approach.

\subsection{UFP in paths and trees}
We consider the following routing/packing problem.  Let $T=(V,E)$ be a
capacitated tree with $u_e$ denoting the capacity of edge $e \in
E$. We are given $k$ distinct node pairs $s_1t_1,\ldots,s_kt_k$ with
pair $i$ having a non-negative demand $d_i$. We assume that the
instance satisfies the no-bottleneck condition, that is, $d_{\max} =
\max_i d_i \le u_{\min} = \min_e u_e$. We say that an instance is a
unit-demand instance if $d_i = 1$ for each $i \in N$ and $u_e$ is a
non-negative integer for each $e \in E$.

Let $N = \{1,\ldots,k\}$, and
for $i\in N$, we denote by $Q_i\subseteq E$ the edges on the unique
path between $s_i$ and $t_i$ in $T$.
We say that $S\subseteq N$ is \emph{routable} if, when routing $d_i$ units of
flow from $s_i$ to $t_i$ over $Q_i$ for each $i\in S$, then the total
flow on any edge $e$ is at most $u_e$. More formally, $S$ is routable if
\begin{equation*}
\sum_{i\in S: e\in Q_i} d_i \leq u_e \quad \forall e\in E.
\end{equation*}
We are interested in finding a routable set $S\subseteq N$ that
maximizes some weight function on $N$. The case of linear weights was
considered in~\cite{ChekuriMS03}. Here, a weight $w_i\geq 0$ is given
for $i\in N$, and the goal is to find a routable set $S\subseteq N$
that maximizes $\sum_{i\in S}w_i$. A constant factor approximation has
been presented for this problem~\cite{ChekuriMS03}, and moreover
it is known that the problem is APX-hard even for unit-demands and
unit-weights~\cite{GargVY97}.

We are interested in more general submodular weights.
 Let $\cI = \{S \subseteq N \mid S \text{ is
  routable}\}$.  The problem we consider is $\max_{S \in \cI} f(S)$, where
$f$ is a given non-negative submodular function. We present a \crs scheme
for this problem that implies a constant factor approximation through
our framework. We start by presenting a \crs scheme for unit demands,
which we then extend to general demands.

\medskip

A natural (packing) LP
relaxation for $P_{\cI}$ has a variable $x_i \in [0,1]$ for each pair
$i$ and a constraint $\sum_{i: e \in Q_i} d_i x_i \le u_e$ for each
edge $e$; recall that $Q_i$ is the set of edges on the
unique $s_i$-$t_i$ path in $T$.

\mypara{\crs scheme for unit-demands:}
\begin{itemize}
\item Root $T$ arbitrarily. Let depth of pair $s_it_i$ be the depth of 
the least common ancestor of $s_i$ and $t_i$ in $T$.
\item Let $R \subseteq N$ be random set obtained by including each $i$
  independently with probability $\bbb x_i$. 
\item Let $I = \emptyset$. 
\item Consider pairs in $R$ in increasing order of depth.
  \begin{itemize}
  \item Add $i$ to $I$ if $I \cup \{i\}$ is routable, otherwise reject $i$.
  \end{itemize}
\item Output $I$.
\end{itemize}

\smallskip
The techniques in \cite{CCGK07,ChekuriMS03} give the following lemma.

\begin{lemma}
  \label{lem:ufp-crs}
  For any $\bbb \in  (0,\frac{1}{3e})$ the above is a $(\bbb, 1- \frac{2e \bbb}{1-e
    \bbb})$-balanced \crs scheme.
\end{lemma}
\begin{proof}
  Let $\bx \in b \cdot P_{\mathcal{I}}$.
  Consider a fixed pair $i^*$ and let $v$ be the least common ancestor
  of $s_{i^*}$ and $t_{i^*}$ in the rooted tree $T$; note that $v$
  could be one of $s_{i^*}$ or $t_{i^*}$. Let $P$ be the unique path
  in $T$ from $v$ to $s_{i^*}$ and $P'$ be the path from $v$ to $t_{i^*}$.
  Without loss of generality assume that $v \neq s_{i^*}$ and hence $P$ is
  non-empty.  We wish to upper bound $\Pr[i^* \not \in I \mid i^* \in R]$,
  that is, the probability that $i^*$ is rejected conditioned on it
  being included in the random set $R$. The reason that $i^*$ gets
  rejected is that at least one edge $e \in P \cup P'$ is already full
  from the pairs that have been accepted into $I$ prior to considering
  $i^*$. We upper bound the probability of this event happening for some
  edge in $P$ and use a symmetric argument for $P'$.

  Let $e_1,e_2,\ldots,e_h$ be the edges in $P$ from $v$ to $s_{i^*}$.
  Let $\mathcal{E}_j$ be the event that $i^*$ gets rejected at $e_j$, that is,
  the capacity of $e_j$ is full when $i^*$ is considered for addition
  to $I$. Note that these events are correlated. We claim the
  following: if $j > h$ and $u_{e_j} \ge u_{e_h}$ then $\mathcal{E}_j$ happens
  only if $\mathcal{E}_h$ happens. The reason for this is the order in which the
  pairs in $R$ are considered for insertion. When $i^*$ is considered,
  the only pairs inserted in $I$ prior to it are those whose depth is
  no larger, and hence the total capacity used on an edge decreases as
  we traverse the path $P$ from $v$ to $s_i$.  Thus, to analyze the
  probability of rejection it suffices to consider a subsequence of
  $e_1,e_2,\ldots,e_h$ starting with $e_1$ such that the capacity of
  the next edge in the sequence is strictly smaller than the previously
  added one. For notational simplicity we will therefore assume that
  $u_{e_1} > u_{e_2} > \ldots > u_{e_h} \ge 1$.  

  Let $S_j = \{ i \neq i^* \mid e \in Q_{i}\}$ be the set of pairs
  other than $i^*$ that contain $e$ in their path $Q_i$. Let
  $\mathcal{E}'_j$ be the event that $|R \cap S_j| \ge u_{e_j}$. It is
  easy to see that $\Pr[\mathcal{E}_j] \le \Pr[\mathcal{E}'_j]$. Since
  $\frac{1}{b}\bx$ is a feasible solution to the LP relaxation we have
  $\sum_{i \in S_j} x_i < \bbb u_{e_j}$.
  Letting $X_i$ be the event that $i \in R$, and $X =
  \sum_{i \in S_j} X_i$, we have $\Pr[\mathcal{E}'_j] = \Pr[X \ge
  u_{e_j}]$. Since $X$ is the sum of independent $[0,1]$ random
  variables $X_i$, and has expectation $\bbb u_{e_j}$,
  we obtain by standard Chernoff bounds:
  $$\Pr[\mathcal{E}'_j]=\Pr[X\geq u_{e_j}]\leq
  (e^\delta / (1+\delta)^{1+\delta})^\mu
  \leq(e/(1+\delta))^{(1+\delta)\mu},$$
  where $\mu = \bbb u_{e_j}$ and
  $\delta = 1/\bbb - 1$. Hence, $\Pr[\mathcal{E}'_j] \le (e
  \bbb)^{u_{e_j}}$. Taking the union bound over all edges in the path,
  the probability of rejection of $i^*$ on some edge in $P$ is at most
  $\sum_{j=1}^h (e \bbb)^{u_{e_j}} \le \sum_{\ell=1}^\infty (e
  \bbb)^\ell = \frac{e \bbb}{1-e \bbb}$, where the inequality is due to
  the fact that the edge capacities are strictly decreasing and lower
  bounded by $1$, and the equality is due to the fact that $e\bbb < 1$
  (recall that $\bbb \in (0,\frac{1}{3e})$). By a union bound over $P$
  and $P'$ we have that the probability of $i^*$ being rejected
  conditioned on it being in $R$ is at most $\frac{2e \bbb}{1-e
    \bbb}$.
\end{proof}

\smallskip
\noindent {\bf \crs scheme for general demands:} A \crs scheme for
general demands can be obtained as follows. The linear program
$P_{\cI}$ is a packing LP of the form $A \bx \le \bb, \bx \in [0,1]$
where $A$ is column-restricted (all the non-zero values in a column
have the same value). For such column-restricted packing integer
programs (CPIPs), when demands satisfy the no-bottleneck assumption,
one can use grouping and scaling techniques first suggested by
Kolliopoulos and Stein \cite{KolliopoulosS04} (see also
\cite{ChekuriMS03}) to show that the integrality gap for a CPIP
with matrix $A$ is
at most a fixed constant factor worse than that of the underlying
$0$-$1$ matrix $A'$ (obtained from $A$ by placing a $1$ in each
non-zero entry). Note that in the context of the UFP problem, the
matrix $A$ corresponds to the problem with arbitrary demands while the
matrix $A'$ corresponds to the one with unit-demands.  One can use the
same grouping and scaling techniques to show that a monotone
$(b,1-b')$-balanced \crs scheme for $A'$ can be used to obtain a monotone
$(b/6, (1- b')/2)$-balanced \crs scheme for $A$. We give a proof in 
Section~\ref{sec:CPIP}, see Theorem~\ref{thm:cpip-crscheme}. Using this general
conversion theorem and Lemma~\ref{lem:ufp-crs}, one can obtain a
$(\bbb, \bbb')$-balanced \crs scheme for UFP in trees for some
sufficiently small but absolute constants $\bbb$ and $\bbb'$. This
suffices to obtain a constant factor approximation for maximizing a
non-negative submodular function of routable requests in a capacitated
tree. However, the $(b/6, (1-b')/2)$-balanced \crs scheme does not allow
composition with other constraints via Lemma~\ref{lem:compose} since
$(1-b')/2$ does not tend to zero even if $b'$ does. However,
Theorem~\ref{thm:cpip-crscheme} gives a more refined statement that
is helpful in applications in light of Remark~\ref{remark:subadditive}.

 Without the no-bottleneck assumption, the linear program has an
 $\Omega(n)$ integrality gap even for UFP on paths~\cite{CCGK07}. One
 can still apply the grouping and scaling techniques without the
 no-bottleneck assumption under a mild restriction; we refer the
 reader to~\cite{ChekuriEK09}.

\subsection{Column-restricted packing constraints}
\label{sec:CPIP}

Here we consider \crs schemes for CPIPs. We follow the notation from
\cite{ChekuriMS03}.  Let $A$ be an arbitrary $m \times n$ $\{0,1\}$-matrix,
and $d$ be an $n$-element non-negative vector with $d_j$
denoting the $j$th entry in $d$. Let $A[d]$ denote the matrix obtained
by multiplying every entry of column $j$ in $A$ by $d_j$. A CPIP is a
problem of the form $\max w \bx$, subject to $A[d] \bx \le \bb, \bx
\in \{0,1\}^n$. Note that all non-zero entries in $A[d]$ for any given
column have the same value and hence the name column-restricted.  Here
we are interested in submodular objective functions and the goal is
obtain a \crs scheme for the polytope $P_{\cI}$ induced by the relaxation
$A[d] \bx \le \bb, \bx \in [0,1]^n$. Instead of focusing on the polytope
for a given $d$ and $b$, we consider the class of polytopes induced by
all $d,b$.

\begin{theorem}
  \label{thm:cpip-crscheme}
  Suppose there is a monotone $(\beta, 1-\beta')$ \crs scheme for the
  polytope $A \bx \le \bb, x \in [0,1]^n$ for every $\bb \in
  \mathbb{Z}_+$ where $A$ is $\{0,1\}$-matrix. Then there is a
  monotone $(\beta/6, (1 - \beta')/2)$-balanced \crs scheme for the polytope
  $A[d] \bx \le \bb, x \in [0,1]^n$ for all $d,b$ such that $\dmax =
  \max_j d_j \le \bmin = \min_j b_j$. Moreover there is a monotone
  $(\beta/6, 1-\beta')$-balanced \crs scheme if all $d_j \le \bmin/3$ or if all
  $d_j \ge \bmin/3$.
\end{theorem}

We sketch the proof of the above theorem which follows the grouping
and scaling ideas previously used in~\cite{KolliopoulosS04,ChekuriMS03}.
We have chosen some specific constants in the theorem for simplicity.  One
can obtain some generalizations and variations of the above theorem
via the same ideas.

Let $N = \{1,\ldots, n\}$ be a ground set corresponding to the
columns. Given $d$, for integer $h \ge 0$ we let $N_h = \{ j \in N
\mid d_j \in (\dmax/3^{h+1}, \dmax/3^h]\}$. We think of the columns in
$N_0$ as {\em large} and the rest as {\em small}. The overall idea is
to focus either on the large demands or the small demands. Moreover,
we will see that small demands can be treated independently within
each group $N_h$. Let $\bz$ be a feasible solution to the system $A[d]
\bx \le \bb, \bx \in [0,1]^n$. For integer $h \ge 0$ we let $\bz^h$
denote the vector obtained from $\bz$ as follows: $\bz^h_j = \bz_j/6$
if $j \in N_h$ and $\bz^h_j = 0$ otherwise. The vector $\bz^h$
restricts the solution $\bz$ to elements in $N_h$ and scales it down
by a small constant factor. We also define a corresponding vector
$\bb^h$ where $b^h_i = \ceil{A_i\bz^h}$ for each row $i$.  We have the
following lemma which is a restatement of corresponding statements
from~\cite{KolliopoulosS04,ChekuriMS03}.

\begin{lemma}
  \label{lem:group-scale}
  For $h \ge 0$, let $\by^h \in \{0,1\}^n$ be a feasible integral
  solution to $A\bx \le \bb^h, \bx \in [0,1]^n$ such that $\by_j = 0$
  if $\bz^h_j = 0$. Then $A[d] \by^0 \le \bb$ and $\sum_{h \ge 1} A[d]
  \by^h \le \bb$.
\end{lemma}
\begin{proof}
  Fix some $h$ and consider the $i$-th row of $A[d] \by^h$ which
  is equal to $\sum_{j \in N_h} d_j A_{ij} y^h_i$. We upper bound
  this quantity as follows:
  \begin{eqnarray*}
    \sum_{j \in N_h} d_j A_{ij} y^h_i & \le & \frac{\dmax}{3^h}\sum_{j \in N_h} A_{ij} y^h_i  \quad \quad \mbox{(from definition of $N_h$)}\\
    & \le & \frac{\dmax}{3^h}b^h_i  \quad \quad \mbox{(feasibility of $\by^h$)}\\
    & \le & \frac{\dmax}{3^h} \left(1 + \sum_{j \in N_h} A_{ij} z^h_i\right)
       \quad \quad \mbox{(definition of $\bb^h$ and using $\ceil{a} \le 1+ a$)}\\
    & \le & \frac{\dmax}{3^h}\left(1+ \sum_{j \in N_h} A_{ij} z_j/6\right)
       \quad \quad \mbox{(from definition of $\bz^h$)} \\
    & \le & \frac{\dmax}{3^h} + \frac{1}{2}\sum_{j \in N_h} A_{ij} d_j z_j \quad \quad \mbox{($d_j > \dmax/3^{h+1}$ for $j \in N_h$)}.    
  \end{eqnarray*}
 For $h=0$  we need a slight variant of the above where we replace
  $\bb^h_i$ by $\max\{1, 2 \sum_{j \in N_0} A_{ij} z^h_i\}$ since
  $\ceil{a} \le \max\{1,2a\}$. Then we obtain that 
$$\sum_{j \in N_0} d_j A_{ij} y^0_i \le
  \max\{\dmax, \sum_{j \in N_o} A_{ij} d_j z_j\} \le  b_i,$$
since $\dmax \le \bmin$ and $\bz$ is feasible. Thus $A[d]\by^0 \le \bb$.

For the second part of the claim, consider a row $i$.
\begin{eqnarray*}
  \sum_{h \ge 1} \sum_{j \in N_h} d_j A_{ij} y^h_i & \le &\sum_{h \ge 1}
   \left(\frac{\dmax}{3^h} + \frac{1}{2}\sum_{j \in N_h} A_{ij} d_j z_j\right) \\
  & \le & \sum_{h \ge 1} \frac{\dmax}{3^h} +  \sum_{h \ge 1}  \frac{1}{2}\sum_{j \in N_h} A_{ij} d_j z_j \\
  & \le & \frac{\dmax}{2} + \frac{b_i}{2} \\
  & \le & b_i.
  \end{eqnarray*}
The penultimate inequality is from the feasibility of $\bz$, and the last
inequality is from the assumption that $\dmax \le \bmin$.
\end{proof}

With the above claim in place we can describe the \crs scheme 
claimed in the theorem. Let $\bz$ be a feasible solution
and let $\bz^h$ for $h \ge 0$ be constructed from $\bz$ as described
above.

\mypara{\crs scheme:}
\begin{itemize}
\item For each $h \ge 0$ {\em independently} run the
  $(\beta,1-\beta')$-balanced \crs scheme for the polytope $A\bx \le \bb^h,
  x\in [0,1]^n$ with fractional solution $\bz^h$ to obtain integral
  vectors $\by^h$, $h \ge 0$.
\item With probability $1/2$ output $\by^0$, otherwise output
  $\sum_{h \ge 1}\by^h$.
\end{itemize}

We claim that the above scheme is a monotone
$(\beta/6,(1-\beta')/2)$-balanced \crs scheme.
Note that we use the unit-demand scheme in a black-box
fashion.  First, we observe via Lemma~\ref{lem:group-scale} that the
output of the scheme is a feasible integral solution.  An alternative
description of the scheme is as follows.
We are given a point $\bx = \frac{\beta}{6} \bz$ with
$\bz\in [0,1]^n$, $A\bz \leq \bb$.
Obtain a set $R\subseteq N$
by independently sampling each $j \in N$ with probability
$x_j=\beta/6 \cdot z_j$. Let $R_h = R \cap N_h$. For each $h$ obtain $I_h \subseteq
R_h$ as the output of the scheme for $A\by \le \bb^h, \by\in [0,1]^n$
given the random set $R_h$. With probability $1/2$ output $I = I_0$ otherwise
output $I = \cup_{h \ge 1} I_h$. For $j \in N_h$ we have that $\Pr[j \in I_h
\mid j \in R_h] \ge 1-\beta'$. Further, $\Pr[j \in I \mid j \in I_h] = 1/2$
by the choice of the algorithm in the second step. Therefore
$\Pr[ j \in I \mid j \in R] \ge (1-\beta')/2$. It is easy to verify
the scheme is monotone.

Further, if we only have large demands or only small demands then the
second step is not necessary and hence we obtain a
$(\beta/6,(1-\beta'))$-balanced \crs scheme.

\mypara{Acknowledgments:} We thank Mohit Singh for
helpful discussions on contention resolution schemes for matroids, and
Shipra Agrawal for discussions concerning the correlation gap.
We thank two anonymous reviewers for their comments which helped
us improve the presentation of the details in the paper.

\appendix

\section{Approximation for general polytopes}
\label{sec:general-polytopes}

In this section, we formulate an approximation result for the problem $\max \{F(\bx): \bx \in P\}$
when $P$ is a general solvable polytope (not necessarily down-monotone).
This result is included only for the sake of compleness; we do not have any concrete applications for it.
Our result generalizes (while losing a factor of 4) the result for matroid base polytopes from \cite{Vondrak09},
which states that a $\frac12 (1-\frac{1}{\nu}-o(1))$-approximation can be achieved,
provided that the fractional base packing number is at least $\nu$ where
$\nu \in [1,2]$.
As observed in \cite{Vondrak09}, the fractional base packing number being at least $\nu$
is equivalent to the condition $P \cap [0,\frac{1}{\nu}]^N \neq \emptyset$.
This is the condition we use for general polytopes.
We state the algorithm only in its continuous form; we omit the discretization details.

\begin{algorithm}
\label{alg:gen-local-search}
Let $t \in [0,1]$ be a parameter such that $P \cap [0,t]^N \neq \emptyset$.
Initialize $\bx \in P \cap [0,t]^N$ arbitrarily. As long as there is $\by \in P \cap [0,\frac12 (1+t)]^N$ such that
$(\by-\bx) \cdot \nabla F(\bx) > 0$ (which can be found by linear
programming), move $\bx$ continuously in the direction $\by-\bx$. If
there is no such $\by \in P \cap [0,\frac12 (1+t)]^N$, return $\bx$.
\end{algorithm}

Note that even though we require $P \cap [0,t]^N \neq \emptyset$, the local search works inside a larger polytope
$P \cap [0,\frac12 (1+t)]^N$. This is necessary for the analysis.

\begin{theorem}
For any solvable polytope such that $P \cap [0,t]^N \neq \emptyset$,
Algorithm~\ref{alg:gen-local-search} approximates the problem $\max \{F(\bx): \bx \in P\}$
within a factor of $\frac18 (1-t)$.
\end{theorem}

\begin{proof}
The algorithm maintains the invariant $\bx \in P \cap [0,\frac12 (1+t)]^N$.
Suppose that the algorithm returns a point $\bx$. Then we know that for every $\by \in P \cap [0,\frac12 (1+t)]^N$,
$(\by - \bx) \cdot \nabla F(\bx) \leq 0$. We use a particular point $\by$ defined as follows:
Let $\bx^*$ be the optimum, i.e.~$F(\bx^*) = \max \{F(\bx): \bx \in P\}$,
and let $\bx_0$ be any point in $P \cap [0,t]^N$, for example the starting point.
Then we define $\by = \frac12 (\bx_0 + \bx^*)$. By convexity, we have $\by \in P$,
and since $\bx^* \in [0,1]^N$, we also have $\by \in [0,\frac12 (1+t)]^N$.
Therefore, by the local-search condition, we have $(\by - \bx) \cdot \nabla F(\bx) \leq 0$.
By Lemma~\ref{lem:G-bound},
$$ 2 F(\bx) \geq F(\bx \vee \by) + F(\bx \wedge \by) \geq F(\bx \vee \by). $$
Let $\bx' = \bx \vee \by$. The point $\bx'$ has the following properties: $\bx' = \bx \vee \frac12 (\bx_0 + \bx^*) \geq \frac12 \bx^*$,
and also $\bx' \in [0,\frac12 (1+t)]^N$. Consider the ray $\frac12 \bx^* + \xi (\bx' - \frac12 \bx^*)$ parameterized by $\xi \geq 0$.
Observe that this ray has a positive direction in all coordinates, and it is possible to go beyond $\xi=1$ and still stay inside $[0,1]^N$:
in particular, for $\xi = \frac{2}{1+t}$ we get a point $\frac12 \bx^* + \frac{2}{1+t} (\bx' - \frac12 \bx^*) \leq \frac{2}{1+t} \bx' \in [0,1]^N$.
Using this fact, we can express $\bx'$ as a convex combination:
$$ \bx' = \frac{1+t}{2} \cdot \left( \frac12 \bx^* + \frac{2}{1+t} (\bx' - \frac12 \bx^*) \right) + \frac{1-t}{2} \cdot \frac12 \bx^* $$
(the reader can verify that this is an identity). By the concavity of $F$ in positive directions, we get
$$ F(\bx') \geq \frac{1+t}{2} F\left( \frac12 \bx^* + \frac{2}{1+t} (\bx' - \frac12 \bx^*) \right) + \frac{1-t}{2} F\left(\frac12 \bx^*\right). $$
As we argued, $\frac12 \bx^* + \frac{2}{1+t} (\bx' - \frac12 \bx^*) \in [0,1]^N$, so we can just lower-bound the respective value by $0$,
and we obtain
$$ F(\bx') \geq \frac{1-t}{2} F\left( \frac12 \bx^*\right) \geq \frac{1-t}{4} F(\bx^*).$$
Finally, our solution satisfies
$$ F(\bx) \geq \frac12 F(\bx \vee \by) = \frac12 F(\bx') \geq \frac{1-t}{8} F(\bx^*) = \frac{1-t}{8} \opt.$$
\end{proof}

\section{Missing proofs of Section~\ref{sec:CRS}}%
\label{sec:appendixCR}

\subsection*{Proof of Theorem~\ref{thm:submod-balanced-strict}}

\begin{proof}%
As observed in the proof of Theorem~\eqref{thm:submod-balanced}, it
suffices to show~\eqref{eq:pruneToProve} (assuming an arbitrary ordering
of the elements $N=\{1,\dots, n\}$).
Let us take the expectation in two steps, first over $I$ conditioned on
$R$, and then over $R$:
\begin{eqnarray*}
\E[f(I \cap [i]) - f(I \cap [i-1])]
 & \geq & \E_R[\E_I[ \b1_{i \in I} f_{R \cap [i-1]}(i) \mid R]]  \\
 & = & \E_R[\Pr[i \in I \mid R] \, f_{R \cap [i-1]}(i)].
\end{eqnarray*}
Note that $\Pr[i \in I \mid R]$ can be nonzero only if $i \in R$,
therefore we can restrict our attention to this event:
$$ \E[f(I \cap [i]) - f(I \cap [i-1])]
\geq \Pr[i \in R] \cdot \E[\Pr[i \in I \mid R] f_{R \cap [i-1]}(i)
\mid i \in R]. $$ On the product space associated with the
distribution of $R$ conditioned on $i \in R$, both $\Pr[i \in I \mid
R]$ and $f_{R \cap [i-1]}(i)$ are non-increasing functions, due to $I$
being monotone with respect to $R$, and $f$ being submodular.
Therefore, the FKG inequality (see \cite{AlonS-book}) implies that
\begin{eqnarray*}
 \E_R[\Pr[i \in I \mid R] f_{R \cap [i-1]}(i) \mid i \in R] & \geq &  
 \E_R[\Pr[i \in I \mid R] \mid i \in R] \cdot \E_R[f_{R \cap [i-1]}(i) \mid i \in R]\\
 & = &  \Pr[i \in I \mid i \in R] \cdot \E[f_{R \cap [i-1]}(i)].
\end{eqnarray*}
since the marginal value $f_{R \cap [i-1]}(i)$ does not depend on $i\in R$.
By the $(b,c)$-balanced property, $\Pr[i \in I \mid i \in R] \geq c$;
in addition, $f$ is either monotone or we assume that 
$\Pr[i \in I \mid i \in R] = c$. In both cases,
$\Pr[i \in I \mid i \in R] \cdot \E[f_{R \cap [i-1]}(i)] \geq c \cdot \E[f_{R \cap [i-1]}(i)]$.
We summarize:
\begin{eqnarray*}
\E[f(I \cap [i]) - f(I \cap [i-1])] & \geq & \Pr[i \in R] \cdot c \,\E[f_{R \cap [i-1]}(i)] \\
 & = & c \, \E[f(R \cap [i]) - f(R \cap [i-1])].
\end{eqnarray*}
Therefore,
$$ \E[f(I)] = f(\emptyset) + \sum_{i=1}^{n} \E[f(I \cap [i]) - f(I \cap [i-1])]
 \geq f(\emptyset) + c \, \sum_{i=1}^{n} \E[f(R \cap [i]) - f(R \cap [i-1])]
  \geq c \, \E[f(R)].$$
\end{proof}

\subsection*{Proof of Lemma~\ref{lem:correl-gap-stronger}}

To prove Lemma~\ref{lem:correl-gap-stronger} we use a property of
submodular functions presented in~\cite{V07}, which is stated below
as Lemma~\ref{lem:janSubIneq}.

\begin{lemma}[\cite{V07}]\label{lem:janSubIneq}
Let $f: 2^N \rightarrow \mathbb{R}_+$ be a monotone submodular
function, and $A_1,\dots, A_m \subseteq N$. For each $j\in [m]$
independently, sample a random subset $A_j(q_j)$ which
contains each element of $A_j$ with probability $q_j$.
Let $J$ be a random subset of $[m]$ containing each element
$j\in [m]$ independently with probability $q_j$. 
Then
\begin{equation*}
\E[f(A_1(q_1)\cup \dots \cup A_m(q_m))] \geq
\E\left[f\left(\bigcup_{j\in J} A_j\right)\right].
\end{equation*}
\end{lemma}

Lemma~\ref{lem:strongSubIneq} below is a generalization of Lemma~4.2
in~\cite{V07}. We follow the same proof 
technique as used in~\cite{V07}. The lemma
contains two statements. The first 
is a simpler statement that may be of independent
interest. The second, which can be 
seen to be a slightly stronger version of the
first statement, turns out to imply
Lemma~\eqref{lem:correl-gap-stronger},
as we will show in the following.

\begin{lemma}\label{lem:strongSubIneq}
Let $f: 2^N \rightarrow \mathbb{R}_+$ be a monotone submodular
function, and $A_1,\dots, A_m \subseteq N$. For each $j\in [m]$
independently sample a random subset $A_j(q_j)$ which
contains each element of $A_j$ with probability $q_j$.
Let $q=\sum_{j=1}^m q_j$. Then
\begin{equation*}
\E[f(A_1(q_1)\cup \dots \cup A_m(q_m))] \geq
\frac{1}{q} \left(1 - \left(1 - \frac{q}{m} \right)^m \right)
\sum_{j=1}^m q_j f(A_j).
\end{equation*}
Furthermore for $m\geq 2$ and any $s,t\in [m]$ with $s\neq t$, 
\begin{align*}
\E[f(A_1(q_1)\cup &\dots \cup A_m(q_m))]\\
&\geq
\frac{1}{q-q_s q_t}
 \left(1 - \left(1 - \frac{q-q_s q_t}{m-1} \right)^{m-1} \right)
\left(-q_s q_t \min\{f(A_s),f(A_t)\} + \sum_{j=1}^m q_j f(A_j)\right).
\end{align*}
\end{lemma}
\begin{proof}
Observe that the first statement is a consequence
of the second one: it suffices to add an arbitrary
additional set $A_{m+1}$ with probability $q_{m+1}=0$
to the family of sets and invoke the second
part of the lemma with $s=1$ and $t=m+1$.
Hence, we only prove the second part of the lemma.

By Lemma~\ref{lem:janSubIneq} it suffices to estimate
$\E\left[f(\cup_{j\in J} A_j)\right]$, where $J$ is a
random subset of $[m]$ containing element $j\in [m]$
independently of the others with probability $q_j$. Assume
$f(A_1) \geq \dots \geq f(A_m)$ and without loss of generality
we assume $t>s$. We define
for $k\in [m]$,
\begin{equation*}
\mathcal{J}_k = \{I\subseteq [m] \mid \min(I) = k\}.
\end{equation*}
By monotonicity of $f$, we have 
$f(\cup_{j\in J} A_j)\geq f(A_k)$ if $J\in \mathcal{J}_k$.
Hence,
\begin{align*}
\E[f(\cup_{j\in J} A_j)] \leq \sum_{j=1}^m
\Pr[J\in \mathcal{J}_j] f(A_j)
= \sum_{j=1}^m f(A_j) q_j \prod_{\ell = 1}^{j-1} (1-q_\ell).
\end{align*}
Thus, it suffices to prove
\begin{equation}
\sum_{j=1}^m f(A_j) q_j \prod_{\ell = 1}^{j-1} (1-q_\ell)
\geq \frac{1}{q-q_s q_t}
  \left(1 - \left(1 - \frac{q-q_s q_t}{m-1} \right)^{m-1} \right)
\left(- q_s q_t f(A_t)
  + \sum_{j=1}^m q_j f(A_j)\right).
\end{equation}
Since the above inequality is linear in the parameters
$f(A_j)$, it suffices to prove it for the special case
$f(A_1)= \dots  = f(A_r)=1$ and
$f(A_{r+1}) = \dots = f(A_m) = 0$ (A general decreasing
sequence of $f(A_j)$ can be obtained as a positive linear
combination of such special cases.) Hence, it remains
to prove
\begin{equation}\label{eq:toProveSI}
\sum_{j=1}^r q_j \prod_{\ell = 1}^{j-1} (1-q_\ell)
\geq \frac{1}{q-q_s q_t}
  \left(1 - \left(1 - \frac{q-q_s q_t}{m-1} \right)^{m-1} \right)
\left(- \b1_{r\geq t} \cdot p_s p_t + \sum_{j=1}^r q_j\right),
\end{equation}
where $\b1_{r\geq t}$ is equal to $1$ if $r\geq t$ and $0$
otherwise.
To prove~\eqref{eq:toProveSI} we distinguish two cases
depending on whether $r<t$ or $r\geq t$.
\medskip

\textbf{Case $r<t$}:
Expanding the left-hand side of~\eqref{eq:toProveSI},
we obtain
\begin{equation*}
\sum_{j=1}^r q_j \prod_{\ell =1}^{j-1} (1-q_\ell)
= 1 - \prod_{j=1}^r (1-q_j) \geq
1 - \left( 1- \frac{1}{r}\sum_{j=1}^r q_j \right)^r,
\end{equation*}
using the arithmetic-geometric mean inequality.
Finally, using concavity of
$\phi_r(x) = 1-(1-\frac{x}{r})^r$ and $\phi_r(0)=0$,
we get
\begin{align*}
1 - \left( 1- \frac{1}{r}\sum_{j=1}^r q_j \right)^r &=
\phi_r\left(\sum_{j=1}^r q_j\right) \geq
 \phi_r(q - q_s q_t) \frac{\sum_{j=1}^r q_j}{q - q_s q_t} 
 \geq \phi_{m-1}(q - q_s q_t) \frac{\sum_{j=1}^r q_j}{q - q_s q_t}\\
&= \frac{1}{q - q_s q_t}
  \left(1-\left(1-\frac{q-q_s q_t}{m-1}\right)^{m-1}\right)
 \sum_{j=1}^r q_j,
\end{align*}
where the last inequality follows from the fact $\phi_r(x)$
is decreasing in $r$ and by using $r\leq m-1$.
Notice that we used the
fact $\sum_{j=1}^r q_j \leq q-q_s q_t$ for the
first inequality in the
reasoning above, which holds since
$r < t$ and therefore
$\sum_{j=1}^r q_j \leq q - q_t \leq q - q_s q_t$.
\medskip

\textbf{Case $r \geq t$}:
As in the previous case we start by expanding
the left-hand side of~\ref{eq:toProveSI}. This
time we bundle the two terms $(1-q_s)$ and $(1-q_t)$
when applying the arithmetic-geometric mean inequality.
\begin{align*}
\sum_{j=1}^r q_j \prod_{\ell =1}^{j-1} (1-q_\ell)
& = 1 - \prod_{j=1}^r (1-q_j)
= 1 - (1-q_s)(1-q_t)
\prod_{\substack{j\in [r]\\ j\notin \{s,t\}}} (1-q_j)\\
&\geq 1 - \left( 1- \frac{1}{r-1}
\left(-q_s q_t + \sum_{j=1}^{r} q_j \right) \right)^{r-1}.
\end{align*}
The remaining part of the proof is similar to the previous case.
\begin{align*}
1 - \left( 1- \frac{1}{r-1}\left(- q_s q_t +
  \sum_{j=1}^r q_j\right) \right)^{r-1}
&= \phi_{r-1} \left(-q_s q_t + \sum_{j=1}^r q_j\right)
 \geq \phi_{r-1} (q - q_s q_t)
   \frac{-q_s q_t + \sum_{j=1}^r q_j}{q - q_s q_t} \\
&\geq \phi_{m-1} (q - q_s q_t)
   \frac{-q_s q_t + \sum_{j=1}^r q_j}{q - q_s q_t} \\
&= \frac{1}{q - q_s q_t}
  \left(1-\left(1-\frac{q-q_s q_t}{m-1}\right)^{m-1}\right)
 \left(- p_s p_t + \sum_{j=1}^r q_j\right),
\end{align*}
again using concavity of $\phi_{r-1}$ and the fact that
$\phi_{r-1}(q-q_s q_t)$ is decreasing in $r$.
This finishes the proof of~\eqref{eq:toProveSI} and thus
completes the proof of the lemma.
\end{proof}

Leveraging Lemma~\ref{lem:strongSubIneq} we are now ready
to prove Lemma~\ref{lem:correl-gap-stronger}.
We recall that for a nonnegative submodular function
$f:2^N\rightarrow \mathbb{R}_+$ and $\bp\in [0,1]^N$,
its concave closure $f^+$ is defined by
\begin{equation*}
f^+(\bp)= \max\left\{\sum_{S\subseteq N}\alpha_S f(S) \;\Big\vert\;
  \alpha_S \geq 0 \; \forall S\subseteq N,
  \sum_{S\subseteq N} \alpha_S = 1,
\sum_{S\subseteq N, i\in S} \alpha_S = p_i \;\forall i\in N\right\}.
\end{equation*}

\begin{proof}[Proof of Lemma~\ref{lem:correl-gap-stronger}]
Consider a basic solution $(\alpha_j, A_j)_{j\in [m]}$
to the linear program that defines
$f^+(\bp)$, i.e.,
$f^+(\bp) = \sum_{j=1}^m \alpha_j f(A_j)$,
with $A_j\subseteq N, \alpha_j \geq 0$ for $j\in [m]$, 
$\sum_{j=1}^m \alpha_j=1$ and
$\sum_{j\in [m], i\in A_j}\alpha_j = p_i$ for $i\in N$.
Notice that since we chose a basic solution and
the LP defining $f^+(\bp)$ only has $n+1$ constraints
apart from the nonnegativity constraints, we have $m\leq n+1$.
Let $R(b \bp)$ be a random subset of $N$ containing each
element $i\in N$ independently with probability $b p_i$.
We distiguish two cases depending on whether $m\leq n$
or $m= n+1$.

\medskip
\textbf{Case $m\leq n$}:
Consider the random set
\begin{equation*}
A = \bigcup_{j\in [m]} A_j\left(b \cdot \alpha_j \right),
\end{equation*}
where $A_j(b\alpha_j)$ is a random subset of $N$ containing
each element $i\in N$ with probability $b\alpha_j$,
independently of the others.
Notice that the distribution of $A$ is dominated by the distribution
of $R(b\bp)$ since $A$ contains each element $i\in N$ independently
with probability
\begin{equation*}
\Pr[i\in A] = 1- \prod_{\substack{j\in [m]\\ i\in A_j}}%
\left(1-b \alpha_j \right)
\leq 1 - \left( 1 - \sum_{\substack{j\in [m]\\ i\in A_j}}
b \alpha_j \right) = b p_i = \Pr[i\in R(b\bp)].
\end{equation*}
Hence $F(b\bp) \geq \E[f(A)]$, and we can use
the first statement of Lemma~\ref{lem:strongSubIneq} to obtain
\begin{equation*}
F(b\bp) \geq \E[f(A)] \geq
\frac{1}{\sum_{j=1}^m b\alpha_j}
\left(1 - \left( 1-\frac{\sum_{j=1}^m b \alpha_j}{m}\right)^m \right)
\sum_{j=1}^m b \alpha_j f(A_j)
\geq \left(1- \left(1-\frac{b}{m}\right)^m \right) f^+(\bp),
\end{equation*}
using $\sum_{j=1}^m \alpha_j = 1$ and
the fact that 
$\frac{1-(1-\frac{x}{m})^m}{x}$ is decreasing in $x$.
The proof of this case is completed by observing that
$(1-(1-\frac{b}{m})^m)$ is decreasing in $m$ and $m\leq n$.

\medskip
\textbf{Case $m = n+1$}:
Since $A_j\subseteq N$ for $j\in [n+1]$ and $|N|=n$, there
must be at least one set $A_{t}$ that is
covered by the remaining sets, i.e.,
$A_{t}\subseteq \cup_{j\in [n+1], j\neq t} A_j$.
Furthermore let $s\in [n+1]\setminus \{t\}$ be the index
minimizing $b \alpha_s$.
We define probabilities $q_j$ for $j\in [n+1]$ as follows
\begin{equation*}
q_j = 
\begin{cases}
b \alpha_j & \text{if } j\neq t, \\
\frac{b \alpha_t}{1 - b\alpha_s} & \text{if } j = t.
\end{cases}
\end{equation*}
We follow a similar approach as for the previous
case by considering the random set
\begin{equation*}
A = \bigcup_{j\in [m]} A_j (q_j).
\end{equation*}
Again, we first observe that $A$ is dominated by
the distribution of $R(b \bp)$. For any $i\in N\setminus A_t$
the analysis of the previous case still holds and shows
$\Pr[i\in A] \leq \Pr[i\in R(b\bp)]$.
Consider now an element $i\in A_t$.
Let $A_t, A_{j_1}, \dots, A_{j_r}$ be all sets in the
family $(A_j)_{j\in [n+1]}$ that contain $i$.
By our choice of $A_t$, there is at least one other set
containing $i$, i.e., $r\geq 1$. 
Using 
\begin{equation*}
(1-q_t)(1-q_{j_1})
  = \left(1- \frac{b \alpha_t}{1- b \alpha_s}\right)(1-b\alpha_{j_1})
  \overset{\alpha_s \leq \alpha_{j_1}}{\leq} 1 - b\alpha_t - b\alpha_{j_1},
\end{equation*}
we obtain
\begin{align*}
\Pr[i\in A]
 &= 1 - (1-q_t)(1-q_{j_1}) \prod_{\ell = 2}^r (1-q_{j_\ell})
   = 1 - (1- b \alpha_t - b\alpha_{j_1})
     \prod_{\ell = 2}^r (1-b \alpha_{j_\ell})\\
 &\leq 1 - \left(1-\sum_{\substack{j\in [n+1]\\ i\in A_j}} b\alpha_j\right)
   = b p_i = \Pr[i\in R(b\bp)].
\end{align*}
Therefore, we again have $F(b\bp) \geq \E[f(A)]$.
Notice that $q = \sum_{j=1}^{n+1} q_j$ satisfies
\begin{equation*}
q = b\alpha_t - \frac{b\alpha_t}{1-b\alpha_s}
  + \sum_{j=1}^{n+1} b \alpha_j
  = q_s q_t + \sum_{j=1}^{n+1} b \alpha_j
  = q_s q_t + b.
\end{equation*}
We apply
the second statement of Lemma~\ref{lem:strongSubIneq}
to the family $(A_j(q_j))_{j\in [n+1]}$ and use the
above fact to obtain
\begin{align*}
F(b\bp) \geq \E[f(A)] &\geq
\frac{1}{q- q_s q_t}
\left(1 - \left( 1-\frac{q - q_s q_t}{n}\right)^{n} \right)
\left(-q_s q_t \min\{f(A_s), f(A_t)\} +
  \sum_{j=1}^{n+1} q_j f(A_j)\right)\\
& = \frac{1}{b}\left(1- \left(1-\frac{b}{n}\right)^{n} \right) 
\left(-q_s q_t \min\{f(A_s), f(A_t)\} +
  \left(q_t - b\alpha_t\right)f(A_t)
    + b f^+(\bp)\right)\\
& \geq \frac{1}{b}\left(1- \left(1-\frac{b}{n}\right)^{n} \right) 
\left(-q_s q_t f(A_t) +
  \left(q_t - b\alpha_t\right)f(A_t)
    + b f^+(\bp)\right).
\end{align*}
The claim follows by observing that $q_s q_t = q_t - b\alpha_t$.
\end{proof}

\section{Details in constructing \crs schemes via the ellipsoid algorithm}
\label{sec:appendix-ellipsoid}
Here we give the technical details that are involved in sampling and
approximately solving (DP1) and (LP1) from
Section~\ref{subsec:distrCRschemes}.  
First a short reminder of the primal and dual problem.

\begin{equation*}
\begin{array}{rrrcll}
\multirow{4}{*}{(LP1)} & \max & c \hfill &&&\\
& s.t. & \sum_{\phi \in \Phi^*} q_{i,\phi} \lambda_\phi  &\geq
    & \mathrm{x}_i c &\forall i\in N\\
& &\sum_{\phi \in \Phi^*} \lambda_{\phi} & = &1 &\\
& & \lambda_\phi &\geq &0 &\forall \phi \in \Phi^* \\[0.5em]
\multirow{4}{*}{(DP1)}& \min & \mu \hfill &&&\\
& \text{s.t.} & \sum_{i\in N} q_{i,\phi} y_i & \leq & \mu
   &\forall \phi \in \Phi^*\\
& &\sum_{i\in N} \mathrm{x}_i y_i & = &1 &\\
& & y_i &\geq &0 &\forall i\in N
\end{array}
\end{equation*}

We start by observing that we can obtain strong estimates
of $q_{i,\phi}$ for any $\phi\in \Phi^*$. We assume that
$\phi$ is given as an oracle and can therefore be evaluated
in constant time.

\begin{proposition}\label{prop:qHat}
Let $\phi\in \Phi^*$ and $i\in N$. An estimate
$\hat{q}_{i,\phi}$ of $q_{i,\phi}$ whose error is bounded
by $\pm \epsilon x_i$ with high probability can be obtained
in time polynomial in $n$ and $\frac{1}{\epsilon}$.
\end{proposition}
\begin{proof}
We call a set $S\subseteq N\setminus \{i\}$ \emph{good}, if
$i\in \phi(S\cup\{i\})$. We have
\begin{align*}
q_{i,\phi} &= \Pr[i\in \phi(R)]
  = \Pr[i\in R \text{ and } R\setminus\{i\} \text{ is good}]\\
 &= \Pr[i\in R]\cdot \Pr[R\setminus\{i\}\text{ is good}]
  = x_i \cdot \Pr[R\setminus\{i\} \text{ is good}].
\end{align*}
Notice that we can estimate $\Pr[R\setminus\{i\} \text{ is good}]$ up
to an error of $\pm \epsilon$ with high probability by a standard
Monte Carlo approach, where we draw samples of $R\setminus\{i\}$. This
can be done in time polynomial in $n$ and $\frac{1}{\epsilon}$ and
leads to the claimed estimate $\hat{q}_{i,\phi}$ by the above formula.
\end{proof}

We now discuss how these estimates can be used to obtain 
a near-optimal solution to (DP1) by employing the
ellipsoid method with a weak separation oracle.
After that we show how a near-optimal solution to (LP1) can
be obtained.
Notice that Proposition~\ref{prop:qHat} can easily be used
to obtain estimates $\widetilde{q}_{i,\phi}$ of $q_{i,\phi}$
that satisfy with high probability
$\widetilde{q}_{i,\phi}\in [q_{i,\phi}-\epsilon x_i, q_{i,\phi}]$:
it suffices to
consider an estimate $\hat{q}_{i,\phi}$ of $q_{i,\phi}$ that
satisfies with high probability
$\hat{q}_{i,\phi}\in [q_{i,\phi}-\frac{\epsilon}{2} x_i,%
 q_{i,\phi} + \frac{\epsilon}{2} x_i]$
and to define
$\widetilde{q}_{i,\phi} = \hat{q}_{i,\phi}-\frac{\epsilon}{2}$.
In the following, we assume that all used estimates $\widetilde{q}_{i,\phi}$
satisfy $\widetilde{q}_{i,\phi}\in [q_{i,\phi}-\epsilon x_i, q_{i,\phi}]$.
We can obtain this with
high probability through Proposition~\ref{prop:qHat} since the ellipsoid
that we will apply in the following only uses a polynomial number
of such estimates.
Notice that these estimates are ``pessimistic'' estimates for $q_{i,\phi}$,
i.e., replacing $q_{i,\phi}$ in (LP1) by these estimates leads to a
lower optimal value of the LP.

Furthermore, to simplify the exposition, we will assume that for any
given weight vector $\by\in \mathbb{R}_+^N$ we can find a \crs scheme
$\phi\in \Phi^*$ maximizing $\sum_{i\in N}q_{i,\phi} y_i$. The
following discussion works also if we can only find a \crs scheme
$\phi$ that approximately maximizes this expression.

To apply the ellipsoid algorithm to (DP1) we design a weak separation
oracle (see Chapter~4 in~\cite{KorteVygen}). As a reminder, the weak separation
oracle has to provide the following guarantees. Given is a nonnegative
vector $\by=(y_i)_{i\in N}$ satisfying $\sum_{i\in N}\mathrm{x}_i y_i =1$, and
a value $\mu$. The weak separation oracle has to provide either a
feasible dual solution $(\by', \mu')$ with $\mu' \leq \mu + \epsilon$,
or a hyperplane separating $(\by, \mu)$ from all feasible dual
solutions.
Given $\by$ and $\mu$, let $\phi\in \Phi^*$ be the CR scheme
that maximizes $\sum_{i\in N}q_{i,\phi}y_i$.
If $\sum_{i\in N}\widetilde{q}_{i,\phi} y_i \leq \mu$,
our weak separation oracle returns the dual solution
$(\by, \mu + \epsilon)$. This solution has objective value
$\mu + \epsilon$ as desired and is indeed feasible since
for any $\phi'\in \Phi^*$,
\begin{align*}
\sum_{i\in N}q_{i,\phi'}y_i \leq \sum_{i\in N}q_{i,\phi}y_i
\leq \sum_{i\in N} (\widetilde{q}_{i,\phi}+\epsilon \mathrm{x}_i) y_i
= \epsilon + \sum_{i\in N} \widetilde{q}_{i,\phi} y_i 
 \leq \mu + \epsilon.
\end{align*}
If $\sum_{i\in N}\widetilde{q}_{i,\phi} y_i > \mu$,
our weak separation oracle returns the separating
hyperplane given by the constraint $({\bf a},-1)\cdot ({\bf z}, \nu) \le 0$
where ${\bf a}$ is a $n$-dimensional vector with coefficients 
$a_i = \widetilde{q}_{i,\phi}$ for $1 \le i \le n$
(note that that (DP1) has $n+1$ variables corresponding to $y_1,\ldots,y_n$
and $\mu$). 
First, this hyperplane indeed cuts off the solution $(\by,
\mu)$. Furthermore, if $(\by', \mu')$ is a feasible dual solution then
it satisfies the constraint since $\widetilde{q}_{i,\phi}\leq q_{i,\phi}$:
\begin{align*}
\sum_{i\in N}\widetilde{q}_{i,\phi} y_i'
\leq \sum_{i\in N}q_{i,\phi} y_i' \leq \mu',
\end{align*}
where the second inequality in the above follows from feasibility of
$(\by',\mu')$.
Hence, we obtained a weak separation oracle for (DP1),
and the ellipsoid method
can therefore determine a feasible solution $(\by, \mu)$ to (DP1)
of value $\leq \mu^* + 2\epsilon$, where $\mu^*$ is the value of an
optimal dual solution. Note that since $(\by, \mu)$ is feasible,
we have $\mu^* \le \mu \le \mu^*+2\epsilon$ (see~\cite{KorteVygen}).

Let (DP1') be the linear program obtained
from (DP1) by only considering constraints corresponding
to \crs schemes $\phi$ that were used in the ellipsoid
algorithm while constructing the nearly optimal solution
$(\by,\mu)$ of (DP1), which satisfies $\mu \leq \mu^*+2\epsilon$.
Furthermore, we replace all terms $q_{i,\phi}$ by
their estimates $\widetilde{q}_{i,\phi}$ in (DP1').
Hence, the feasible region of (DP1') consists of all separating
hyperplanes that were generated during the ellipsoid algorithm.
Notice that (DP1') is a relaxation of (DP1) since our
estimates satisfy $\widetilde{q}_{i,\phi} \leq q_{i,\phi}$.
Hence, the optimal value $\mu'$ of (DP1') satisfies
$\mu' \leq \mu^*$.
The ellipsoid algorithm actually certifies the approximation
quality of the generated solution $(\mu, \by)$ by comparing
against the best solution satisfying the generated constraints,
i.e.,
\begin{equation*}
\mu \leq \mu' + 2\epsilon.
\end{equation*}
Let (LP1') be the dual of (DP1'), and let
$(\blam', c')$ be an optimal solution to (LP1'),
which can be efficiently determined since (LP1')
has polynomial size. We return $(\blam', c')$ as the
solution to (LP1).
First, notice that $(\blam', c')$ is feasible for
(LP1) since $\widetilde{q}_{i,\phi}\leq q_{i,\phi}$.
Furthermore,
\begin{align*}
c' = \mu' \geq \mu - 2\epsilon \geq \mu^* -2\epsilon = c^* - 2\epsilon,
\end{align*}
where the two equalities follow by strong duality.

\end{document}